\title{Turing machines deciders, part I}
\author{
The bbchallenge Collaboration\thanks{\url{https://bbchallenge.org}} \and
Justin Blanchard \and
Konrad Deka \and
Nathan Fenner \and
Tony Guilfoyle \and
Iijil \and
Maja Kądziołka \and
Pavel Kropitz \and
Shawn Ligocki \and
Pascal Michel \and
Mateusz Na\'{s}ciszewski \and
Tristan Stérin
}

\documentclass[a4paper,british]{article}

\usepackage{babel}
\usepackage[utf8]{inputenc}
\usepackage[T1]{fontenc}
\usepackage[margin=1in]{geometry}
\usepackage[hidelinks]{hyperref}
\usepackage{caption}
\usepackage{floatpag}
\usepackage{subcaption}
\usepackage{tikz}

\usepackage{algorithm}
\usepackage[noend]{algpseudocode}

\usepackage{graphicx}
\usepackage{mathtools}

\usepackage{amsmath,amsfonts,amssymb,amsthm}

\theoremstyle{definition} 
\newtheorem{theorem}{Theorem}[section]
\newtheorem{definition}{Definition}[section]
\newtheorem{lemma}{Lemma}[section]

\newtheorem{corollary}{Corollary}[section]
\numberwithin{equation}{section}

\theoremstyle{definition} 

\newtheorem{example}{Example}[section]
\newtheorem{remark}{Remark}[section]

\usepackage{xassoccnt}
\DeclareCoupledCountersGroup{theorems}
\DeclareCoupledCounters[name=theorems]{theorem,definition,lemma,proposition,corollary,remark,observation,example}

\usepackage{microtype,xspace,wrapfig,multicol} 
\usepackage[textsize=tiny,color=lightgray]{todonotes} 
\usepackage[normalem]{ulem} 
\usepackage{stmaryrd}

\newcommand{\tabi}{\hspace{\algorithmicindent}}

\newcommand{\N}{\mathbb{N}}
\newcommand{\Z}{\mathbb{Z}}

\newcommand{\lhead}[1]{\stackrel{#1}\triangleleft}
\newcommand{\rhead}[1]{\stackrel{#1}\triangleright}

\usepackage{xcolor}

\definecolor{colorA}{RGB}{255,0,0}
\definecolor{colorB}{RGB}{255,128,0}
\definecolor{colorC}{RGB}{0,0,255}
\definecolor{colorD}{RGB}{0,255,0}
\definecolor{colorE}{RGB}{255,0,255}

\usepackage{titling}
\begin{document}
\date{}

\maketitle
\vspace{-3.7em}
\begin{abstract}
  The Busy Beaver Challenge (or bbchallenge) aims at collaboratively solving the following conjecture: ``$S(5) = 47{,}176{,}870$'' [Radó, 1962]\nocite{Rado_1962}, [Marxen and Buntrock, 1990]\nocite{Marxen_1990}, [Aaronson, 2020]\nocite{BusyBeaverFrontier}. This conjecture says that if a 5-state Turing machine runs for more than 47,176,870 steps without halting, then it will never halt -- starting from the all-0 tape. Proving this conjecture amounts to deciding whether 181,385,789  Turing machines with 5 states halt or not -- starting from the all-0 tape \cite{bbchallenge2025}. To do so, we write \textit{deciders}: programs that take as input a Turing machine and output either \texttt{HALT}, \texttt{NONHALT}, or \texttt{UNKNOWN}. Each decider is specialised in recognising a particular type of non-halting behavior.\footnote{Because the halting problem is undecidable, there is no \textit{universal} decider, i.e.~that never returns \texttt{UNKNOWN}.}

  After two years of work, the Busy Beaver Challenge achieved its goal in July 2024 by delivering a proof of ``$S(5) = 47{,}176{,}870$'' formalised in Coq\footnote{Coq has been renamed Rocq. Also, the formalised proof is available at \url{https://github.com/ccz181078/Coq-BB5}.} \cite{bbchallenge2025}. In this document, we present deciders that were developed before the Coq proof and which \textbf{were not used} in the proof;\footnote{Apart from the verifier part of Finite Automata Reduction, Section~\ref{sec:finite-automata-reduction}, which is used in \cite{bbchallenge2025}.} nonetheless, they are relevant techniques for analysing Turing machines. Part II of this work is the decider section of our paper showing ``$S(5) = 47{,}176{,}870$'' \cite{bbchallenge2025}, presenting the deciders that were used in the Coq proof.

\end{abstract}

\setcounter{tocdepth}{2}
\tableofcontents

\newpage
\section{Conventions}\label{sec:conventions}

\begin{table}[h!]
  \centering
  \begin{tabular}{lll}
      & 0     & 1   \\
    A & 1RB   & 1LC \\
    B & 1RC   & 1RB \\
    C & 1RD   & 0LE \\
    D & 1LA   & 1LD \\
    E & - - - & 0LA
  \end{tabular}
  \caption{Transition table of the current 5-state busy beaver champion: it halts after 47,176,870 steps.\\\url{https://bbchallenge.org/1RB1LC_1RC1RB_1RD0LE_1LA1LD_---0LA&status=halt}}
\end{table}\label{table:bb5}

The set $\mathbb{N}$ denotes $\{0,1,2\dots\}$.

\paragraph*{Turing machines.}The Turing machines that are studied in the context of bbchallenge use a binary alphabet and a single bi-infinite tape. Machine transitions are either undefined (in which case the machine halts) or given by (a) a symbol to write (b) a direction to move (right or left) and (c) a state to go to. Table~\ref{table:bb5} gives the transition table of the current 5-state busy beaver champion. The machine halts after 47,176,870 steps (starting from all-0 tape) when it reads a 0 in state E, which is undefined.

A \textit{configuration} of a Turing machine is defined by the 3-tuple: (i) state (ii) position of the head (iii) content of the memory tape. In the context of bbchallenge, \textit{the initial configuration} of a machine is always (i) state is A, i.e. the first state to appear in the machine's description (ii) head's position is 0 (iii) the initial tape is all-0 -- i.e. each memory cell is containing 0. We write $c_1 \vdash_\mathcal{M} c_2$ if a configuration $c_2$ is obtained from $c_1$ in one computation step of machine $\mathcal{M}$. We omit $\mathcal{M}$ if it is clear from context. We let $c_1 \vdash^s c_2$ denote a sequence of $s$ computation steps, and let $c_1 \vdash^* c_2$ denote zero or more computation steps. 
We write $c_1 \vdash \bot$ if the machine halts after executing one computation step from configuration $c_1$. In the context of bbchallenge, halting happens when an undefined machine transition is met i.e. no instruction is given for when the machine is in the state, tape position and tape corresponding to configuration $c_1$.

When discussing concrete configurations, we write
$0^\infty\; s_1\; \cdots\; s_{k-1}\; [s_k]_q\; s_{k+1}\; \cdots\; s_n\; 0^\infty$
to mean the configuration where the machine is in state $q$, with the head
positioned on the symbol $s_k$, and the tape both starts and end by an infinite sequence of 0s, represented $0^\infty$. Thus, the initial configuration of the machine
can be written as $0^\infty\; [0]_A\; 0^\infty$.

\paragraph*{Directional Turing machines.} We will sometimes prefer to think of the tape head as being between
symbols. Thus, we write $l \lhead{q} r$, with $l,r\in\{0,1\}^*$ to mean that the head is at the rightmost symbol
of $l$, and $l \rhead{q} r$ to mean that the head is at the leftmost symbol of $r$.
For example, $0^\infty\; [1]_A\; 0^\infty$ can be written as
$0^\infty\; 1 \lhead A 0^\infty$ or $0^\infty \rhead A 1\; 0^\infty$.

\paragraph*{Space-time diagram.} We use space-time diagrams to give a visual representation of the behavior of a given machine. The space-time diagram of machine $\mathcal{M}$ is an image where the $i^\text{th}$ row of the image gives:
\begin{enumerate}
  \item The content of the tape after $i$ steps (black is 0 and white is 1).
  \item The position of the head is colored to give state information using the following colours for 5-state machines: \textcolor{colorA}{A},  \textcolor{colorB}{B},  \textcolor{colorC}{C},  \textcolor{colorD}{D},  \textcolor{colorE}{E}.
\end{enumerate}


\newpage
\section{Cyclers}\label{sec:cyclers}

\begin{figure}[h!]
  \centering
  \includegraphics[width=0.4\textwidth]{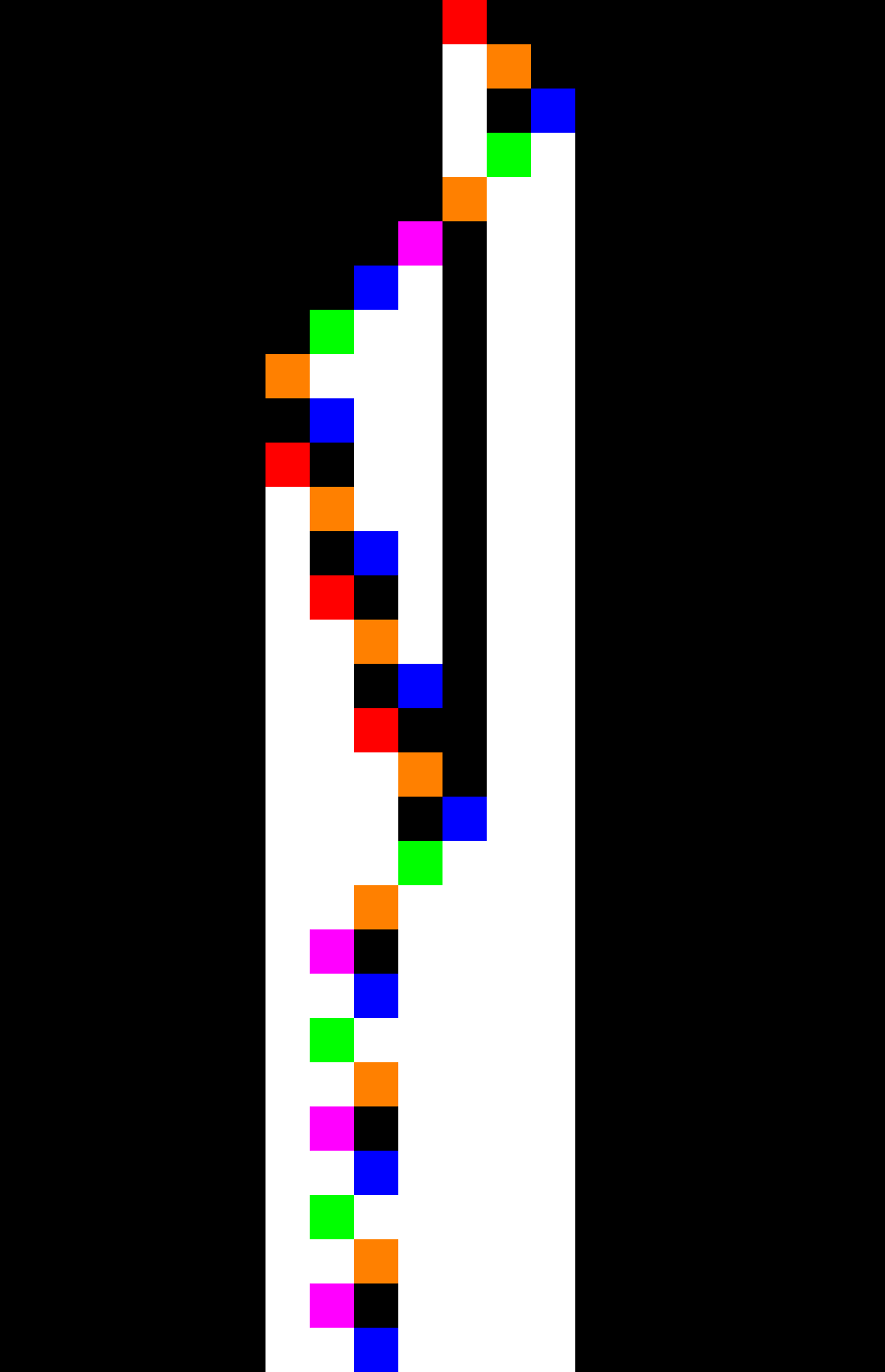}
  \hspace{2ex}
  \includegraphics[width=0.4\textwidth]{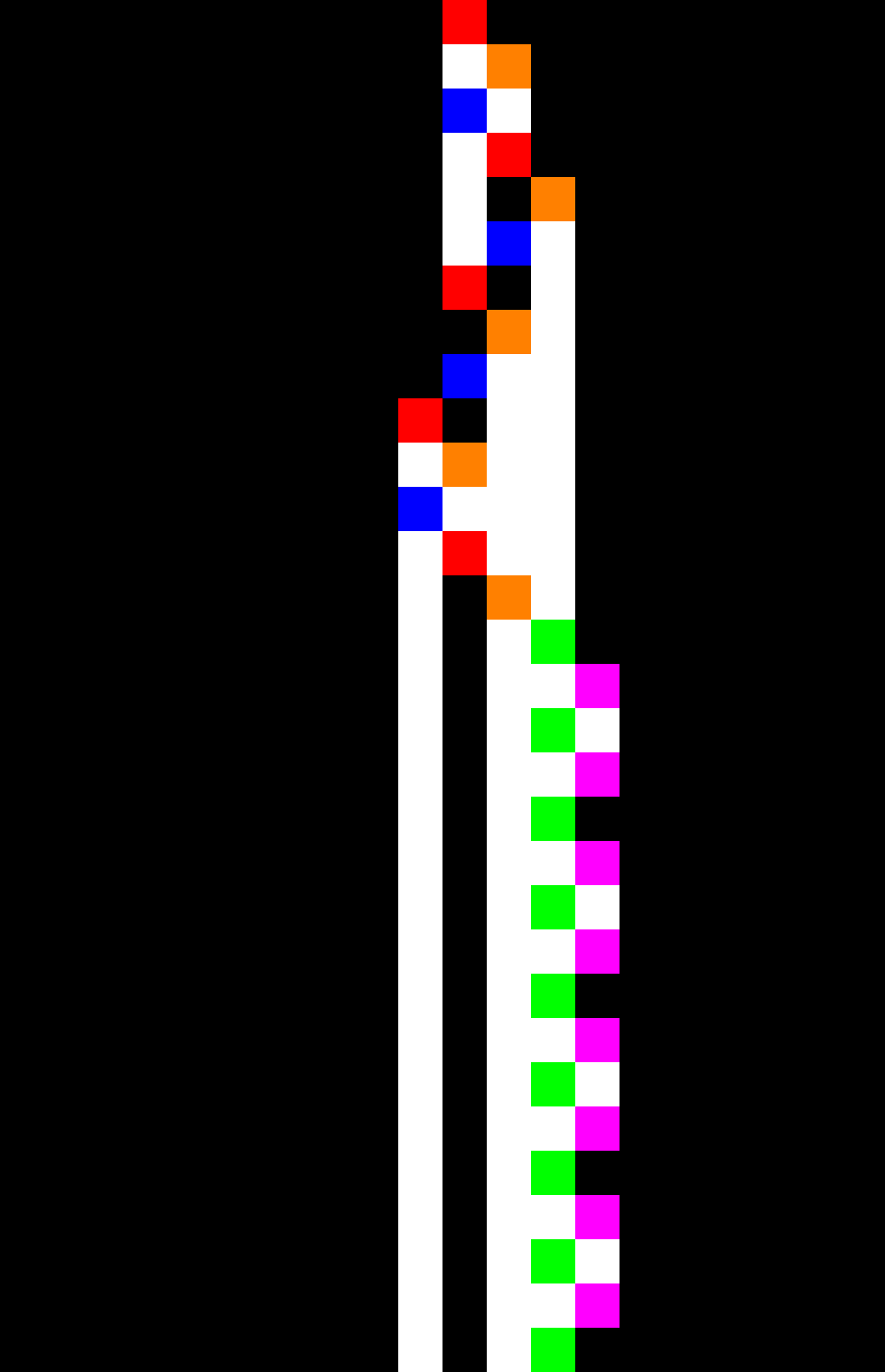}
  \caption{Space-time diagrams of the 30 first steps of bbchallenge's machines \#279,081 (left) and \#4,239,083 (right) which are both ``Cyclers'': they eventually repeat the same configuration for ever. \\
    Access the machines at \url{https://bbchallenge.org/279081} and
    \url{https://bbchallenge.org/4239083}.}\label{fig:cyclers}
\end{figure}

The goal of this decider is to recognise Turing machines that cycle through the same configurations forever. Such machines never halt. The method is simple: remember every configuration seen by a machine and return \texttt{true} if one is visited twice. A time limit (maximum number of steps) is also given for running the test in practice: the algorithm recognises any machine whose cycle fits within this limit\footnote{In practice, for machines with 5 states the decider was run with 1000 steps time limit.}.

\begin{example}\normalfont
  Figure~\ref{fig:cyclers} gives the space-time diagrams of the 30 first iterations of two ``Cyclers'' machines: bbchallenge's machines \#279,081 (left) and \#4,239,083 (right). Refer to \url{https://bbchallenge.org/279081} and
  \url{https://bbchallenge.org/4239083} for their transition tables. From these space-time diagrams we see that the machines eventually repeat the same configuration.
\end{example}

\subsection{Pseudocode}

We assume that we are given a Turing Machine type \textbf{TM} that encodes the transition table of a machine as well as a procedure \textbf{TuringMachineStep}(machine,configuration) which computes the next configuration of a Turing machine from the given configuration or \textbf{nil} if the machine halts at that step. The pseudocode is given in Algorithm~\ref{alg:cyclers}.

\begin{algorithm}
  \caption{{\sc decider-cyclers}}\label{alg:cyclers}

  \begin{algorithmic}[1]

    \State \textbf{struct} Configuration \{
    \State \tabi\textbf{State} state
    \State \tabi\textbf{int} headPosition
    \State \tabi\textbf{int $\boldsymbol{\to}$ int} tape
    \State \}
    \State
    \Procedure{\textbf{bool} {\sc decider-cyclers}}{\textbf{TM} machine,\textbf{int} timeLimit}
    \State \textbf{Configuration} currConfiguration = \{.state = \textcolor{colorA}{A}, .headPosition = 0, .tape = \{0:0\}\}
    \State \textbf{Set$\boldsymbol{<}$Configuration$\boldsymbol{>}$} configurationsSeen = \{\}
    \State \textbf{int} currTime = 0

    \While{currTime $\leq$ timeLimit}
    \If{currConfiguration \textbf{in} configurationsSeen} \label{j:alg:cyclers}
    \State \textbf{return} true
    \EndIf
    \State configurationsSeen.\textbf{insert}(currConfiguration) \label{i:alg:cyclers}

    \State currConfiguration = \textbf{TuringMachineStep}(machine,currConfiguration)
    \State currTime += 1

    \If{currConfiguration == \textbf{nil}}
    \State \textbf{return} false //machine has halted, it is not a Cycler
    \EndIf
    \EndWhile

    \State \textbf{return} false
    \EndProcedure

  \end{algorithmic}
\end{algorithm}

\subsection{Correctness}

\begin{theorem}\label{th:cyclers} Let $\mathcal{M}$ be a Turing machine and $t \in \mathbb{N}$ a time limit. Let $c_0$ be the initial configuration of the machine. There exists $i\in\mathbb{N}$ and $j\in\mathbb{N}$ such that $c_0 \vdash^i c_i \vdash^{j-i} c_i$ with $i < j \leq t$ if and only if {\sc decider-cyclers}($\mathcal{M}$,$t$) returns \texttt{true} (Algorithm~\ref{alg:cyclers}).
\end{theorem}
\begin{proof}
  This follows directly from the behavior of {\sc decider-cyclers}($\mathcal{M}$,$t$): all configurations from $c_0$ to $c_t$ are recorded and the algorithm returns \texttt{true} if and only if one is visited twice. This mathematically translates to
  there exists $i\in\mathbb{N}$ and $j\in\mathbb{N}$ such that $c_0 \vdash^i c_i \vdash^{j-i} c_i$ with $i < j \leq t$, which is what we want. Index $i$ corresponds to the first time that $c_i$ is seen (l.\ref{i:alg:cyclers} in Algorithm~\ref{alg:cyclers}) while index $j$ corresponds to the second time that $c_i$ is seen (l.\ref{j:alg:cyclers} in Algorithm~\ref{alg:cyclers}).
\end{proof}

\begin{corollary}
  Let $\mathcal{M}$ be a Turing machine and $t \in \mathbb{N}$ a time limit. If {\sc decider-cyclers}($\mathcal{M}$,$t$) returns \texttt{true} then the behavior of $\mathcal{M}$ from all-0 tape has been decided: $\mathcal{M}$ does not halt.
\end{corollary}
\begin{proof}
  By Theorem~\ref{th:cyclers}, there exists $i\in\mathbb{N}$ and $j\in\mathbb{N}$ such that $c_0 \vdash^i c_i \vdash^{j-i} c_i$ with $i < j \leq t$. It follows that for all $k\in\mathbb{N}$, $c_0 \vdash^{i+k(j-i)} c_i$. The machine never halts as it will visit $c_i$ infinitely often.
\end{proof}

\subsection{Implementation}

The decider was coded in \texttt{golang} and is accessible at this link: \url{https://github.com/bbchallenge/bbchallenge-deciders/tree/main/decider-cyclers}.



\newpage
\section{Translated Cyclers}\label{sec:translated-cyclers}

\begin{figure}[h!]
  \centering
  \includegraphics[width=0.54\textwidth]{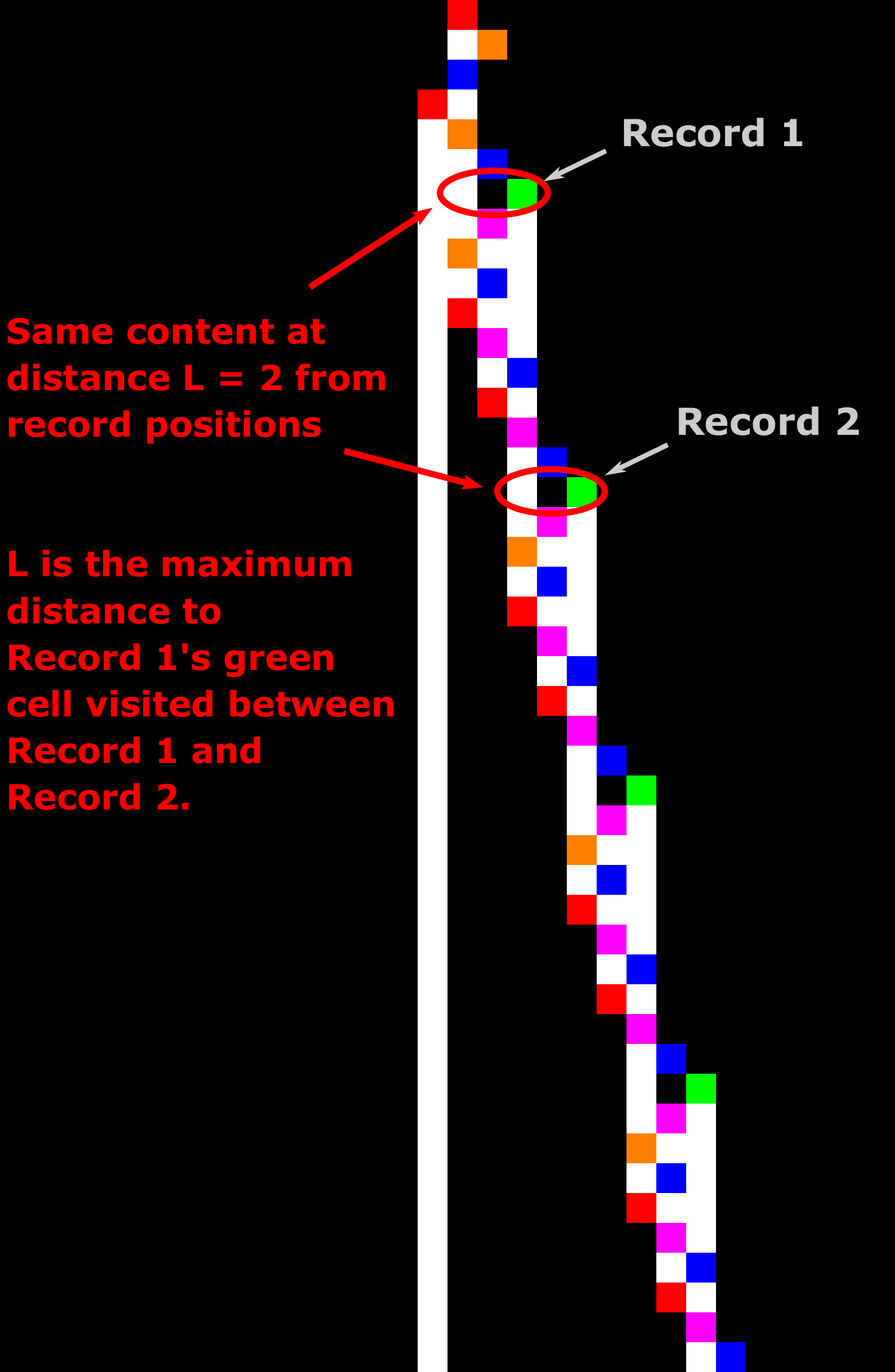}

  \caption{Example ``Translated cycler'': 45-step space-time diagram of bbchallenge's machine \#44,394,115. See \url{https://bbchallenge.org/44394115}. The same bounded pattern is being translated to the right forever. The text annotations illustrate the main idea for recognising ``Translated Cyclers'': find two configurations that break a record (i.e. visit a memory cell that was never visited before) in the same state (here state \textcolor{colorD}{D}) such that the content of the memory tape at distance $L$ from the record positions is the same in both record configurations. Distance $L$ is defined as being the maximum distance to record position 1 that was visited between the configuration of record 1 and record 2.}\label{fig:translated-cyclers}
\end{figure}

The goal of this decider is to recognise Turing machines that translate a bounded pattern forever. We call such machines ``Translated cyclers''. They are close to ``Cyclers'' (Section~\ref{sec:cyclers}) in the sense that they are only repeating a pattern but there is added complexity as they are able to translate the pattern in space at the same time, hence the decider for Cyclers cannot directly apply here.

The main idea for this decider is illustrated in Figure~\ref{fig:translated-cyclers} which gives the space-time diagram of a ``Translated cycler'': bbchallenge's machine \#44,394,115 (c.f. \url{https://bbchallenge.org/44394115}). The idea is to find two configurations that break a record (i.e. visit a memory cell that was never visited before) in the same state (here state \textcolor{colorD}{D}) such that the content of the memory tape at distance $L$ from the record positions is the same in both record configurations. Distance $L$ is defined as being the maximum distance to record position 1 that was visited between the configuration of record 1 and record 2. In those conditions, we can prove that the machine will never halt. Similar ideas have been developed in \cite{Lin1963}.

The translated cycler of Figure~\ref{fig:translated-cyclers} features a relatively simple repeating pattern and transient pattern (pattern occurring before the repeating patterns starts). The translated cycler of Figure~\ref{fig:translated-cyclers-more} features a significantly more complex pattern. The method for detecting the behavior is the same but more resources are needed.

\begin{figure}
  \centering
  \includegraphics[width=0.7\textwidth]{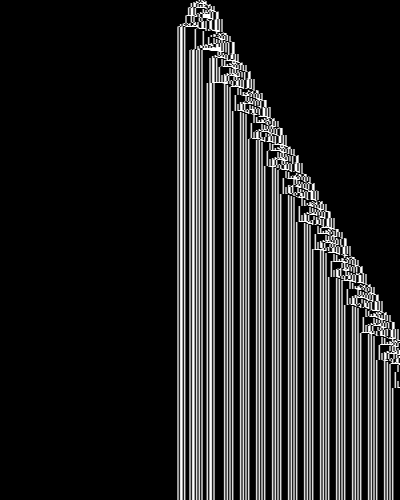}

  \caption{More complex ``Translated cycler'': 10,000-step space-time diagram (no state colours) of bbchallenge's machine \#59,090,563. See \url{https://bbchallenge.org/59090563}.}\label{fig:translated-cyclers-more}
\end{figure}

\subsection{Pseudocode}

We assume that we are given a Turing Machine type \textbf{TM} that encodes the transition table of a machine as well as a procedure \textbf{TuringMachineStep}(machine,configuration) which computes the next configuration of a Turing machine from the given configuration or \textbf{nil} if the machine halts at that step.

One minor complication of the technique described above is that one has to track record-breaking configurations on both sides of the tape: a configuration can break a record on the right or on the left. Also, in order to compute distance $L$ (see above or Definition~\ref{def:distL}) it is useful to add to memory cells the information of the last time step at which it was visited.

We also assume that we are given a routine {\sc get-extreme-position}(tape,sideOfTape) which gives us the rightmost or leftmost position of the given tape (well defined as we always manipulate finite tapes).

\begin{algorithm}
  \caption{{\sc decider-translated-cyclers}}\label{alg:translated-cyclers}

  \begin{algorithmic}[1]
    \State \textbf{const int} RIGHT, LEFT = 0, 1
    \State \textbf{struct} ValueAndLastTimeVisited \{
    \State \tabi\textbf{int} value
    \State \tabi\textbf{int} lastTimeVisited

    \State \}
    \State \textbf{struct} Configuration \{
    \State \tabi\textbf{State} state
    \State \tabi\textbf{int} headPosition
    \State \tabi\textbf{int $\boldsymbol{\to}$ ValueAndLastTimeVisited} tape
    \State \}
    \State

    \Procedure{\textbf{bool} {\sc decider-translated-cyclers}}{\textbf{TM} machine,\textbf{int} timeLimit}
    \State \textbf{Configuration} currConfiguration = \{.state = \textcolor{colorA}{A}, .headPosition = 0, .tape = \{0:\{.value = 0, .lastTimeVisited = 0\}\}\}
    \State // 0: right records, 1: left records
    \State \textbf{List$\boldsymbol{<}$Configuration$\boldsymbol{>}$}
    recordBreakingConfigurations[2] = [[],[]]
    \State \textbf{int} extremePositions[2] = [0,0]
    \State \textbf{int} currTime = 0

    \While{currTime $<$ timeLimit}
    \State \textbf{int} headPosition = currConfiguration.headPosition
    \State currConfiguration.tape[headPosition].lastTimeVisited = currTime
    \If{headPosition $>$ extremePositions[RIGHT] \textbf{or} headPosition $<$ extremePositions[LEFT]}
    \State \textbf{int} recordSide = (headPosition $>$ extremePositions[RIGHT]) ? RIGHT : LEFT
    \State extremePositions[recordSide] = headPosition
    \If{{\sc aux-check-records}(currConfiguration, recordBreakingConfigurations[recordSide], recordSide)}
    \State \textbf{return} true
    \EndIf
    \State recordBreakingConfigurations[recordSide].\textbf{append}(currConfiguration)
    \EndIf

    \State currConfiguration = \textbf{TuringMachineStep}(machine,currConfiguration)
    \State currTime += 1

    \If{currConfiguration == \textbf{nil}}
    \State \textbf{return} false //machine has halted, it is not a Translated Cycler
    \EndIf
    \EndWhile

    \State \textbf{return} false
    \EndProcedure

  \end{algorithmic}
\end{algorithm}
\begin{algorithm}
  \begin{algorithmic}[1]
    \caption{{\sc compute-distance-L} and {\sc aux-check-records}}\label{alg:translated-cyclers-aux}

    \Procedure{\textbf{int} {\sc compute-distance-L}}{\textbf{Configuration} currRecord, \textbf{Configuration} olderRecord, \textbf{int} recordSide}
    \State \textbf{int} olderRecordPos = olderRecord.headPosition
    \State \textbf{int} olderRecordTime = olderRecord.tape[olderRecordPos].lastTimeVisited
    \State \textbf{int} currRecordTime = currRecord.tape[currRecord.headPosition].lastTimeVisited
    \State \textbf{int} distanceL = 0
    \For{\textbf{int} pos \textbf{in} currRecord.tape}
    \If{pos $>$ olderRecordPos \textbf{and} recordSide == RIGHT}
    \textbf{continue}
    \EndIf
    \If{pos $<$ olderRecordPos \textbf{and} recordSide == LEFT}
    \textbf{continue}
    \EndIf

    \State \textbf{int} lastTimeVisited = currRecord.tape[pos].lastTimeVisited
    \If{lastTimeVisited $\geq$ olderRecordTime \textbf{and} lastTimeVisited $\leq$ currRecordTime}
    \State distanceL = \textbf{max}(distanceL,\textbf{abs}(pos-olderRecordPos))
    \EndIf

    \EndFor
    \State \textbf{return} distanceL
    \EndProcedure
    \State
    \Procedure{\textbf{bool} {\sc aux-check-records}}{\textbf{Configuration} currRecord, \textbf{List$\boldsymbol{<}$Configuration$\boldsymbol{>}$} olderRecords, \textbf{int} recordSide}

    \For{\textbf{Configuration} olderRecord \textbf{in} olderRecords}
    \If{currRecord.state != olderRecord.state}
    \State \textbf{continue}
    \EndIf
    \State \textbf{int} distanceL = {\sc compute-distance-L}(currRecord,olderRecord,recordSide)
    \State \textbf{int} currExtremePos = {\sc get-extreme-position}(currRecord.tape,recordSide)
    \State \textbf{int} olderExtremePos = {\sc get-extreme-position}(olderRecord.tape,recordSide)
    \State \textbf{int} step = (recordSide == RIGHT) ? -1 : 1
    \State \textbf{bool} isSameLocalTape = true
    \For{\textbf{int} offset = 0; \textbf{abs}(offset) $\leq$ distanceL; offset += step}
    \If{currRecord.tape[currExtremePos+offset].value != \newline olderRecord.tape[olderExtremePos+offset].value}
    \State isSameLocalTape = false
    \State \textbf{break}
    \EndIf
    \EndFor
    \If{isSameLocalTape}
    \State \textbf{return} true
    \EndIf
    \EndFor
    \State \textbf{return} false
    \EndProcedure

  \end{algorithmic}
\end{algorithm}

\subsection{Correctness}

\begin{definition}[record-breaking configurations]
  Let $\mathcal{M}$ be a Turing machine and $c_0$ its busy beaver initial configuration (i.e. state is 0, head position is 0 and tape is all-0).
  Let $c$ be a configuration reachable from $c_0$, i.e. $c_0 \vdash^* c$.
  Then $c$ is said to be \textit{record-breaking} if the current head position had never been visited before. Records can be broken to the \textit{right} (positive head position) or to the left (negative head position).
\end{definition}

\begin{definition}[Distance $L$ between record-breaking configurations]\label{def:distL}
  Let $\mathcal{M}$ be a Turing machine and $r_1,r_2$ be two record-breaking configurations on the same side of the tape at respective times $t_1$ and $t_2$ with $t_1 < t_2$. Let $p_1$ and $p_2$ be the tape positions of these records. Then, distance $L$ between $r_1$ and $r_2$ is defined as $\max\{|p_1 - p|\}$ with $p$ any position visited by $\mathcal{M}$ between $t_1$ and $t_2$ that is not beating record $p_1$ (i.e. $p \leq p_1$ for a record on the right and $p \geq p_1$ for a record on the left).
\end{definition}

\begin{lemma}\label{lem:translated-cyclers}
  Let $\mathcal{M}$ be a Turing machine. Let $r_1$ and $r_2$ be two configurations that broke a record in the same state and on the same side of the tape at respective times $t_1$ and $t_2$ with $t_1 < t_2$. Let $p_1$ and $p_2$ be the tape positions of these records. Let $L$ be the distance between $r_1$ and $r_2$ (Definition~\ref{def:distL}). If the content of the tape in $r_1$ at distance $L$ of $p_1$ is the same as the content of the tape in $r_2$ at distance $L$ of $p_2$ then $\mathcal{M}$ never halts.
\end{lemma}

\begin{proof}
  Let's suppose that the record-breaking configurations are on the right-hand side of the tape. By the hypotheses, we know the machine is in the same state in $r_1$ and $r_2$ and that the content of the tape at distance $L$ to the left of $p_1$ in $r_1$ is the same as the content of the tape at distance $L$ to the left of $p_2$ in $r_2$. Note that the content of the tape to the right of $p_1$ and $p_2$ is the same: all-0 since they are record positions. Furthermore, by Definition~\ref{def:distL}, we know that distance $L$ is the maximum distance that $\mathcal{M}$ can travel to the left of $p_1$ between times $t_1$ and $t_2$. Hence that after $r_2$, since it will read the same tape content the machine will reproduce the same behavior as it did after $r_1$ but translated at position $p_2$: after $t_2 - t_1$ steps, there will be a record-breaking configuration $r_3$ such that the distance between record-breaking configurations $r_2$ and $r_3$ is also $L$ (Definition~\ref{def:distL}). Hence the machine will keep breaking records to the right forever and will not halt. Analogous proof for records that are broken to the left.
\end{proof}

\begin{theorem}\label{th:translated-cyclers}
  Let $\mathcal{M}$ be a Turing machine and $t$ a time limit. The conditions of Lemma~\ref{lem:translated-cyclers} are met before time $t$ if and only if {\sc decider-translated-cyclers}($\mathcal{M}$,$t$) outputs \texttt{true} (Algorithm~\ref{alg:translated-cyclers}).
\end{theorem}
\begin{proof}
  The algorithm consists of a main function {\sc decider-translated-cyclers} (Algorithm~\ref{alg:translated-cyclers}) and two auxiliary functions {\sc compute-distance-L} and {\sc aux-check-records} (Algorithm~\ref{alg:translated-cyclers-aux}).

  The main loop of {\sc decider-translated-cyclers} (Algorithm~\ref{alg:translated-cyclers} l.18) simulates the machine with the particularity that (a) it keeps track of the last time it visited each memory cell (l.20) and (b) it keeps track of all record-breaking configurations that are met (l.21) before reaching time limit $t$. When a record-breaking configuration is found, it is compared to all the previous record-breaking configurations on the same side in seek of the conditions of Lemma~\ref{lem:translated-cyclers}. This is done by auxiliary routine {\sc aux-check-records} (Algorithm~\ref{alg:translated-cyclers-aux}).

  Auxiliary routine {\sc aux-check-records} (Algorithm~\ref{alg:translated-cyclers-aux}, l.14) loops over all older record-breaking configurations on the same side as the current one (l.15), and only examines older configurations that are in the same state as the current one (l.16). It computes distance $L$ (Definition~\ref{def:distL}) between the older and the current record-breaking configuration (l.18). This computation is done by auxiliary routine {\sc compute-distance-L}.

  Auxiliary routine {\sc compute-distance-L} (Algorithm~\ref{alg:translated-cyclers-aux}, l.1) uses the ``pebbles'' that were left on the tape to give the last time a memory cell was seen (field \texttt{lastTimeVisited}) in order to compute the farthest position from the old record position that was visited before meeting the new record position (l.10). Note that we discard intermediate positions that beat the old record position (l.7-8) as we know that the part of the tape after the record position in the old record-breaking configuration is all-0, same as the part of the tape after current record position in the current record-breaking position (part of the tape to the right of the red-circled green cell in Figure~\ref{fig:translated-cyclers}).

  Thanks to the computation of {\sc compute-distance-L} the routine {\sc aux-check-records} is able to check whether the tape content at distance $L$ of the record-breaking position in both record-holding configurations is the same or not (Algorithm~\ref{alg:translated-cyclers-aux}, l.23). The routine returns \texttt{true} if they are the same and the function {\sc decider-translated-cyclers} will return \texttt{true} as well in cascade (Algorithm~\ref{alg:translated-cyclers} l.24). That scenario is reached if and only if the algorithm has found two record-breaking configurations on the same side that satisfy the conditions of Lemma~\ref{lem:translated-cyclers}, which is what we wanted.
\end{proof}

\begin{corollary}
  Let $\mathcal{M}$ be a Turing machine and $t \in \mathbb{N}$ a time limit. If {\sc decider-translated-cyclers}($\mathcal{M}$,$t$) returns \texttt{true} then the behavior of $\mathcal{M}$ from all-0 tape has been decided: $\mathcal{M}$ does not halt.
\end{corollary}
\begin{proof}
  Immediate by combining Lemma~\ref{lem:translated-cyclers} and Theorem~\ref{th:translated-cyclers}.
\end{proof}

\subsection{Implementation}

The decider was coded in \texttt{golang} and is accessible at this link: \url{https://github.com/bbchallenge/bbchallenge-deciders/tree/main/decider-translated-cyclers}.



\newpage
\section{Backward Reasoning}\label{sec:backward-reasoning}

\begin{figure}
  \centering
  \begin{subfigure}[m]{0.45\textwidth}
    \centering
    \includegraphics[width=\textwidth]{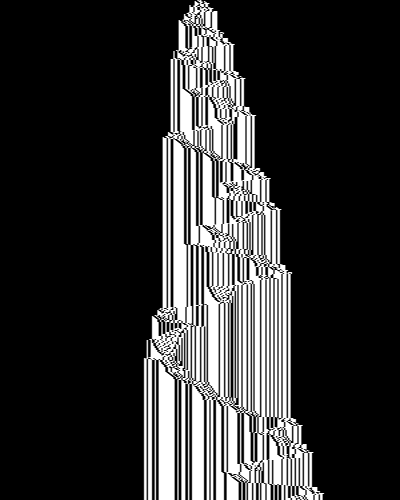}
    \caption{10,000-step space-time diagram of bbchallenge's machine \#55,897,188. \url{https://bbchallenge.org/55897188}}
    \label{fig:y equals x}
  \end{subfigure}
  \hfill
  \begin{subfigure}[m]{0.45\textwidth}
    \centering
    \begin{tabular}{lll}
                            & 0                       & 1   \\
      \textcolor{colorA}{A} & 1R\textcolor{colorB}{B} & 0LD \\
      \textcolor{colorB}{B} & 1L\textcolor{colorC}{C} & 0RE \\
      \textcolor{colorC}{C} & - - -                   & 1LD \\
      D                     & 1LA                     & 1LD \\
      E                     & 1RA                     & 0RA
    \end{tabular}

    \caption{Transition table of machine \#55,897,188.}

  \end{subfigure}

  \begin{subfigure}[m]{1\textwidth}
    \vspace{5ex}
    \centering
    \includegraphics[width=0.9\textwidth]{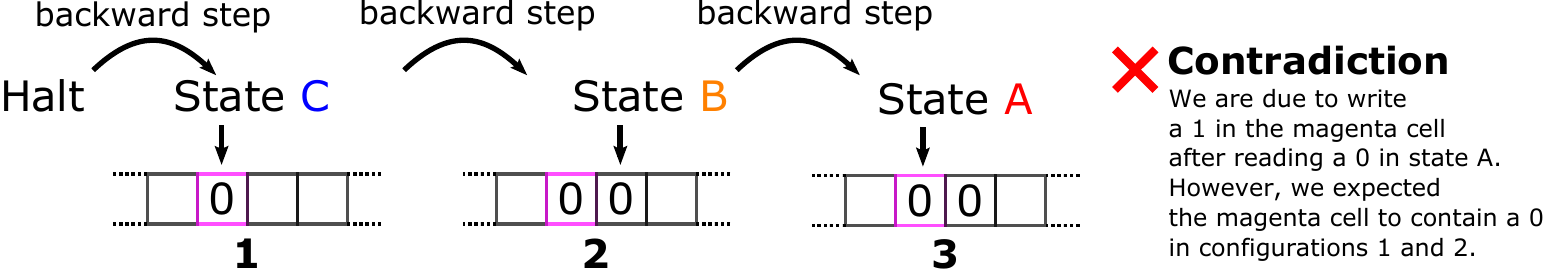}

    \caption{Contradiction reached after 3 backward steps: machine \#55,897,188 cannot reach its halting configuration hence it does not halt.}

  \end{subfigure}

  \caption{Applying backward reasoning on bbchallenge's machine \#55,897,188. (a) 10,000-step space-time diagram of machine \#55,897,188. The \textit{forward} behavior of the machine looks very complex. (b) Transition table. (c) We are able to deduce that the machine will never halt thanks to only 3 backward reasoning steps: because a contradiction is met, it is impossible to reach the halting configuration in more than 3 steps -- and, by (a), the machine did not halt in 10,000 steps starting from all-0 tape.}
  \label{fig:backward-reasoning}
\end{figure}

Backward reasoning, as described in \cite{Marxen_1990}, takes a different approach than what has been done with deciders in Sections~\ref{sec:cyclers} and \ref{sec:translated-cyclers}. Indeed, instead of trying to recognise a particular kind of machine's behavior, the idea of backward reasoning is to show that, independently of the machine's behavior, the halting configurations are not reachable. In order to do so, the decider simulates the machine \textit{backwards} from halting configurations until it reaches some obvious contradiction.

Figure~\ref{fig:backward-reasoning} illustrates this idea on bbchallenge's machine \#55,897,188. From the space-time diagram, the \textit{forward} behavior of the machine from all-0 tape looks to be extremely complex, Figure~\ref{fig:backward-reasoning}a. However, by reconstructing the sequence of transitions that would lead to the halting configuration (reading a 0 in state \textcolor{colorC}{C}), we reach a contradiction in only 3 steps, Figure~\ref{fig:backward-reasoning}c. Indeed, the only way to reach state \textcolor{colorC}{C} is to come from the right in state \textcolor{colorB}{B} where we read a 0. The only way to reach state \textcolor{colorB}{B} is to come from the left in state \textcolor{colorA}{A} where we read a 0. However, the transition table (Figure~\ref{fig:backward-reasoning}b) is instructing us to write a 1 in that case, which is not consistent with the 0 that we assumed was at this position in order for the machine to halt.

Backward reasoning in the case of Figure~\ref{fig:backward-reasoning} was particularly simple because there was only one possible previous configuration for each backward step -- i.e. there is only one transition that can reach state \textcolor{colorC}{C} and same for state \textcolor{colorB}{B}. In general, this is not the case and the structure created by backward reasoning is a tree of configurations instead of just a chain. If all the leaves of a backward reasoning tree of depth $D$ reach a contradiction, we know that if the machine runs for $D$ steps from all-0 tape then the machine cannot reach a halting configuration and thus does not halt.

\subsection{Pseudocode}

\begin{algorithm}
  \caption{{\sc decider-backward-reasoning}}\label{alg:backward-reasoning}

  \begin{algorithmic}[1]
    \State{\textbf{const int} RIGHT, LEFT = 0, 1}
    \State \textbf{struct} Transition \{
    \State \tabi\textbf{State} state
    \State \tabi\textbf{int} read, write, move
    \State \}
    \State \textbf{struct} Configuration \{
    \State \tabi\textbf{State} state
    \State \tabi\textbf{int} headPosition
    \State \tabi\textbf{int $\boldsymbol{\to}$ int} tape
    \State \tabi\textbf{int} depth
    \State \}

    \State

    \Procedure{\textbf{Configuration} {\sc apply-transition-backwards}}{\textbf{Configuration} conf,\textbf{Transition} t}
    \State \textbf{int} reversedHeadMoveOffset = (t.move == RIGHT) ? -1 : 1
    \State \textbf{int} previousPosition = conf.headPosition+reversedHeadMoveOffset
    \If{previousPosition \textbf{in} conf.tape \textbf{and} conf.tape[previousPosition] != t.write}
    \State \textbf{return} \textbf{nil} // Backward contradiction spotted
    \EndIf
    \State \textbf{Configuration} previousConf = \{.state = t.state, .headPosition = previousPosition, .tape = conf.tape, .depth = conf.depth + 1\}
    \State previousConf.tape[previousPosition] = t.read
    \State \textbf{return} previousConf
    \EndProcedure
    \State
    \Procedure{\textbf{bool} {\sc decider-backward-reasoning}}{\textbf{TM} machine,\textbf{int} maxDepth}

    \State \textbf{Stack$\boldsymbol{<}$Configuration$\boldsymbol{>}$} configurationStack
    \For{\textbf{Pair$\boldsymbol{<}$State, int$\boldsymbol{>}$} (state,read) \textbf{in} {\sc get-undefined-transitions}(machine)}
    \State \textbf{Configuration} haltingConfiguration = \{.state = state, .headPosition = 0, .tape = \{0: read\}, .depth = 0\}
    \State configurationStack.\textbf{push}(haltingConfiguration)
    \EndFor
    \While{!configurationStack.\textbf{empty}()}
    \State \textbf{Configuration} currConf = configurationStack.\textbf{pop}()
    \If{currConf.depth $>$ maxDepth} \textbf{return} false \EndIf

    \For{\textbf{Transition} transition \textbf{in} {\sc get-transitions-reaching-state}(machine,currConf.state)}
    \State \textbf{Configuration} previousConf = {\sc apply-transition-backwards}(currConf, transition)
    \State // If no contradiction
    \If{previousConf != \textbf{nil}}
    \State configurationStack.\textbf{push}(previousConf)
    \EndIf
    \EndFor
    \EndWhile

    \State \textbf{return} true
    \EndProcedure

  \end{algorithmic}
\end{algorithm}

We assume that we are given routine {\sc get-undefined-transitions}(machine) which returns the list of (state,readSymbol) pairs of all the undefined transitions in the machine's transition table, for instance [(\textcolor{colorC}{C},0)] for the machine of Figure~\ref{fig:backward-reasoning}b. We also assume that we are given routine {\sc get-transitions-reaching-state}(machine,targetState) which returns the list of all machine's transitions that go to the specified target state, for instance [(\textcolor{colorA}{A},1,0LD),(\textcolor{colorC}{C},1,1LD),(D,1,1LD)] for target state D in the machine of Figure~\ref{fig:backward-reasoning}b. These two routines contain very minimal logic as they only lookup in the description of the machine for the required information.

\subsection{Correctness}

\begin{theorem}\label{th:backward-reasoning}
  Let $\mathcal{M}$ be a Turing machine and $D\in\mathbb{N}$.
  Then, {\sc decider-backward-reasoning}($\mathcal{M}$,$D$) returns \texttt{true} if and only if no undefined transition of $\mathcal{M}$ can be reached in more than $D$ steps.
\end{theorem}
\begin{proof}
  The tree of backward configurations is maintained in a DFS fashion through a stack (Algorithm~\ref{alg:backward-reasoning}, l.23). Initially, the stack is filled with the configurations where only one tape cell is defined and state is set such that the corresponding transition is undefined (i.e. the machine halts after that step), l.24-26.

  Then, the main loop runs until either (a) the stack is empty or (b) one leaf exceeded the maximum allowed depth, l.27 and l.29. Note that running the algorithm with increased maximum depth increases its chances to contradict all branches of the backward simulation tree. At each step of the loop, we remove the current configuration from the stack and we try to apply all the transitions that lead to this configuration backwards by calling routine {\sc apply-transition-backwards}(configuration, transition).

  The only case where it is not possible to apply a transition backwards, i.e. the case where a contradiction is reached, is when the tape symbol at the position where the transition comes from (i.e. to the right if transition movement is left and vice-versa) is defined but is not equal to the write instruction of the transition. Indeed, that means that the future (i.e. previous backward steps) is not consistent with the current transition's write instruction. This logic is in l.16. Otherwise, we can construct the previous configuration (i.e. next backward step) and augment depth by 1. We then stack this configuration in the main routine (l.34).

  The algorithm returns \texttt{true} if and only if the stack ever becomes empty which means that all leaves of the backward simulation tree of depth $D$ have reached a contradiction and thus, no undefined transition of the machine is reachable in more than $D$ steps.
\end{proof}

\begin{corollary}
  Let $\mathcal{M}$ be a Turing machine and $D\in\mathbb{N}$. If {\sc decider-backward-reasoning}($\mathcal{M}$,$D$) returns \texttt{true} and machine $\mathcal{M}$ can run $D$ steps from all-0 tape without halting then the behavior of $\mathcal{M}$ from all-0 tape has been decided: $\mathcal{M}$ does not halt.
\end{corollary}
\begin{proof}
  By Theorem~\ref{th:backward-reasoning} we know that no undefined transition of $\mathcal{M}$ can be reached in more than $D$ steps. Hence, if machine $\mathcal{M}$ can run $D$ steps from all-0 tape without halting, it will be able to run the next $D+1^{\text{th}}$ step. From there, the machine cannot halt or it would contradict the fact that halting trajectories have at most $D$ steps. Hence, $\mathcal{M}$ does not halt from all-0 tape.
\end{proof}

\subsection{Implementation}\label{sec:backward-reasoning-results}

The decider was coded in \texttt{golang} and is accessible at this link: \url{https://github.com/bbchallenge/bbchallenge-deciders/blob/main/decider-backward-reasoning}. Note that collaborative work allowed to find a bug in the initial algorithm that was implemented\footnote{Thanks to collaborators \url{https://github.com/atticuscull} and \url{https://github.com/modderme123}.}.



\newcommand{\HS}{Halting Segment\xspace}

\newpage
\section{Halting Segment}\label{sec:halting-segment}

\paragraph{Acknowledgement.} Sincere thanks to bbchallenge's contributor Iijil who initially presented this method and the first implementation\footnote{See: \url{https://discuss.bbchallenge.org/t/decider-halting-segment}.}. Other contributors have contributed to this method by producing alternative implementations (see Section~\ref{sec:hs-implem}) or discussing and writing the formal proof presented here: Mateusz Naściszewski (Mateon1), Nathan Fenner, Tony Guilfoyle, Justin Blanchard, Tristan Stérin (cosmo), and, Pavel Kropitz (uni).

\subsection{Overview}

\begin{figure}[h!]
  \centering
  \includegraphics[width=1\textwidth]{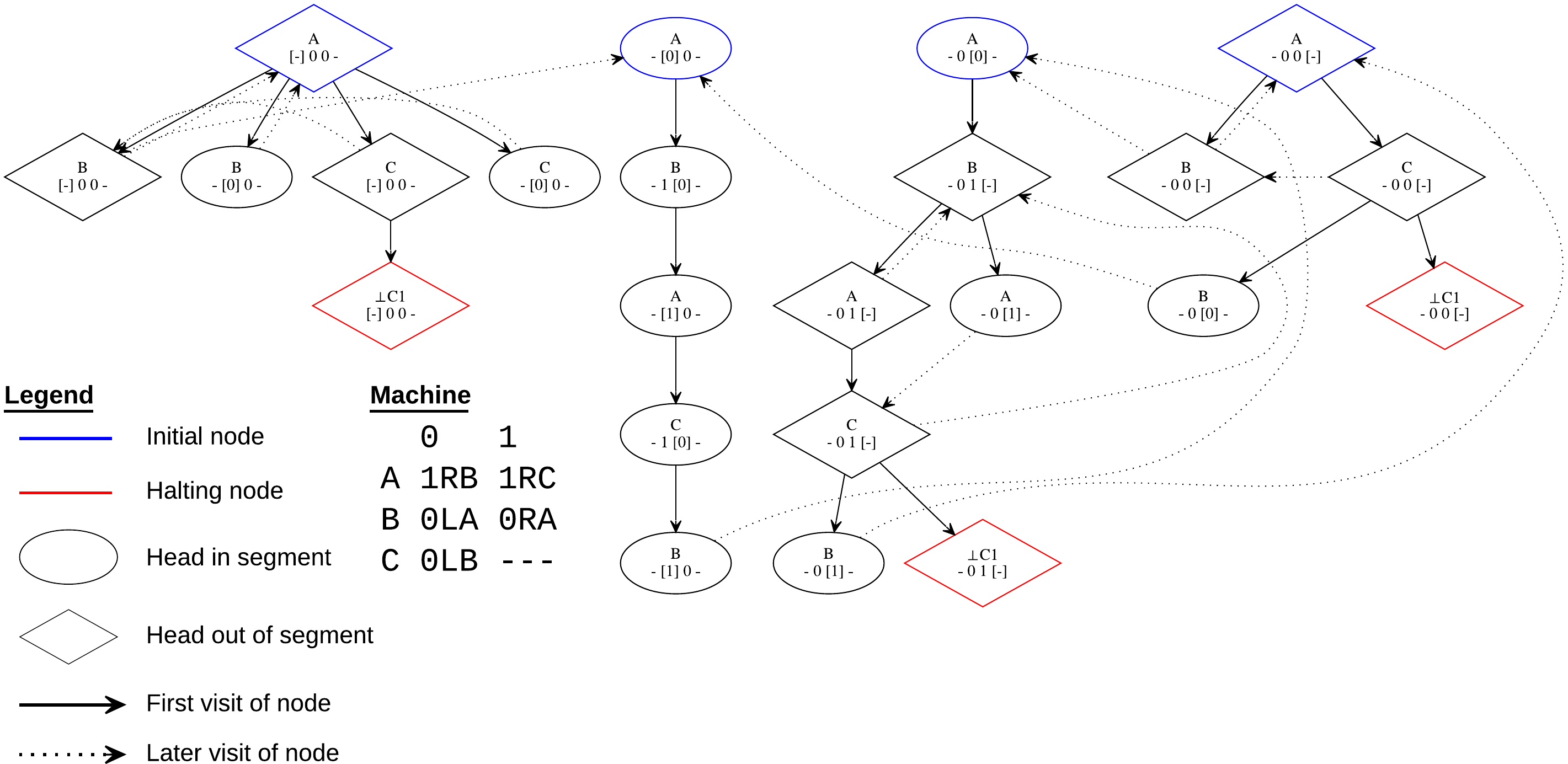}
  \caption{\HS graph for the 3-state machine \url{https://bbchallenge.org/1RB1RC_0LA0RA_0LB---} and segment size 2, see Definition~\ref{def:hs-graph}. Nodes of this graph correspond to \textit{segment configurations} (Definition~\ref{def:hs-conf}), i.e. configurations of the machine on a finite segment (here, of size 2). In a node, the machine's head position is represented between brackets and the symbol \texttt{-} represents the outside of the segment (either to the left or to the right). Nodes where the machine's head is within the segment (circle shape) only one have child corresponding to the next step of the machine and nodes where the head is outside of the segment (diamond shape) may have multiple children corresponding to all the theoretically possible ways (deduced from the machine's transition table) that the machine can enter the segment back or continue to stay out of it. In order to improve readability, edges that revisit a node are dotted. The machine presented here does not halt because the halting nodes (red outline) that are reachable from the initial nodes (blue outline) do not cover all the positions of the segment (there is no halting node for any of the two internal positions of the segment), by contraposition of Theorem~\ref{th:hs}. }\label{fig:hs}
\end{figure}

The idea of the \HS technique is to simulate a Turing machine on a finite segment of tape. When the machine leaves the segment in a certain state, we consider all the possible ways that it can re-enter the segment or stay out of it, based on the machine's transition table. For a given machine and segment size, this method naturally gives rise to a graph, the \HS graph (formally defined in Definition~\ref{def:hs-graph}).

Figure~\ref{fig:hs} gives the \HS graph of the 3-state machine\footnote{We chose a 3-state machine in order to have a graph of reasonable size.} \url{https://bbchallenge.org/1RB1RC_0LA0RA_0LB---} for segment size 2. Let's describe this graph in more details:

\begin{itemize}
  \item Nodes correspond to \textit{segment configurations} (Definition~\ref{def:hs-conf}), i.e. the state in which the machine is together with the content of the segment and the position of the head in the segment (or outside of it). For instance, the leftmost node in blue and diamond shape in Figure~\ref{fig:hs} is \texttt{A [-] 0 0 -} which means that the machine is in state A, that the segment currently contains \texttt{0 0} and that the machine's head is currently outside of the segment, to the left of it.

  \item Initial nodes (blue outline) correspond to all segment configurations that match the initial configuration of the machine (all-0 tape and state A). There are $n+2$ initial nodes with $n$ the size of the segment. Halting nodes (red outline) give the segment configurations where the machine has halted together with the halting transition that was used. For instance, in Figure~\ref{fig:hs}, the leftmost halting node \texttt{$\bot$ C1 [-] 0 0 -} signifies that the machine has halted ($\bot$), using halting transition \texttt{C1} (reading a 1 in state C), to the left of the segment which contains \texttt{0 0}.

  \item Nodes with a circle shape correspond to segment configurations where the tape's head is \textbf{inside} the segment. Such nodes only have one child, which corresponds to the next machine configuration.

  \item Nodes with a diamond shape correspond to segment configurations where the head is \textbf{outside} the segment. These nodes may have several children corresponding to all the ways that the head, in the current state, can stay outside of the segment or enter it back. For instance, the leftmost node in blue and diamond shape in Figure~\ref{fig:hs}, \texttt{A [-] 0 0 -}, has 4 children: \texttt{B [-] 0 0 -} and \texttt{B - [0] 0 -} and \texttt{C [-] 0 0 -} and \texttt{C - [0] 0 -}. This is because the transitions of the machine in state A are \texttt{1RB} and \texttt{1RC} and that the move \texttt{R} allows either to enter the segment back or to continue being out of it (if the head is far from the segment's left frontier). Note that the write symbol \texttt{1} of the transitions are ignored since we do not keep track of the tape outside of the segment.

  \item In order to increase the readability of Figure~\ref{fig:hs}, only one entrant edge for each node has been drawn with a solid line, corresponding to the first visit of that node in the particular order that the graph was visited. Later visits were drawn with a dotted line.
\end{itemize}

What is special about the \HS graph? We show in Theorem~\ref{th:hs} that if a machine halts, then, for all segment sizes, its \HS graph contains a set of halting nodes (red outline), for the same halting transition, that covers the entire segment and its outside, i.e. such that there is at least one such node per segment's position and outside of it (left and right). By contraposition, if there is no set of covering halting nodes for a halting transition, the machine does not halt. In Figure~\ref{fig:hs}, we deduce that machine \url{https://bbchallenge.org/1RB1RC_0LA0RA_0LB---} does not halt since the halting nodes of halting transition \texttt{C1} are \texttt{$\bot$ C1 [-] 0 0 -}, \texttt{$\bot$ C1 - 0 1 [-]} and \texttt{$\bot$ C1 - 0 0 [-]} which does not cover the entire segment (both internal segment positions are not covered).

Interestingly, \HS is the method that was used by Newcomb Greenleaf to prove\footnote{\url{http://turbotm.de/~heiner/BB/TM4-proof.txt}} that Marxen \& Buntrock's chaotic machine\footnote{\url{https://bbchallenge.org/76708232}} \cite{Marxen_1990} does not halt.

\subsection{Formal proof}

\begin{definition}[Segment configurations]\label{def:hs-conf}
  Let the \textit{segment size} be $n \in \N$. A \textit{segment configuration} is a 3-tuple: (i) state, (ii) $w \in \{0,1\}^n$ which is the segment's content and (iii) the position of the machine's head is an integer $p \in \llbracket -1, n \rrbracket$ where positions $\llbracket 0,n \llbracket$ correspond to the interior of the segment, position $-1$ for outside to the left and $n$ for outside to the right. \textit{Halting segment configurations} are segment configurations where the state is $\bot$ and with an additional information (iv) of which halting transition of the machine has been used to halt.
\end{definition}

\begin{example}
  In Figure~\ref{fig:hs} we have $n=2$ and, the leftmost node in blue and diamond shape corresponds to segment configuration \texttt{A [-] 0 0 -} (i) state A, (ii) $w = \texttt{00}$ and (iii) $p = -1$. The rightmost node in red and diamond shape corresponds to halting segment configuration $\bot$ \texttt{C1 - 0 0 [-]} (i) state $\bot$, (ii) $w = \texttt{00}$, (iii) $p = 2$ and (iv) halting transition \texttt{C1}.
\end{example}

\begin{definition}[\HS graph]\label{def:hs-graph}
  Let $\mathcal{M}$ be a Turing machine and $n \in \N$ a segment size. The \HS graph for $\mathcal{M}$ and $n$ is a directed graph where the nodes are segment configurations (Definition~\ref{def:hs-conf}). The graph is generated from $n+2$ \textit{initial nodes} (blue outline in Figure~\ref{fig:hs}) that are all in state A with segment content $0^n$ ($n$ consecutive 0s) but where the head is at each of the $n+2$ possible positions, one per each initial node, see the blue nodes in Figure~\ref{fig:hs} for an example.
  Then, edges that go out of a given node $r$ are defined as follows:
  \begin{itemize}
    \item If $r$'s head position is inside the segment (circle nodes in Figure~\ref{fig:hs}), then $r$ only has one child corresponding to the next simulation step for machine $\mathcal{M}$. For instance, in Figure~\ref{fig:hs}, node \texttt{A - [0] 0 -} has a unique child \texttt{B - 1 [0] -}, following machine's transition \texttt{A0} which is \texttt{1RB}. That child can be a halting segment configuration if the transition to take is halting.
    \item If $r$'s head position is outside the segment (diamond nodes in Figure~\ref{fig:hs}), then, we consider each transition of $r$'s state. There are three cases:
          \begin{enumerate}
            \item If the transition is halting, we add a child to $r$ which is the halting segment configuration node corresponding to this transition. For instance, in Figure~\ref{fig:hs}, \texttt{C [-] 0 0 -} has halting child $\bot$ \texttt{C1 [-] 0 0 -} corresponding to halting transition \texttt{C1}.
            \item If the transition's movement goes further away from the segment (e.g. we are to the left of the segment, $p=-1$, and the transition movement is \texttt{L}), we add one child for this transition that only differs from its parent in the new state that it moves into. For instance, in Figure~\ref{fig:hs}, \texttt{A - 0 0 [-]} has child \texttt{B - 0 0 [-]} for transition \texttt{A0} which is \texttt{1RB}.
            \item If the transition's movement goes in the direction of the segment (e.g. we are to the left of the segment, $p=-1$, and the transition movement is \texttt{R}), we add two children for this transition. One corresponding to the case where that movement is made at the border of the segment and allows to re-enter the segment and the other one corresponding to the case where that movement is made farther away from the border and does not re-enter yet. For instance, in Figure~\ref{fig:hs}, node \texttt{A [-] 0 0 -} has children \texttt{B [-] 0 0 -} and \texttt{B - [0] 0 -} for transition \texttt{A0} which is \texttt{1RB}.
          \end{enumerate}
  \end{itemize}

  Halting nodes are nodes corresponding to halting segment configurations (red outline in Figure~\ref{fig:hs}).

\end{definition}

\begin{theorem}[\HS]\label{th:hs}
  Let $\mathcal{M}$ be a Turing machine and $n \in \N$ a segment size. Let $G$ be the \HS graph for $\mathcal{M}$ and $n$ (Definition~\ref{def:hs-graph}). If $\mathcal{M}$ halts in halting transition $T$ when started from state A and all-0 tape, then $G$ must contain a halting node for transition $T$ for each of the $n+2$ possible values of the head's position $p \in \llbracket -1, n \rrbracket$.
\end{theorem}
\begin{proof}
  Consider the trace of configurations of $\mathcal{M}$ (full configurations, not segment configurations, as defined in Section~\ref{sec:conventions}) from the initial configuration (state A and all-0 tape) to the halting configuration which happens using halting transition $T$. Starting from the halting configuration, construct the halting segment configuration (with segment size $n$) for $T$ using any position $p \in \llbracket -1, n \rrbracket$ in the segment and fill the segment's content from what is written on the tape around the head in the halting configuration of $\mathcal{M}$. From there, work your way up to the initial configuration: at each step construct the associated segment configuration. This sequence of segment configurations constitute a set of nodes in the \HS graph $G$ of $\mathcal{M}$ for segment size $n$ such that each node points to the next one. At the top of that chain there will be a node matching the initial configuration: state A, all-0 segment and head position somewhere in $\llbracket -1, n \rrbracket$, i.e. an initial node.

  Hence we have shown that all halting nodes for transition $T$ for each of the $n+2$ possible values of the head's position $p \in \llbracket -1, n \rrbracket$ are reachable from some initial node(s).
\end{proof}

\begin{remark}
  By contraposition of Theorem~\ref{th:hs}, if, for all halting transitions $T$ there is at least one halting node (red outline in Figure~\ref{fig:hs}) for some position in the segment that is not reachable from one of the initial node (blue outline in Figure~\ref{fig:hs}) then the machine does not halt. That way, in Figure~\ref{fig:hs}, we can conclude that machine \url{https://bbchallenge.org/1RB1RC_0LA0RA_0LB---} does not halt since the halting nodes of halting transition \texttt{C1} are \texttt{$\bot$ C1 [-] 0 0 -}, \texttt{$\bot$ C1 - 0 1 [-]} and \texttt{$\bot$ C1 - 0 0 [-]} which does not cover the entire segment (both internal segment positions are not covered).

  Note that if all of the segment's positions are covered for some halting transition, we cannot conclude that the machine does not halt, but it does not mean that the machine necessarily halts either.
\end{remark}

\begin{remark}
  Some non-halting machines cannot be decided using \HS for any segment size. Such a machine is for instance \url{https://bbchallenge.org/1RB---_1LC0RB_1LB1LA}.
\end{remark}

\subsection{Implementations}\label{sec:hs-implem}

Here are the implementations of the method that were realised. Almost all of them construct the \HS graph from the halting nodes (backward implementation) instead of from the initial nodes (forward implementation):

\begin{enumerate}
  \item Iijil's who originally proposed the method, \url{https://github.com/bbchallenge/bbchallenge-deciders/tree/main/decider-halting-segment}, and was independently reproduced by Tristan Stérin (cosmo) \url{https://github.com/bbchallenge/bbchallenge-deciders/tree/main/decider-halting-segment-reproduction} (backward implementation)
  \item Mateusz Naściszewski (Mateon1)'s: \url{https://gist.github.com/mateon1/7f5e10169abbb50d1537165c6e71733b} (forward implementation)
  \item Nathan Fenner's which has the interesting feature of being written in a language for formal verification (Dafny): \url{https://github.com/Nathan-Fenner/bbchallenge-dafny-deciders/blob/main/halting-segment.dfy} (backward implementation)
  \item Tony Guilfoyle: \url{https://github.com/TonyGuil/bbchallenge/tree/main/HaltingSegments} (backward implementation)
\end{enumerate}

Iijil's implementation (1) is a bit different from what is presented in this document because the \HS graph is constructed backward (i.e. from the halting nodes instead of from the initial nodes). Also, the method adopts a lazy strategy consisting in testing only odd segment sizes (up to size $n_\text{max}$) and placing the head's position at the center of the tape. Finally, the information of state is not stored for nodes where the head is outside the segment. These implementation choices make the implementation a bit weaker than what was presented here.



\newpage
\section{Finite Automata Reduction (FAR)}\label{sec:finite-automata-reduction}

\paragraph{Acknowledgement.} Sincere thanks to bbchallenge's contributor Justin Blanchard who initially presented this method and the first implementation\footnote{See: \url{https://discuss.bbchallenge.org/t/decider-finite-automata-reduction/}.}. Others have contributed to this method by producing alternative implementations (see Section~\ref{sec:far-implem}) or discussing and writing the formal proof presented here: Tony Guilfoyle, Tristan Stérin (cosmo), Nathan Fenner, Mateusz Naściszewski (Mateon1), Konrad Deka, Iijil, Shawn Ligocki. 

\subsection{Method overview}\label{far-overview}

The core idea of the method presented in this section is to find, for a given Turing machine, a regular language that contains the set of the machine's eventually-halting configurations (with finitely many 1s). Then, provided that the all-0 configuration is not in the regular language, we know that the machine does not halt.

A dual idea has been explored by other authors under the name Closed Tape Languages (CTL) as described in S. Ligocki's blog \cite{ShawnCTL} and credited to H. Marxen in collaboration with J. Buntrock.
The CTL technique for proving a Turing machine doesn't halt is to exhibit a set $C$ of configurations such that:

\begin{enumerate}
  \item $C$ contains the all-0 initial configuration\footnote{
          Criteria 1--2 give a strict definition; in \cite{ShawnCTL}, $C$ only needs to contain some descendant of the initial configuration and some descendant of the successor to each $c\in C$.
          In that case, the set of ancestor configurations to those in $C$ meets the strict definition.
        }
  \item $C$ is \textit{closed} under transitions: for any $c \in C$, the configuration one step later belongs to $C$\addtocounter{footnote}{-1}\addtocounter{Hfootnote}{-1}\footnotemark
  \item $C$ does not contain any halting configuration
\end{enumerate}

If such a set $C$ exists then the machine does not halt.
The CTL approach has proven to be practical and powerful when we search for $C$ among regular languages \cite{ShawnCTL} \cite{BruteforceCTL}.

Here, we develop an original \textit{co-CTL} technique\footnote{By co-CTL we mean a set whose complement is a CTL, characterized by closure criteria inverse---or equivalently converse---to 1--3. In other words, a co-CTL contains all halting configurations, any configuration which can \emph{precede} any member configuration by one TM transition, and not the initial configuration.}, based on the algebraic description of Nondeterministic Finite Automata (NFA), for finding a regular language which contains a machine's eventually halting configurations (in general a superset).

One important aspect of the technique is that, given a Turing machine and its constructed NFA---if found---it is a computationally simple task to verify that the NFA's language does indeed recognise all eventually-halting configurations (with finitely many 1s) of the machine.

\usetikzlibrary {automata, positioning}

\begin{figure}
  \begin{center}
    \begin{tikzpicture}[scale=0.2]
      \tikzstyle{every node}+=[inner sep=0pt]
      \draw [black] (32.4,-20.7) circle (3);
      \draw (32.4,-20.7) node {$X$};
      \draw [black] (19.9,-32.9) circle (3);
      \draw (19.9,-32.9) node {$Y$};
      \draw [black] (43.9,-32.9) circle (3);
      \draw (43.9,-32.9) node {$Z$};
      \draw [black] (43.9,-32.9) circle (2.4);
      \draw [black] (12.2,-32.9) -- (16.9,-32.9);
      \fill [black] (16.9,-32.9) -- (16.1,-32.4) -- (16.1,-33.4);
      \draw [black] (32.4,-13.6) -- (32.4,-17.7);
      \fill [black] (32.4,-17.7) -- (32.9,-16.9) -- (31.9,-16.9);
      \draw [black] (31.917,-23.654) arc (-16.01031:-75.38142:12.771);
      \fill [black] (31.92,-23.65) -- (31.22,-24.28) -- (32.18,-24.56);
      \draw (29.58,-29.75) node [below] {$\alpha$};
      \draw [black] (20.894,-35.718) arc (47.15723:-240.84277:2.25);
      \draw (18.84,-40.19) node [below] {$\alpha$};
      \fill [black] (18.27,-35.4) -- (17.36,-35.65) -- (18.09,-36.33);
      \draw [black] (20.283,-29.932) arc (165.64786:102.96041:12.277);
      \fill [black] (29.42,-21.01) -- (28.53,-20.7) -- (28.76,-21.68);
      \draw (22.58,-23.71) node [above] {$\beta$};
      \draw [black] (34.46,-22.88) -- (41.84,-30.72);
      \fill [black] (41.84,-30.72) -- (41.66,-29.79) -- (40.93,-30.48);
      \draw (37.62,-28.27) node [left] {$\beta$};
      \draw [black] (34.454,-18.529) arc (164.32314:-123.67686:2.25);
      \draw (39.49,-17.55) node [right] {$\beta$};
      \fill [black] (35.37,-21.01) -- (36.01,-21.71) -- (36.28,-20.74);
      \draw [black] (45.807,-35.201) arc (67.3925:-220.6075:2.25);
      \draw (46.66,-40.1) node [below] {$\alpha,\beta$};
      \fill [black] (43.23,-35.81) -- (42.47,-36.36) -- (43.39,-36.74);
      \draw [black] (41.27,-34.338) arc (-65.18014:-114.81986:22.321);
      \fill [black] (22.53,-34.34) -- (23.05,-35.13) -- (23.47,-34.22);
      \draw (31.9,-36.9) node [below] {$\alpha$};
    \end{tikzpicture}
  \end{center}
  \caption{Example Nondeterministic Finite Automaton (NFA) with 3 states X, Y and Z, alphabet $\mathcal{A} = \{\alpha,\beta\}$, initial states X and Y, and accepting state Z. The linear-algebra representation of this NFA is given in Example~\ref{ex:nfa}. Example accepted words are: $\beta$, $\alpha\beta$, $\alpha\alpha\beta\beta$. Example rejected words are: $\alpha$, $\alpha\alpha$, $\alpha\alpha\alpha$.}\label{fig:example_nfa}
\end{figure}
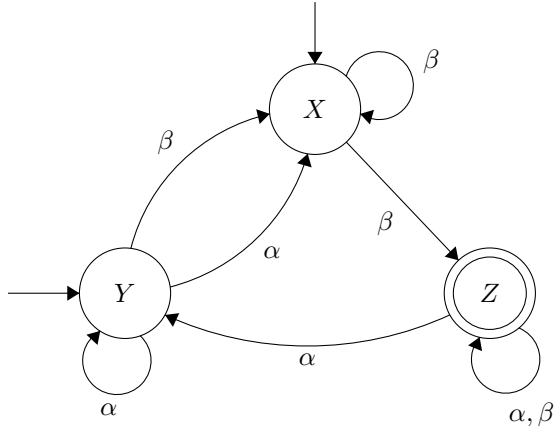

\begin{figure}

  \centering
  \includegraphics*[width=\textwidth]{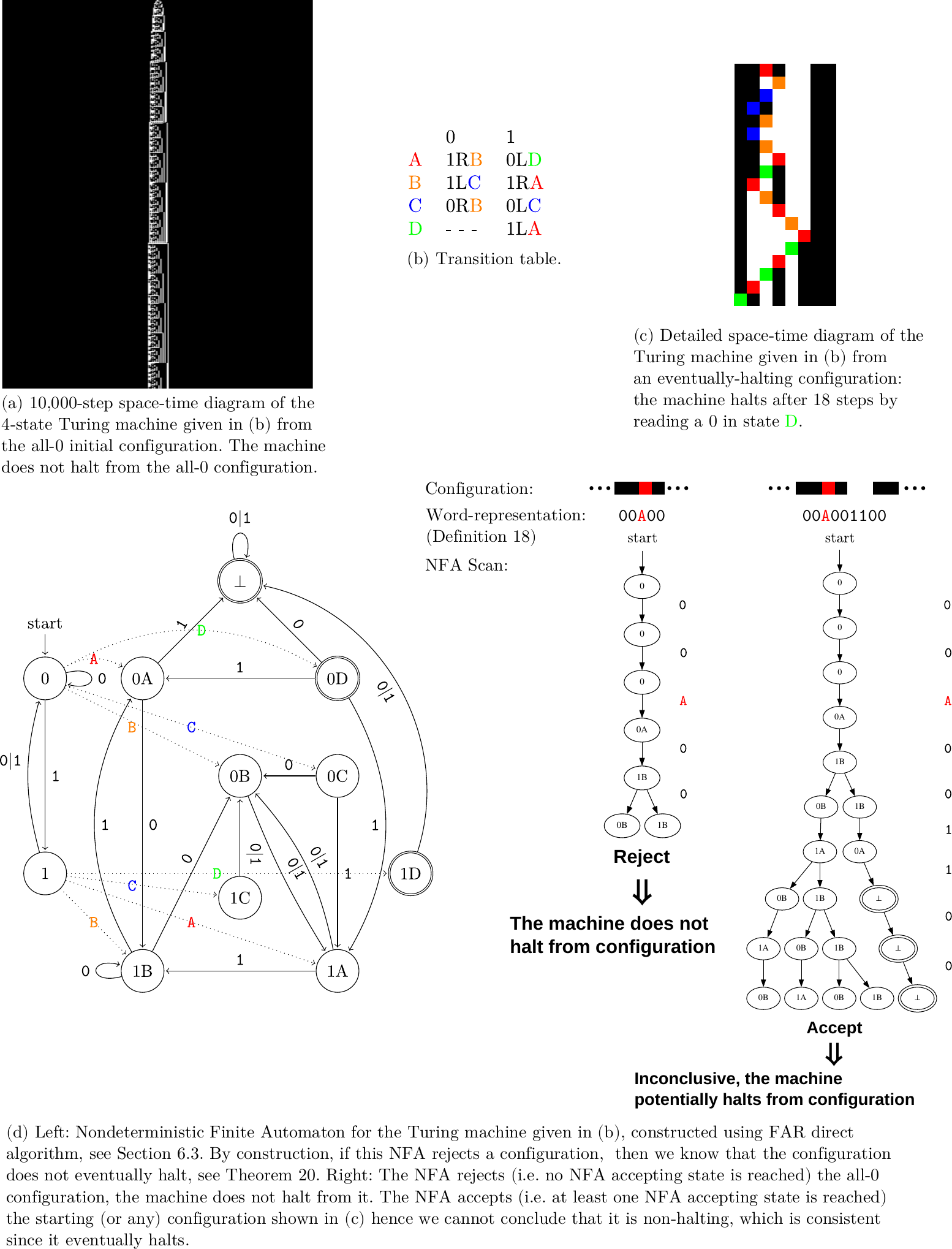}

  \caption{\small A Nondeterministic Finite Automaton, used as follows to decide a 4-state Turing machine\protect\footnotemark:
    (a) Space-time diagram showing the first few descendants of the all-0 configuration for the machine. The machine actually runs forever from the all-0 configuration, adopting a ``counting'' behavior.
    (b) Transition table for the TM.
    (c) The TM halts in 18 steps from a different configuration; these 18 rows depict \emph{eventually-halting} configurations.
    (d) A Nondeterministic Finite Automaton, constructed using the direct FAR algorithm (Section~\ref{far-algo-direct}), that recognises at least all eventually-halting configurations (with finitely many 1s) of the machine. Inputting the top row of (c), encoded as word \texttt{00A001100} (see Definition~\ref{def:wordc}), the NFA transitions by reading each successive symbol of the input, through NFA states: 0, 0, 0A, 1B, \{0B, 1B\}, \{1A, 0A\}, \{0B, 1B, $\bot$\}, \{1A, 0B, 1B, $\bot$\} and finally \{0B, 1A, 0B, 1B, $\bot$\}. Since NFA state $\bot$ is accepting (doubly circled in (d)), the NFA accepts \texttt{00A001100}, classifying this configuration as potentially eventually-halting. However, the NFA does not accept input \texttt{A0}, which corresponds to the all-0 configuration, hence this TM cannot halt from there.}

  \thisfloatpagestyle{empty}
  \label{fig:finite-automata-reduction}
\end{figure}

\footnotetext{\url{https://bbchallenge.org/1RB0LD_1LC1RA_0RB0LC_---1LA}, the machine exhibits a non-trivial counting behavior.}

\subsection{Potential-halt-recognizing automata}
\newcommand{\M}{\mathcal{M}}
\newcommand{\T}{^T}
\newcommand{\row}{\text{row}}
\label{far-defs-recognizer}
For a given Turing machine, we aim at building an NFA that recognises at least all its eventually-halting configurations (with finitely many 1s). In other words, the NFA recognises configurations that \textit{potentially} eventually halt, which is why we call the NFA \textit{potential-halt-recognizing}. Importantly, if the NFA does not recognise the all-0 initial configuration then we know that the Turing machine does not halt from it. Figure~\ref{fig:finite-automata-reduction} gives a potential-halt-recognizing NFA for a 4-state Turing machine, constructed using the results of Section~\ref{far-algo-direct}.

Let's first recall how Nondeterministic Finite Automata (\textbf{NFA}) can be described using linear algebra. Let $\mathbf{2}$ denote the Boolean semiring\footnote{A semiring is a ring without the requirement to have additive inverses, e.g. the set of natural numbers $\N=\{0,1,2\dots\}$ is a semiring.} $\{0,1\}$ with operations $+$ and $\cdot$ respectively implemented by $\operatorname{OR}$ and $\operatorname{AND}$ \cite{CUNINGHAMEGREEN1991251}.
Let $\M_{m,n}$ be the set of matrices with $m$ rows and $n$ columns over $\mathbf{2}$. We may define a Nondeterministic Finite Automaton (NFA) with $n$ states and alphabet $\mathcal{A}$ as a tuple $(q_0, \{T_\gamma\}_{\gamma \in \mathcal{A}}, a)$ where $q_0 \in \M_{1,n}$ and $a \in \M_{1,n}$ respectively represent the initial states and accepting states of the NFA. (i.e. if the $i^\text{th}$ state of the NFA is an initial state then the $i^\text{th}$ entry of $q_0$ is set to 1 and the rest are 0, and the $i^\text{th}$ entry of $a$ is set to 1 if and only if the $i^\text{th}$ state of the NFA is accepting), and where transitions are matrices $T_\gamma\in \M_{n,n}$ for each $\gamma\in\mathcal{A}$ (i.e. the entry $(i,j)$ of matrix $T_\gamma$ is set to 1 iff the NFA transitions from state $i$ to state $j$ when reading $\gamma$). Furthermore, for any word $u=\gamma_1\dots\gamma_\ell \in \mathcal{A}^*$, let $T_u = T_{\gamma_1} T_{\gamma_2} \dots T_{\gamma_\ell}$ be the state transformation resulting from reading word $u$ (Note: $T_\epsilon = I$). A word $u=\gamma_1\dots\gamma_\ell \in \mathcal{A}^*$ is accepted by the NFA iff there exists a path from an initial state to an accepting state that is labelled by the symbols of $u$, which algebraically translates to $q_0 T_u a\T = 1$ with $a\T \in \M_{n,1}$ the transposition of $a$.

\begin{example}\label{ex:nfa}
  The NFA depicted in Figure~\ref{fig:example_nfa}, with states X, Y, Z and alphabet $\mathcal{A}=\{\alpha,\beta\}$, is algebraically encoded as follows: $q_0 = (1,1,0)$, $a=(0,0,1)$, $T_\alpha=\begin{bmatrix}
      0 & 0 & 0 \\
      1 & 1 & 0 \\
      0 & 1 & 1
    \end{bmatrix}$ and $T_\beta= \begin{bmatrix}
      1 & 0 & 1 \\
      1 & 0 & 0 \\
      0 & 0 & 1
    \end{bmatrix}$. The reader can check that words $\beta$, $\alpha\beta$ and $\alpha\alpha\beta\beta$ are accepted, i.e. $q_0T_\beta a\T = 1$, $q_0T_\alpha T_\beta a\T = 1$ and $q_0T_\alpha T_\alpha T_\beta T_\beta a \T = 1$. But, words $\alpha$, $\alpha\alpha$ and $\alpha\alpha\alpha$ are rejected, i.e. $q_0T_\alpha a\T = 0$, $q_0T_\alpha T_\alpha a\T = 0$ and $q_0T_\alpha T_\alpha T_\alpha a \T = 0$.
\end{example}

Now, we describe how we transform Turing machine configurations that have finitely many 1s into finite words that will be read by our NFA. First recall that a Turing machine configuration is defined by the 3-tuple: (i) state in which the machine is (ii) position of the head (iii) content of the memory tape, see Section~\ref{sec:conventions}. Then, a word-representation of a configuration is defined by:

\begin{definition}[Word-representations of a configuration]\label{def:wordc}
  Let $c$ be a Turing machine configuration with finite support, i.e. there are finitely many 1s on the memory tape of the configuration. A word-representation of the configuration $c$ is a word $\hat{c}$ constructed by concatenating (from left to right) the symbols of any finite region of the tape that contains all the 1s, and adding the state (a letter between A and E in the case of 5-state TMs) just before the position of the head.
\end{definition}

\begin{example}
  A word-representation of the configuration on the top row of Figure~\ref{fig:finite-automata-reduction}(c), is $\hat{c} = \texttt{00A001100}$.
\end{example}

Note that two word-representations of the same configuration will only differ in the number of leading and trailing 0s that they have. Hence, if $\mathcal{L}$ is the regular language of the NFA that we wish to construct to recognise the eventually-halting configurations (with finitely many 1s) of a given TM, it is natural that we require the following:
\begin{align*}
  u \in \mathcal{L} \iff 0u \in \mathcal{L} &  & \text{(leading zeros ignored)}
  \\
  u \in \mathcal{L} \iff u0 \in \mathcal{L} &  & \text{(trailing zeros ignored)}
\end{align*}

These are implied by the following, generally stronger, conditions on the transition matrix $T_0 \in \M_{n,n}$:
\begin{align}
  q_0 T_0 & = q_0
  \label{far-cond-leading-0}
  \\
  T_0 a\T & = a\T
  \label{far-cond-trailing-0}
\end{align}

Note that Condition~\ref{far-cond-trailing-0}, $T_0 a\T = a\T$, means that for all accepting states of the NFA, reading a 0 is possible and leads to an accepting state. Indeed, $T_0 a\T$ describes the set of NFA states that reach the set of accepting states $a$ after reading a $0$.

Then, we want our NFA's language $\mathcal{L}$ to include all eventually-halting configurations (with finitely many 1s) of a given Turing machine $\mathcal{M}$.  Inductively, we want that:
\begin{align*}
  c\vdash\bot                                    & \implies \hat{c} \in \mathcal{L}  \\
  (c_1\vdash c_2)\land \hat{c}_2 \in \mathcal{L} & \implies\hat{c}_1 \in \mathcal{L}
\end{align*}

With $c, c_1, c_2$ configurations of the TM (with finite support) and $\hat{c}, \hat{c}_1, \hat{c}_2$ any of their finite word-representations, see Definition~\ref{def:wordc}. Let $f,t \in \{A,B,C,D,E\}$ denote TM states (the ``from'' and ``to'' states in a TM transition), and $r,w,b \in \{0,1\}$ denote bits (a bit ``read'', a bit ``written'', and just a bit), then the above conditions turn into:
\begin{align*}
  \forall u,z\in\{0, 1\}^*: \; ufrz \in \mathcal{L},\;                                                           & \text{if $(f,r) \to \bot$ is a halting transition of $\mathcal{M}$}
  \\
  \forall u,z\in\{0, 1\}^*,\,\forall b \in \{0, 1\}: utbwz \in \mathcal{L} \implies ubfrz \in \mathcal{L},\;     & \text{if $(f,r) \to (t,w,\text{left})$ is a transition of $\mathcal{M}$}
  \\
  \forall u,z\in\{0, 1\}^*,\,\forall b \in \{0, 1\}: u w t z \in \mathcal{L} \implies u f r z \in \mathcal{L},\; & \text{if $(f,r) \to (t,w,\text{right})$ is a transition of $\mathcal{M}$}
\end{align*}

Which algebraically becomes:
{\small
\begin{align*}
  \forall u,z\in\{0, 1\}^*: \; q_0 T_u T_f T_r T_z a\T = 1, \;                                                                           & \text{if $(f,r) \to \bot$ is a halting transition of $\mathcal{M}$}
  \\
  \forall u,z\in\{0, 1\}^*,\,\forall b \in \{0, 1\}: q_0 T_{u} T_t T_b T_w T_{z} a\T = 1 \implies q_0 T_{u} T_b T_f T_r T_{z} a\T = 1,\; & \text{if $(f,r) \to (t,w,\text{left})$ is a transition of $\mathcal{M}$}
  \\
  \forall u,z\in\{0, 1\}^*,\,\forall b \in \{0, 1\}: q_0 T_{u} T_w T_t T_{z} a\T = 1 \implies q_0 T_{u} T_f T_r T_{z} a\T = 1,\;         & \text{if $(f,r) \to (t,w,\text{right})$ is a transition of $\mathcal{M}$}
\end{align*}
}

These conditions are unwieldy. Let's seek stronger (thus still sufficient) conditions which are simpler:

\begin{itemize}

  \item For machine transitions going left/right, simply require $T_t T_b T_w\preceq T_b T_f T_r$ and $T_w T_t\preceq T_f T_r$, respectively with $\preceq$ the following relation on same-size matrices: $M\preceq M'$ if $M_{ij}\le M'_{ij}$ element-wise, that is, if the second matrix has at least the same 1-entries as the first matrix.

  \item To simplify the condition for halting machine transitions: define an \emph{accepted steady state-set} $s$ to be a row vector such that $sa\T = 1$, $s T_0\succeq s$, and $s T_1\succeq s$. Given such $s$, we have that: $\forall q\in\M_{1,n}\; q \succeq s\implies \forall z\in\{0, 1\}^*: qT_{z}a\T = 1$. Assuming that such $s$ exists we can simply require: $\forall u\in\{0, 1\}^*: q_0T_u T_f T_r \succeq s$ which is stronger than $\forall u,z\in\{0, 1\}^*: \; q_0 T_u T_f T_r T_z a\T = 1$ with $(f,r) \to \bot$ a halting transition.



\end{itemize}

Combining the above, we get our main result:

\begin{theorem}
  \label{far-main-theorem}
  Machine $\mathcal{M}$ doesn't halt from the initial all-0 configuration if there is an NFA $(q_0, \{T_\gamma\}, a)$ and row vector $s$ satisfying the below:
  \begin{align}
    \label{far-cond-first}
    q_0 T_0                                & = q_0
                                           &                     & \text{(leading zeros ignored)}
    \tag{\ref{far-cond-leading-0}}
    \\
    T_0a\T                                 & = a\T
                                           &                     & \text{(trailing zeros ignored)}
    \tag{\ref{far-cond-trailing-0}}
    \\
    sa\T                                   & = 1
                                           &                     & \text{($s$ is accepted)}
    \label{far-cond-ass-accepted}
    \\
    sT_0,sT_1                              & \succeq s
                                           &                     & \text{($s$ is a steady state)}
    \label{far-cond-ass-steady}
    \\
    \forall u\in\{0, 1\}^*: q_0T_u T_f T_r & \succeq s
                                           &                     & \text{if $(f,r) \to \bot$ is a halting transition of $\mathcal{M}$}
    \label{far-cond-halt}
    \\
    \forall b\in\{0, 1\}: T_b T_f T_r      & \succeq T_t T_b T_w
                                           &                     & \text{if $(f,r) \to (t,w,\text{left})$ is a transition of $\mathcal{M}$}
    \label{far-cond-left}
    \\
    T_f T_r                                & \succeq T_w T_t
                                           &                     & \text{if $(f,r) \to (t,w,\text{right})$ is a transition of $\mathcal{M}$}
    \label{far-cond-last}
    \\
    q_0 T_A a\T                            & = 0
                                           &                     & \text{(initial configuration rejected)}
    \label{far-cond-reject-start}
  \end{align}
\end{theorem}
\begin{proof}
  Conditions \eqref{far-cond-leading-0}--\eqref{far-cond-last} ensure that the NFA's language includes at least all eventually halting configurations of $\mathcal{M}$. Condition~\eqref{far-cond-reject-start} ensures that the initial all-0 configuration of the machine is rejected, hence not eventually halting. Hence, if conditions \eqref{far-cond-leading-0}--\eqref{far-cond-reject-start} are satisfied, we can conclude that $\mathcal{M}$ does not halt from the initial all-0 configuration.
\end{proof}

\begin{remark}[Verification]\label{far-remark-verification}
  Theorem~\ref{far-main-theorem} has the nice property of being suited for the purpose of \textit{verification}: given a TM, an NFA and a vector $s$, the task of verifying that equations \eqref{far-cond-first}--\eqref{far-cond-reject-start} hold and thus that the TM does not halt, is computationally simple\footnote{Note that although equation~\eqref{far-cond-halt} has a $\forall$ quantifier, the set of NFA states reachable after reading an arbitrary $u \in \{0,1\}^*$ is computable, and we just have to consider one instance of equation~\eqref{far-cond-halt} replacing $q_0 T_u$ per such state.}. Verifiers have been implemented for Theorem~\ref{far-main-theorem}, see Section~\ref{sec:far-implem}.
\end{remark}

Now, we want to design an efficient search algorithm that will, for a given TM, try to find an NFA satisfying Theorem~\ref{far-main-theorem}. For that search to be feasible, we impose more structure on the NFA so that (a) the search space of NFAs is smaller (b) a subset of Conditions \eqref{far-cond-first}--\eqref{far-cond-last} is automatically satisfied by these NFAs.

\subsection{Search algorithm: direct FAR algorithm}
\label{far-algo-direct}

We design an efficient search algorithm for Theorem~\ref{far-main-theorem} that we call the \textit{direct FAR algorithm}. We start by adding more structure to our NFAs as follows:

\begin{enumerate}

  \item The NFA is constructed from two sub-NFAs: one NFA responsible for handling the left-hand side of the tape (i.e. before reading the tape-head state) and one NFA for handling the right-hand side of the tape (i.e. after reading the tape-head state).
  \item The sub-NFA for the left-hand side of the tape is a Deterministic Finite Automaton (DFA).
  \item Edges labelled by a tape-head state are only those that start in the left-hand side DFA and end in the right-hand side NFA. Furthermore, we require that no such two edges reach the same state in the right-hand side NFA. Hence, the right-hand side NFA has at least $5l$ states with $l$ the number of states in the left-hand side DFA.\label{pt:injective}
  \item In fact, we require that the right-hand side NFA has exactly $5l+1$ states with the extra state $\bot$ that we call the \textit{halt state}.

\end{enumerate}

\begin{example}
  The structure described above is followed by the NFA depicted in Figure~\ref{fig:finite-automata-reduction}(d)~Left. Note that, following above Point~\ref{pt:injective}, it is natural to name states in the right-hand side NFA by prepending left-hand side DFA states to the transitions' TM state letter, e.g. state 1C in Figure~\ref{fig:finite-automata-reduction} is reached from DFA state 1 after reading TM state letter C.
\end{example}

This structure might seem arbitrary but it has a very nice property that we demonstrate here: once the left-hand side DFA is chosen, there is at most one right-hand side NFA (minimal for $\succeq$) such that the overall NFA satisfies Theorem~\ref{far-main-theorem}.

Indeed, let's rewrite the above structural points algebraically:

\begin{enumerate}
  \item We write the state space of the NFA as the direct sum $\mathbf{2}^l \oplus \mathbf{2}^d$ with $l$ the number of states of the left-hand side DFA and $d=5l+1$ the number of states of the right-hand side NFA. Initial state is $\begin{bmatrix}q_0&0\end{bmatrix}$ with $q_0 \in \mathbf{2}^l$,
        transitions
        $T_b=\begin{bmatrix}L_b&0\\0&R_b\end{bmatrix}$ ($b\in\{0,1\}$) with $L_b \in \M_{l,l},\, R_b \in \M_{d,d}$ and
        $T_f=\begin{bmatrix}0&M_f\\0&0\end{bmatrix}$ ($f\in\{A,\ldots,E\}$) with $M_f \in \M_{l,d}$,
        and acceptance $\begin{bmatrix}0&a\end{bmatrix}$ with $a \in \M_{1,d}$.
  \item $(q_0,\{L_0, L_1\})$ comes from a DFA with transition function $\delta: [l] \times \{0,1\} \to [l]$ (with $[l]$ the set $\{0,\dots,l-1\}$) that ignores leading zeros, i.e. $\delta(0,0) = 0$. That ensures \eqref{far-cond-leading-0} of Theorem~\ref{far-main-theorem}.
  \item Row vectors of matrices $M_f$ (with $f\in\{A,\ldots,E\}$) are the standard basis row vectors $e_0,\, \dots,\, e_{5l-1} \in \M_{1,d}$ where basis vector $e_i$ has its $i^\text{th}$ entry set to 1 and the other entries set to 0.\label{pt:basis}
  \item The right-hand side NFA has \textit{halt state} $\bot$ and $e_{5l} = e_\bot$ as its corresponding basis row vector.

\end{enumerate}

For a given Turing machine, our direct FAR algorithm will enumerate left-hand side DFAs and for each, find an associated right-hand side NFA by solving Theorem~\ref{far-main-theorem} \eqref{far-cond-first}--\eqref{far-cond-last} for $R_0$, $R_1$, and $a$. If Condition \eqref{far-cond-reject-start} is also satisfied then, by Theorem~\ref{far-main-theorem}, the Turing machine is proven non-halting and we stop the search.

For a given left-hand side DFA with transition function $\delta$, the right-hand side NFA is constructed by rewriting Theorem~\ref{far-main-theorem} conditions~\eqref{far-cond-ass-steady}--\eqref{far-cond-last} in the following way, where we set the accepted steady state-set to $s=\begin{bmatrix}0&e_\bot\end{bmatrix}$. The algebra is helped by the general observation that for any $i$, the condition $\row_i(M) \succeq v$ with $\row_i(M)$ the $i^\text{th}$ row of matrix $M$ and $v$ some row vector, is equivalent to $M\succeq e_i\T v$ with $e_i$ the $i^\text{th}$ standard basis vector\footnote{This is why we asked that row vectors of matrices $M_f$ are standard basis vectors, Point~\ref{pt:basis} above.}.


\begin{align}
  R_r                                        & \succeq (e_\bot)\T e_\bot
                                             &                                               & \text{for }r\in\{0,1\}
  \tag{\ref{far-cond-ass-steady}'}
  \\
  \forall i\in[l]: R_r                       & \succeq \row_i(M_f)\T e_\bot
                                             &                                               & \text{if $(f,r) \to \bot$ is a halting transition of $\mathcal{M}$}
  \tag{\ref{far-cond-halt}'}
  \\
  \forall b \in\{0,1\}, \forall i\in[l]: R_r & \succeq
  \row_{\delta(i,b)}(M_f)\T \row_i(M_t)R_b R_w
                                             &                                               & \text{if $(f,r) \to (t,w,\text{left})$ is a transition of $\mathcal{M}$}
  \tag{\ref{far-cond-left}'}
  \\
  \forall i\in[l]: R_r                       & \succeq \row_i(M_f)\T \row_{\delta(i,w)}(M_t)
                                             &                                               & \text{if $(f,r) \to (t,w,\text{right})$ is a transition of $\mathcal{M}$}
  \tag{\ref{far-cond-last}'}
\end{align}

\begin{lemma}\label{lem:far-unique-min}
  There's a unique minimal solution (w.r.t $\preceq$) to the system  of inequalities (\ref{far-cond-ass-steady}')--(\ref{far-cond-last}') and an effective way to compute it: initialize $R_0$, $R_1$ to zero,
  then set entries to 1 as (\ref{far-cond-ass-steady}'), (\ref{far-cond-halt}') and (\ref{far-cond-last}') demand then iterate (\ref{far-cond-left}') until $R_0$ and $R_1$ stop changing.
\end{lemma}
\begin{proof}

  First notice that (\ref{far-cond-ass-steady}'), (\ref{far-cond-halt}') and (\ref{far-cond-last}') have their right-hand side constant (with respect to $R$) hence they only amount to constant lower bounds for matrices $R_0$ and $R_1$. Then note that, given any lower bound $B_0\preceq R_0$ and $B_1\preceq R_1$ for true solutions of the system, we have  $\row_{\delta(i,b)}(M_f)\T \row_i(M_t)R_b R_w \succeq \row_{\delta(i,b)}(M_f)\T \row_i(M_t)B_b B_w$ by compatibility of $\succeq$ with the performed operations. Hence, iterating (\ref{far-cond-left}') produces an increasing, eventually stationary, sequence of lower bounds for $R_0$ and $R_1$ whose fixed point is solution to the system.
\end{proof}

Now that we have found $R_0$ and $R_1$ we need to find the set of accepting states $\begin{bmatrix}0&a\end{bmatrix}$ with $a \in \M_{1,d}$.
Conditions \eqref{far-cond-trailing-0}, \eqref{far-cond-ass-accepted} of Theorem \ref{far-main-theorem}   translate to:
\begin{align}
  R_0 a\T & = a\T
          &
  \tag{\ref{far-cond-trailing-0}'}
  \\
  a       & \succeq e_{\bot}
          &
  \tag{\ref{far-cond-ass-accepted}'}
\end{align}

Similarly, there is a unique minimal solution (w.r.t $\preceq$) to this system which is found by initially setting $a_0 = e_\bot$ then iterating $a_{k+1} = (R_0 a_{k}\T)\T$ until a fixed point is reached which gives the value of $a$. Indeed, from (\ref{far-cond-ass-steady}'), we see that the sequence $e_\bot\T \preceq  R_0 e_\bot\T\preceq R_0^2 e_\bot\T  \preceq\dots$ is increasing hence it reaches a fixed point, which satisfies (\ref{far-cond-trailing-0}') and (\ref{far-cond-ass-accepted}').

The last condition from Theorem~\ref{far-main-theorem} that we need to satisfy is \eqref{far-cond-reject-start} (rejection of the initial configuration), which translates to:
\begin{align}
  \row_0 (M_A) a\T & =0
                   & \tag{\ref{far-cond-reject-start}'}
\end{align}

By minimality, a solution of (\ref{far-cond-trailing-0}') and (\ref{far-cond-ass-accepted}') will satisfy (\ref{far-cond-reject-start}') if and only if the minimal solution exhibited above does. Hence, we check (\ref{far-cond-reject-start}') for the minimal $a$ that we constructed and there are two cases:

\begin{itemize}
  \item If $a$ satisfies (\ref{far-cond-reject-start}') then we have found an NFA satisfying Theorem~\ref{far-main-theorem} and we can conclude that the Turing machine does not halt from the all-0 initial configuration.
  \item If $a$ does not satisfy (\ref{far-cond-reject-start}') then we cannot conclude and we continue our search for an appropriate left-hand side DFA.
\end{itemize}

This method relies on a way to enumerate DFAs. In Section~\ref{far-defs-dfa} we give an efficient {\sc search-dfa} algorithm for enumerating canonically-represented DFAs. The search space of DFAs is a tree of partial transition functions and we can skip traversing some sub-trees based on a crucial observation. Solutions $R_0$, $R_1$ and $a$ (given by Lemma~\ref{lem:far-unique-min}) for partial DFA transition function $\delta$ are lower bounds of solutions for any $\delta'$ that extends $\delta$. This observation gives that if $a$, constructed from $\delta$, violates (\ref{far-cond-reject-start}') then, any $a'$, constructed from $\delta'$ extending $\delta$, will violate it too. Hence, in that case, descendants of $\delta$ in the DFA search tree can be skipped.
This efficient pruning technique completes the method, shown below as Algorithm \ref{alg:finite-automata-reduction-direct}.

\begin{algorithm}
  \caption{{\sc decider-finite-automata-reduction-direct}}\label{alg:finite-automata-reduction-direct}

  \begin{algorithmic}[1]
    \Procedure{\textbf{bool} {\sc decider-finite-automata-direct}}{\textbf{TM} machine, \textbf{int} n, \textbf{bool} left\_to\_right}
    \If{\textbf{not} left\_to\_right} switch all left-going transitions of the TM to right-going and vice versa
    \EndIf
    \State \textbf{Matrix$\boldsymbol<$bool}, $5*n+1, 5*n+1${}$\boldsymbol>$ $\textrm{R}[2*n+1][2]$ = $[[0,0],\ldots,[0,0]]$
    \State \textbf{ColVector$\boldsymbol<$bool}, $5*n+1${}$\boldsymbol>$ $aT[2*n+1]$ = 0 // $aT$ for transpose as $a$ is row vector in Section~\ref{far-algo-direct}
    \State \(\triangleright\) Basis vector indexing: for $\row_i(M_f)$ use index $5*i+f$, and for $e_\bot$, use index $5*n$.
    \State Initialize R[0] using (\ref{far-cond-ass-steady}') and (\ref{far-cond-halt}')
    \State Initialize aT[0] = $e_\bot\T$
    \Procedure{{\rm CheckResult} {\sc check}}{List$\boldsymbol<$int$\boldsymbol>$ L}
    \State k $\coloneqq$ L.\textbf{length}
    \State R[k], aT[k] = R[k-1], aT[k-1]
    \State Increase R[k] using (\ref{far-cond-last}'), with $(i,w)=\operatorname{divmod}(\textrm{k-1}, 2)$
    \Repeat
    \State Increase R[k] using (\ref{far-cond-left}'), restricted to $2*i+b<\textrm{k}$
    \Until{R[k] stops changing}
    \Repeat
    \State aT[k] = $\textrm{R}[k][0] \cdot \textrm{aT}[k]$
    \Until{aT[k] stops changing}
    \If{$\row_0(M_\textrm{A})\cdot \textrm{aT}[k]\ne 0$}
    \Return SKIP
    \ElsIf{k == 2*n}
    \Return STOP
    \Else\;\Return MORE
    \EndIf
    \EndProcedure
    \State \Return \Call{search-dfa}{check}
    \EndProcedure
  \end{algorithmic}
\end{algorithm}

\subsection{Efficient enumeration of Deterministic Finite Automata}
\label{far-defs-dfa}
The direct FAR algorithm (Section~\ref{far-algo-direct} and Algorithm~\ref{alg:finite-automata-reduction-direct}) relies on a procedure to enumerate Deterministic Finite Automata (DFA). We first recall the formal definition of DFAs then give an efficient algorithm (Algorithm~\ref{alg:search-dfa}) to enumerate them and to prune the search space early based on using Lemma~\ref{lem:far-unique-min} on partial DFA transition functions.

Textbooks define \emph{deterministic} finite automata (on the binary alphabet, with acceptance unspecified) as tuples $(Q, \delta, q_0)$ of: a finite set $Q$ (states), a $q_0\in Q$ (initial state), and $\delta: Q\times\{0, 1\}\to Q$ (transition function).
Though NFAs generalize DFAs, they can be emulated by (exponentially larger) power-set DFAs. \cite{Sipser}

To put this definition in the linear-algebraic framework:
identify $q_0\in Q$ with $0\in [n]\mathrel{\mathop:}=\{0,\ldots,n-1\}$;
represent states $q$ with elementary row vectors $e_q$;
define transition matrices $T_b$ via $e_q T_b = e_{\delta(q, b)}$.

As we did for transition matrices, extend $\delta$ to words: $\delta(q,\epsilon)=q$, $\delta(q,ub)=\delta(\delta(q,u),b)$.

Given a DFA on $[n]$, call its \emph{transition table} the list $(\delta(0,0),\delta(0,1),\ldots,\delta(n-1,0),\delta(n-1,1))$.

Call $\{\delta(q_0,u): u\in\{0,1\}^*\}$ the set of \emph{reachable} states.

When building a larger recognizer,
we expect no benefit from considering DFAs which just relabel others or add unreachable states.
So motivated, we define a canonical form for DFAs:
enumerate the reachable states via breadth-first search from $q_0$,
producing $f:[n]=\mathrel{\mathop:}Q_\textsf{cf}\to Q$.
Explicitly,
$f(0)=q_0$ and $f(k)$ is the first of
$\delta(f(0),0), \delta(f(0),1), \ldots, \delta(f(k-1),0), \delta(f(k-1),1)$ not in $f([k])$,
valid until $f([k])$ is closed under transitions.
This induces $\delta_\textsf{cf}(q,b)\mapsto f^{-1}(f(q), b)$.
(Warning: this definition and terminology aren't standard.)

\begin{lemma}
  \label{far-dfa-canonical form}
  In a DFA with $(Q,q_0)=([n],0)$, the following are equivalent:
  \begin{enumerate}
    \item it's in canonical form ($Q_\textsf{cf}\to Q$ is the identity)
          and ignores leading zeros (equation \eqref{far-cond-leading-0} or $\delta(0,0)=0$);
    \item its transition table includes each of $0,\ldots,n-1$, whose first appearances occur in order,
          and with each $0 < q < n$ appearing before the $2q$ position in the transition table;
    \item the sequence $\{m_k \mathrel{\mathop:}= \max\{\delta(q,b): 2q+b\le k\}\}_{k=0}^{2n-1}$ of cumulative maxima runs from $0$ to $n-1$ in steps of $0$ or $1$,
          with $m_{2q-1}\ge q$ for $0<q<n$.
  \end{enumerate}
\end{lemma}
\begin{proof}
  \begin{description}
    \item[$1\iff 2$:]
          We prove a partial version by induction:
          the DFA ignores leading zeros and $f(q)=q$ for $q\le k$,
          iff $0,\ldots,k$ have ordered first appearances in the transition table
          which precede appearances of any $q>k$
          and occur before the $2k$ position in $\delta$ if $k>0$.
          In case $k=0$, the DFA ignores leading zeros iff $0$ comes first in the transition table by definition.
          (The other conditions are vacuous.)
          In case the claim holds for preceding $k$, $f(k)$ is by definition the first number outside of $f([k])=[k]$ in the transition table---if any---and the inductive step follows.
    \item[$2\iff 3$:]
          If the first appearances of $0,\ldots,n-1$ appear in order, any value at its first index is the largest so far, so $m_k$ takes the same values. The sequence $m_k$ is obviously nondecreasing, so to be gap-free it can only grow in steps of 0 or 1.
          Conversely, if $m_k$ runs from $0$ to $n-1$ in steps of $0$ or $1$, each value $q\in [n]$ must appear in the table at the first index $k$ for which $m_k=q$, and all preceding values in the transition table must be strictly less.

          In case these equivalent conditions are true, that last observation shows that $q$ appears before the $\delta(q,0)$ position iff $m_k$ reaches $q$ by index $k=2q-1$, or equivalently $m_{2q-1}\ge q$.
  \end{description}
\end{proof}

\begin{corollary}
  $\{t_k\}_{k=0}^\ell$ ($\ell<2n$) is a prefix of a canonical, leading-zero-ignoring, $n$-state DFA transition table iff
  $m_k \mathrel{\mathop:}= \max\{t_j\}_{j=0}^k$ runs from $0$ to $m_\ell<n$ in steps of $0$ or $1$, and $m_{2q-1}\ge q$ (for all $2q - 1 \le \ell$).
\end{corollary}
\begin{proof}
  If $\ell=2n-1$, $\{m_k\}$ grows to exactly $n-1$ (since $m_{2(n-1)-1}\ge n-1$), and lemma \ref{far-dfa-canonical form} applies.
  Otherwise, we may extend the sequence with $t_{\ell+1}=\min(m_\ell+1,n-1)$, the same conditions apply.
\end{proof}

So, Algorithm \ref{alg:search-dfa} searches such DFAs incrementally (avoiding partial DFAs already deemed unworkable).

\begin{algorithm}
  \caption{{\sc search-dfa}}\label{alg:search-dfa}

  \begin{algorithmic}[1]
    \State \textbf{enum} CheckResult \{MORE, SKIP, STOP\}
    \Statex
    \Procedure{\textbf{bool} {\sc search-dfa}}{\textbf{int} n, \textbf{function$\boldsymbol<$List$\boldsymbol<$int$\boldsymbol>$}, CheckResult$\boldsymbol>$ check}
    \Require{$\operatorname{check}(t)\ne\textrm{MORE}$ if $t$ is a complete (length-$2n$) table}

    \State \textbf{int} k = 1, t[$2*n$] = $[0,\ldots,0]$, m[$2*n$] = $[0,\ldots,0]$
    \Loop
    \State state = check(length-k prefix of t)
    \If{state == MORE}
    \State \textbf{int} q\_new = m[k-1] + 1
    \State t[k] = (q\_new $<$ \textrm{n} \textbf{and} 2*q\_new-1 == k) ? q\_new : 0
    \ElsIf{state == SKIP}
    \Repeat
    \If{k $\le$ 1}
    \Return false
    \EndIf
    \State k -= 1
    \Until{t[k] $\le$ m[k-1] \textbf{and} t[k] $<$ n-1}
    \State t[k] += 1
    \Else\;\Return true
    \EndIf
    \State m[k] = max(m[k-1], t[k])
    \State k += 1
    \EndLoop
    \EndProcedure

  \end{algorithmic}
\end{algorithm}

\subsection{Generality of the method}
In the preceding sections, we started from the idea of a closed language of word-representations of TM configurations, made a series of simplifying assumptions, and obtained a search algorithm.
This raises a question: if \emph{any} regular language proves a given TM infinite by co-CTL argument\footnote{
  As regular languages' complements are regular, this is the same as a regular CTL (in the strict sense of \S\ref{far-overview}) existing.
},
must a proof of the form used in Theorem~\ref{far-main-theorem}, let alone in \S\ref{far-algo-direct}, exist?

The answer is yes, and we sketch the proof below.
The following definitions and results aren't needed to prove this decider method's soundness---or outside of this subsection at all---but they justify using Remark~\ref{far-remark-verification} to build a universal (regular) CTL verification scheme.
Historically, they were discovered together with Algorithm~\ref{alg:finite-automata-reduction-direct}, and motivated its development.

Any closed language $\mathcal{L}$ classifies the binary words $w\in\{0,1\}^*$ by Nerode congruence: $w\sim_\mathcal{L} w'$ if for every $z\in\{0,1,A,\ldots,E\}^*$, $wz\in\mathcal{L}\iff w'z\in\mathcal{L}$.
We may form a modified version of the Turing machine $\mathcal{M}$, herein called $\mathcal{M}/\sim_\mathcal{L}$, with the following semantics:

A \textit{configuration} of $\mathcal{M}/\sim_\mathcal{L}$ is defined by the 3-tuple: (i) a state of $\mathcal{M}$, (ii) an equivalence class $[w]_{\sim_\mathcal{L}}$ of some $w\in\{0,1\}^*$ representing the (strictly) left-of-head portion of the tape,  (iii) a finite word $w\in\{0,1\}^*$, representing the remainder of the tape.
We additionally define one distinct configuration, named $\bot$, which represents the machine in a halted state.

Note that any finitely supported configuration $c$ of $\mathcal{M}$ maps to a configuration $[c]$ of $\mathcal{M}/\sim_{\mathcal{L}}$, by sending the left-of-head contents to their equivalence class modulo $\sim_\mathcal{L}$.
This is a many-to-one mapping.

The \emph{transitions} of $\mathcal{M}/\sim_{\mathcal{L}}$ are the images of those of $\mathcal{M}$: that is, if $c_1\vdash_\mathcal{M} c_2$, $[c_1]\vdash_{\mathcal{M}/\sim_\mathcal{L}} [c_2]$. Since $c_1\mapsto[c_1]$ is a many-to-one mapping, this definition makes $\mathcal{M}/\sim_\mathcal{L}$ a nondeterministic machine.
In case $c_1\vdash_\mathcal{M}\bot$ (i.e., the $\mathcal{M}$-transition from $c_1$ is undefined), we also define $[c_1]\vdash_{\mathcal{M}/\sim_\mathcal{L}}\bot$.

For any configurations $c_1,c_2$ of $\mathcal{M}$, if $[c_1]\vdash_\mathcal{M}/\sim_\mathcal{L} [c_2]$ and a word-representation of $c_2$ is in $\mathcal{L}$, observe that a word-representation of $c_1$ is also in $\mathcal{L}$.
This follows because the $\mathcal{M}/\sim_\mathcal{L}$-transition must come from some transition $c_1'\vdash_\mathcal{M}c'_2$ of $\mathcal{M}$, where $[c_1]=[c_1']$ and $[c_2]=[c_2']$. Now, $c'_2$ has a word-representation which differs from one of $c_2$ only by substituting a Nerode-congruent prefix, so $c'_2$ also has a word-representation in $\mathcal{L}$. By closure of $\mathcal{L}$, this is true of $c'_1$, and similarly by Nerode congruence this is true of $c_1$.

Similarly, if $[c_1]\vdash_\mathcal{M}/\sim_\mathcal{L}\bot$, $c_1$ has a word-representation in $\mathcal{L}$.

Say that $\mathcal{M}/\sim_\mathcal{L}$ halts from its initial configuration (which is the image of $\mathcal{M}$'s initial configuration) if there exists a sequence of $\mathcal{M}/\sim_\mathcal{L}$-transitions from it to $\bot$.
The point of this is: if $\mathcal{L}$ is a closed language for $\mathcal{M}$, separating its initial configuration from all halting configurations, then that's impossible!
For that would imply a sequence $[c_0]\vdash_{\mathcal{M}/\sim_\mathcal{L}}\dots\vdash_{\mathcal{M}/\sim_\mathcal{L}}[c_n]\vdash_{\mathcal{M}/\sim_\mathcal{L}}\bot$, where $c_0$ is the initial configuration of $\mathcal{M}$ and $\{c_i\}_{i=1}^n$ is a sequence of other configurations.
By the above, this would imply that $c_0$ has a word-representation in $\mathcal{L}$, contrary to assumption that $\mathcal{L}$ provides a CTL proof that $c_0\not\vdash^*_\mathcal{M}\bot$.

We now seek to recover $\mathcal{L}$, or another regular language which leads to a CTL proof, by studying the halting problem of $\mathcal{M}/\sim_\mathcal{L}$.
In fact, the work has been done already: observe that, just as any Turing machine $\mathcal{M}$ is equivalent to a PDA equipped with two stacks (corresponding to the strict left-hand side of the tape and the rest of it), the machine $\mathcal{M}/\sim_\mathcal{L}$ is equivalent to a  standard nondeterministic PDA. (The control-state space of the PDA to is simply $\{0,1\}^*/\sim_\mathcal{L} \times \{A,\ldots,E\}$---which is a finite set by the Myhill-Nerode theorem. A transition of $\mathcal{M}/\sim_\mathcal{L}$ corresponding to a leftward TM transition pushes the written bit onto the stack. A transition of $\mathcal{M}/\sim_\mathcal{L}$ corresponding to a rightward TM transition pops the read bit off the stack.)
The halting problem of a PDA is solved in  \cite{BEM_1997}: the eventually-halting configurations of any PDA are in fact described by a regular language, whose construction corresponds exactly to the procedure of \S\ref{far-algo-direct}.

In summary: we may take the Myhill-Nerode DFA for the original language $\mathcal{L}$, restrict it to the alphabet $\{0,1\}$, apply the construction from \S\ref{far-algo-direct} to obtain an NFA which recognizes precisely the halting configurations of $\mathcal{M}/\sim_\mathcal{L}$, and combine the DFA/NFA to form a recognizer for some closed language $\mathcal{L}'$ for $\mathcal{M}$; that is, it satisfies \eqref{far-cond-first}--\eqref{far-cond-last}. We also know that a language solving the halting problem of $\mathcal{M}/\sim_\mathcal{L}$ rejects its initial configuration, and so \eqref{far-cond-reject-start} is also satisfied and the constructed $\mathcal{L}'$ provides a CTL proof for $\mathcal{M}$.

\subsection{Search algorithm II: meet-in-the-middle DFA}
\label{far-algo-mtim_dfa}
A symmetric recognizer construction has also shown good results.
Again, pass the left half-tape through a DFA with $l$ states.
Imagine a DFA with $d$ states scanning the (strict) right half-tape right-to-left.

\begin{remark}
  Our definitions require a left-to-right scan direction.
  Any NFA $(e_0, \{R_b\})$ can be transposed.
  (Transposing transition matrices reverses the arrows in the diagram, as with graph adjacency matrices.)
  We can shoehorn this into the preceding framework by making an accept state from R's transposed initial state $e_0\T$,
  defining middle transitions $M_{fr}$ for the configuration's head state/bit,
  superposing all states of R to get our $s$ vector,
  and trying to satisfy conditions like
  $e_0 M_{A0} e_0\T = 0$, $M_{fr}=\sum_L e_q\T s$ (for halt rules),
  $L_b M_{fr} \succeq M_{tb} R^T_w$ (for left rules),
  $M_{fr} R^T_b \succeq L_w M_{tb}$ (for right rules).
  What follows is more intuitive.
\end{remark}

\begin{figure}
  \begin{tikzpicture}[shorten >=1pt, shorten <=1pt]]
    \node[state,initial above] (0L) at (-2, 2) {$0_L$};
    \node[state]           (1L) at (-2, -2) {$1_L$};
    \node[state,initial above] (0R) at (10, 0) {$0_R$};
    \node[state]           (1R) at (8, 0) {$1_R$};
    \node[state]           (2R) at (6, 0) {$2_R$};
    \node[state]           (3R) at (4, 0) {$3_R$};
    \node[state]           (4R) at (2, 2) {$4_R$};
    \node[state]           (5R) at (2, -2) {$5_R$};

    \path[->]  (0L)  edge [loop left]        node {$0$} (0L)
    edge                    node [right] {$1$} (1L)
    (1L)  edge [bend left=15]     node [left] {$0|1$} (0L)
    (0R)  edge [loop right]       node {$0$} (0R)
    edge                    node [above] {$1$} (1R)
    (1R)  edge [bend right]       node [below] {$0$} (0R)
    edge                    node [above] {$1$} (2R)
    (2R)  edge [loop above]       node {$0$} (2R)
    edge                    node [above] {$1$} (3R)
    (3R)  edge                    node [above left] {$1$} (5R)
    (4R)  edge [bend right=15]    node [below left] {$1$} (3R)
    (5R)  edge [loop right]       node {$0|1$} (5R)
    ;
    \path[<->]   (3R)  edge                    node [above right] {$0$} (4R);
  \end{tikzpicture}
  \caption{This pair of DFAs can also recognize halting configurations for the TM of figure \ref{fig:finite-automata-reduction}.
    Configurations are classified by their head state, head bit, and the two half-tapes (processed outside-in by the DFAs.)}
  \label{fig:far_mitm_dfa}
\end{figure}
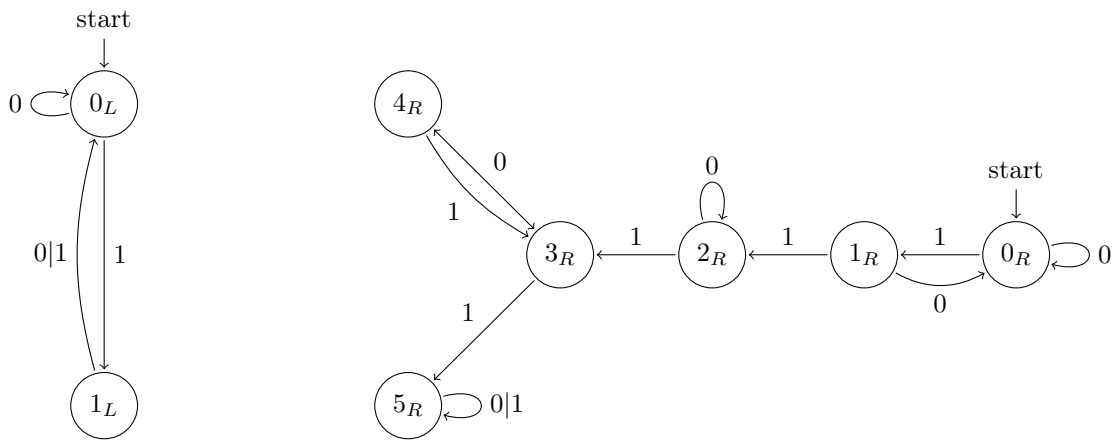

As in Figure \ref{fig:far_mitm_dfa}, let's consider the DFAs on their own terms.
Each one partitions its input into a family of regular languages (one per state).
Accounting for the head state/bit and right half-tape, we obtain $l\cdot 5\cdot 2\cdot d$ classes of TM configuration.
Propose a recognizer which distils this classification into a result.
We'll work out conditions for a good ``accepted'' set $A\subseteq[l]\times\{A,\ldots,E\}\times\{0,1\}\times[d]$.
If they're satisfiable, even if we don't prove the scheme sound, we can feed the left DFA into Algorithm \ref{alg:finite-automata-reduction-direct} to check the result.

Despite the new setting, we can write out closure conditions analogous to \S\ref{far-defs-recognizer}'s, each $\forall i\in[l],j\in[d]$:
\begin{align}
                                                                  & (i,f,r,j)\in A
                                                                  &                               & \text{if $(f,r) \to \bot$ is a halting transition of $\mathcal{M}$}
  \tag{\ref{far-cond-halt}''}
  \\
  \forall b \in \{0,1\},\, (i, t, b, \delta_R(j,w))\in A \implies & (\delta_L(i,b), f, r, j)\in A
                                                                  &                               & \text{if $(f,r) \to (t,w,\text{left})$ is a transition of $\mathcal{M}$}
  \tag{\ref{far-cond-left}''}
  \\
  \forall b \in \{0,1\},\, (\delta_L(i,w), t, b, j)\in A \implies & (i, f, r, \delta_R(j,b))\in A
                                                                  &                               & \text{if $(f,r) \to (t,w,\text{right})$ is a transition of $\mathcal{M}$}
  \tag{\ref{far-cond-last}''}
\end{align}

The goal, analogous to \eqref{far-cond-reject-start}, is $(0, A, 0, 0)\notin A$.

We could now search all DFA pairs, checking if the smallest $A$ closed under (\ref{far-cond-halt}'')--(\ref{far-cond-last}'') rejects $(0,A,0,0)$.
However, to get decent performance, we must express the above as a boolean satisfiability ({\sc sat}) problem.

Other lessons learned in practice:
it was most effective to use the same state count on both sides ($l=d=n$),
and it was decisively faster to impose the canonical form restrictions of Lemma \ref{far-dfa-canonical form}.

Algorithm \ref{alg:finite-automata-reduction-mitm_dfa} shows how this works.
Here especially, actual code can vary from the given pseudocode:
\begin{itemize}
  \item If {\sc sat} solvers use integers for literals (variables and their negations), one needn't ``allocate variables''.
  \item It may be possible to simplify by adding propositional variables for more edge cases.
  \item The ``outcomes are mutually exclusive'' condition may be represented differently.
  \item Checking a solution is valid needn't involve Algorithm \ref{alg:finite-automata-reduction-direct}, if the author proves more.
\end{itemize}

\begin{algorithm}
  \caption{{\sc decider-finite-automata-reduction-MitM-DFA}}\label{alg:finite-automata-reduction-mitm_dfa}

  \begin{algorithmic}[1]
    \Procedure{\textbf{bool} {\sc decider-finite-MitM-DFA}}{\textbf{TM} machine, \textbf{int} n}
    \State \(\triangleright\) Allocate variables.
    \State\textbf{Map$\boldsymbol<$tuple, int$\boldsymbol>$} tk\_eq, tk\_le, mk\_eq, A
    \ForAll{$(\textrm{lr}, \textrm{k}, \textrm{y})\in[2]\times[n]\times[2*n]\times[n+1]$}
    \If{(k, y) == (0, 0)}
    tk\_eq[lr, k, y] = true
    \ElsIf{$0\le y\le\min(k,n-1)$}
    tk\_eq[lr, k, y] = \Call{new-variable}{}
    \Else\;
    tk\_eq[lr, k, y] = false
    \EndIf

    \If{$y\le 0$}
    tk\_le[lr, k, y] = tk\_eq[lr, k, y]
    \ElsIf{$0\le y\le\min(k-1,n-2)$}
    tk\_le[lr, k, y] = \Call{new-variable}{}
    \Else\;
    tk\_le[lr, k, y] = true
    \EndIf

    \If{(k, y) == (2*n-1, n-1)}
    mk\_eq[lr, k, y] = true
    \ElsIf{\textbf{not} $\left\lceil\frac{k}{2}\right\rceil\le y<\min(n,k+1)$}
    mk\_eq[lr, k, y] = false
    \ElsIf{$\min(n, k+1)-((k+1)/2) \le 1$}
    mk\_eq[lr, k, y] = true
    \Else\;
    mk\_eq[lr, k, y] = \Call{new-variable}{}
    \EndIf
    \EndFor

    \ForAll{$(\textrm{i}, \textrm{f}, \textrm{r}, \textrm{j})\in[n]\times[5]\times[2]\times[n]$}
    \If{(k, y) == (2*n-1, n-1)}
    A[i, f, r, j] = false
    \Else\;
    A[i, f, r, j] = \Call{new-variable}{}
    \EndIf
    \EndFor

    \State \(\triangleright\) Transition validity: outcomes are mutually exclusive.
    \ForAll{$(\textrm{lr}, k, \textrm{y})\in[2]\times[2*n]\times[n]$}
    \State \Call{new-clause}{$\textrm{tk\_eq}(\textrm{lr}, k, y)\implies \textrm{tk\_le}(\textrm{lr}, k, y)$}
    \State \Call{new-clause}{$\textrm{tk\_le}(\textrm{lr}, k, y)\implies \textrm{tk\_le}(\textrm{lr}, k, y+1)$}
    \State \Call{new-clause}{$\textrm{tk\_eq}(\textrm{lr}, k, y+1)\implies \neg\textrm{tk\_le}(\textrm{lr}, k, y)$}
    \EndFor
    \State \(\triangleright\) Transition validity: an outcome occurs.
    \ForAll{$(\textrm{lr}, k)\in[2]\times\{1,\ldots,2*n-1\}$}
    \Call{new-clause}{$\bigvee_{y=0}^{\min(k,n-1)} \textrm{tk\_eq}(\textrm{lr}, k, y)$}
    \EndFor

    \State \(\triangleright\) Closure conditions.
    \ForAll{$(i,j,(f,r))\in[n]^2\times$\Call{halt-rules}{machine}}
    \State\Call{new-clause}{$A[i, f, r, j]$}
    \EndFor
    \ForAll{$(i,j,\textrm{ib},\textrm{jw},(f,r,w,L,t))\in[n]^4\times$\Call{left-rules}{machine}}
    \State\Call{new-clause}{$
        \textrm{tk\_eq}[L, i, b, \textrm{ib}]
        \land \textrm{tk\_eq}[R, j, w, \textrm{jw}]
        \land A[i, t, b, \textrm{jw}]
        \implies A[\textrm{ib}, f, r, j]
      $}
    \EndFor
    \ForAll{$(i,j,\textrm{iw},\textrm{jb},(f,r,w,R,t))\in[n]^4\times$\Call{right-rules}{machine}}
    \State\Call{new-clause}{$
        \textrm{tk\_eq}[R, j, b, \textrm{jb}]
        \land \textrm{tk\_eq}[L, i, w, \textrm{iw}]
        \land A[\textrm{iw}, t, b, j]
        \implies A[i, f, r, \textrm{jb}]
      $}
    \EndFor

    \State \(\triangleright\) DFA is in canonical form (Lemma \ref{far-dfa-canonical form}).
    \ForAll{$(\textrm{lr},k)\in[2]\times\{1,\ldots,2*n-1\}$}
    \For{$m=\lfloor k/2\rfloor,\ldots,\min(n, k)$}
    \State\Call{new-clause}{$\textrm{mk\_eq}(\textrm{lr}, k-1, m) \implies \textrm{tk\_le}(\textrm{lr}, k, m+1)$}
    \State\Call{new-clause}{$\textrm{mk\_eq}(\textrm{lr}, k-1, m) \land \textrm{tk\_le}(\textrm{lr}, k, m) \implies \textrm{mk\_eq}(\textrm{lr}, k, m)$}
    \State\Call{new-clause}{$\textrm{mk\_eq}(\textrm{lr}, k-1, m) \land \textrm{tk\_eq}(\textrm{lr}, k, m+1) \implies \textrm{mk\_eq}(\textrm{lr}, k, m+1)$}
    \EndFor
    \EndFor

    \If{\Call{check-sat}{}}
    \State Assert L DFA from the model proves the machine using Algorithm \ref{alg:finite-automata-reduction-direct}.
    \State \Return true
    \Else\;\Return false
    \EndIf
    \EndProcedure
  \end{algorithmic}
\end{algorithm}


\subsection{Implementations}\label{sec:far-implem}

Here are the implementations of the decider that were realised:

\begin{enumerate}
  \item Justin Blanchard's original, optimized Rust implementation: \url{https://github.com/bbchallenge/bbchallenge-deciders/tree/main/decider-finite-automata-reduction}
  \item Tony Guilfoyle's C++ reproduction: \url{https://github.com/TonyGuil/bbchallenge/tree/main/FAR}
  \item Tristan Stérin (cosmo)'s Python reproduction: \url{https://github.com/bbchallenge/bbchallenge-deciders/tree/main/decider-finite-automata-reduction-reproduction}
\end{enumerate}

Verifiers for Theorem~\ref{far-main-theorem} -- i.e. programs that check that a given NFA gives a valid nonhalting proof for a given machine, see Remark~\ref{far-remark-verification} -- have also been given with each of the above deciders and, Nathan Fenner provided one verifier formally verified in Dafny: \url{https://github.com/Nathan-Fenner/busy-beaver-dafny-regex-verifier}.


\newpage
\section{Bouncers}\label{sec:bouncers}

\paragraph{Acknowledgement.}  Sincere thanks to Tony Guilfoyle who initially implemented a decider for bouncers\footnote{See: \url{https://github.com/TonyGuil/bbchallenge/tree/main/Bouncers}.}.
Others have contributed to this method by producing alternative implementations (see Section~\ref{sec:bouncers-implem}) or discussing and writing the formal proof presented here: savask, Iijil, mei, Tristan Stérin (cosmo), Justin Blanchard, Pascal Michel.

\begin{figure}[h!]
    \centering
    \includegraphics*[width=0.9\textwidth]{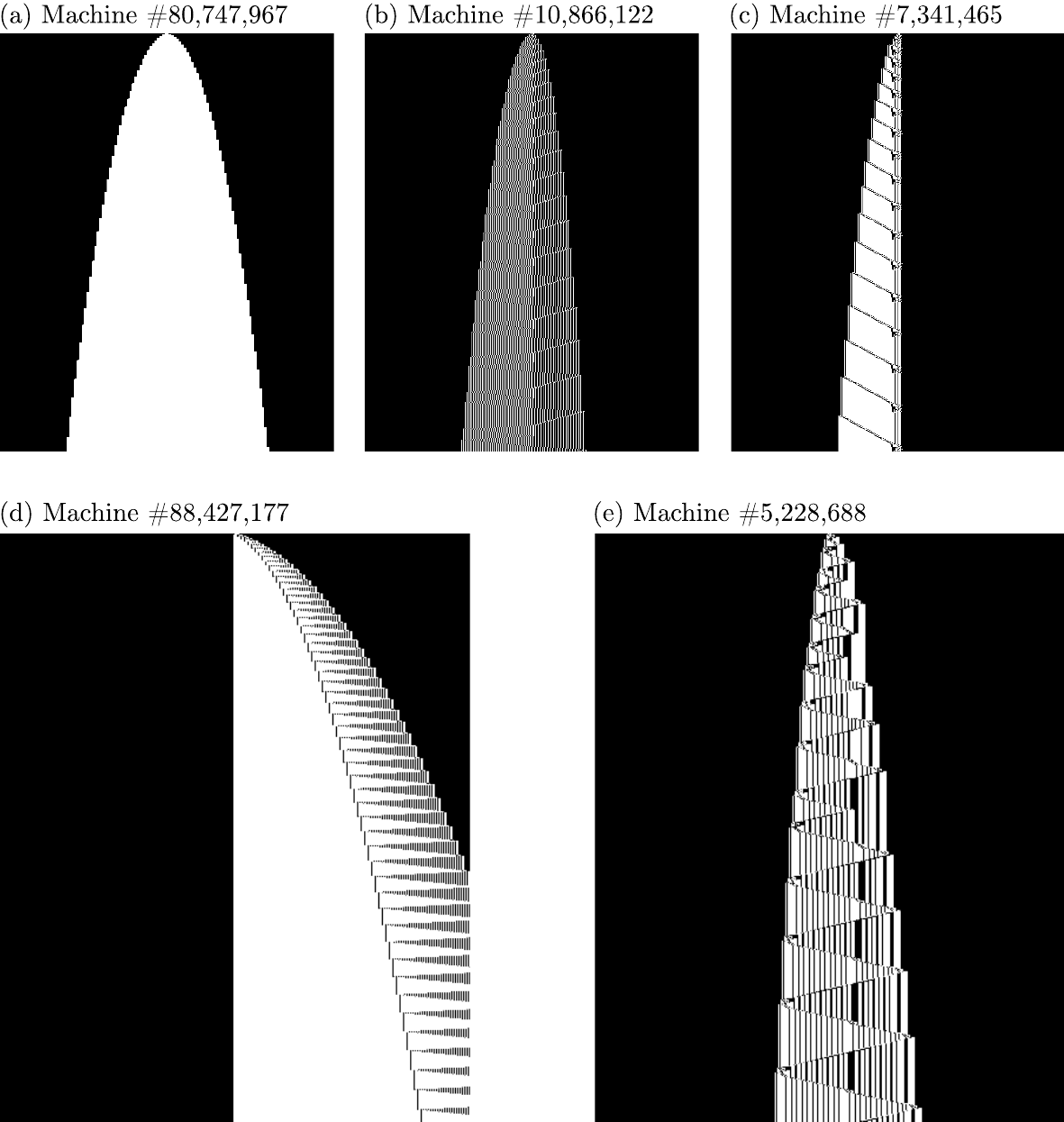}
    \caption{Space-time diagrams (10,000 steps) of several \texttt{bbchallenge} \textit{bouncers}: (a) simple bouncer bouncing back and forth between expanding tape extremeties while writing 1s (b) bouncer with more complex alternating \textit{repeater} patterns left and right of the origin (c) unilateral bouncer with a complex \textit{wall} pattern at the origin (d) unilateral bouncer, main example used throughout this section (e) bouncer entering a repetitive bouncing pattern after $\sim$6,000 steps (bottom half of the image).}\label{fig:bouncers}
\end{figure}

\subsection{Characterising bouncers}

Intuitively, a \emph{bouncer} is a Turing machine that populates a tape with
linearly-expanding patterns, called \textit{repeaters}, possibly separated or enclosed by fixed patterns called \textit{walls}. This intuitive definition corresponds to a wide range of behaviors, from simply bouncing back and forth between the tape's expanding extremities, Figure~\ref{fig:bouncers}~(a), to complex traversal of \textit{repeater} and \textit{wall} patterns, Figure~\ref{fig:bouncers}~(b)-(d), and possibly a delayed onset of the \textit{bouncing} pattern, Figure~\ref{fig:bouncers}~(e). What we call bouncers is a generalisation of the various classes of ``Christmas trees'' used to solve $\text{BB}(4)$ \cite{Brady83}.

The goal of this section is to formally characterise bouncers and show that they do not halt, see Theorem~\ref{th:bouncers}. Then, in Section~\ref{sec:bouncers-implem}, we show how to detect bouncers algorithmically in practice, see Algorithm~\ref{alg:decider-bouncers}.

\subsubsection{Directional Turing machines}\label{sec:bouncers:directionalTM}

We build on the concept of directional Turing machine introduced in Section~\ref{sec:conventions}. Directional Turing machines are an equivalent formulation of Turing machines where the machine head lives in between the tape cells and can point to the left or to the right. We choose here to treat $0^\infty$ as a unique symbol (instead of an infinite collection of 0s) and write $\overline{\Sigma} = \{0^\infty\}\cup\Sigma$, with $\Sigma$ the tape alphabet of the machine -- in the context of $\text{BB}(5)$, we have $\Sigma=\{0,1\}$. For each Turing machine with set of states $S$ we introduce 2$|S|$ new configuration symbols (i.e.\ symbols used to describe machine configurations and not read/write symbols) denoting the machine head in two possible orientations, namely $\Delta = \{\lhead s | s\in S\}\cup\{\rhead s| s \in S\}$.

We define a \textit{tape}\footnote{Note that in the terminology of Section~\ref{sec:conventions}, the definition of a tape that we use here, where the head in part of the tape, corresponds to a (partial) TM configuration.} to be a finite word of the form $uhv$, where $u,v\in \overline{\Sigma}^*$ and $h\in\Delta$, moreover, $u$ and $v$ must have at most one occurrence of $0^\infty$ each, respectively as first symbol for $u$ or last symbol for $v$. We choose the initial tape to be $0^\infty \rhead{\text{A}} 0^\infty$. Now, we define tape rewrite rules which will be used to simulate a directional Turing machine which we fix from now on. Suppose that for $s\in S, x\in \Sigma$ we have $\delta(s,x) = (s',d,x')$ where $s'\in S, d \in \{\text{L},\text{R}\}, x' \in \Sigma$ and $\delta$ the transition function of the machine, then we define the following tape rewrite rules:

\begin{table}[h!]
    \centering
    \begin{tabular}{l|l}
        If $d = \text{L}$                               & If $d = \text{R}$ \\
        \hline
        $\begin{aligned}[t]
                 x \lhead{s}\,\,\, & \to \,\,\, \lhead{s'} x' \\
                 \rhead{s} x\,\,\, & \to \,\,\, \lhead{s'} x'
             \end{aligned}$ & $\begin{aligned}[t]
                                   x \lhead{s}\,\,\, & \to \,\,\, x' \rhead{s'} \\
                                   \rhead{s} x\,\,\, & \to \,\,\, x' \rhead{s'}
                               \end{aligned}$     \\
        \hline
    \end{tabular}\\
    \ \\ If $x=0$, also define: \\
    \begin{tabular}{l|l}
        \hline
        $\begin{aligned}[t]
                 0^\infty \lhead{s}\,\,\, & \to \,\,\, 0^\infty \lhead{s'} x'   \\
                 \rhead{s} 0^\infty\,\,\, & \to \,\,\, \lhead{s'} x'\, 0^\infty
             \end{aligned}$ & $\begin{aligned}[t]
                                   0^\infty \lhead{s}\,\,\, & \to \,\,\, 0^\infty \, x' \rhead{s'} \\
                                   \rhead{s} 0^\infty\,\,\, & \to \,\,\, x' \rhead{s'} 0^\infty
                               \end{aligned}$ \\
        \hline
    \end{tabular}\\
\end{table}

Given a tape $t=uvw$ and a word $t'=uv'w$, with $u,w\in \overline{\Sigma}^*$, $v,v' \in ( \overline \Sigma \cup \Delta)^*$, suppose that there is a rewrite rule $v \to v'$. Then $t'$ is also a tape and in this situation we write $t \vdash t'$ meaning that $t'$ is obtained from $t$ by one simulation step; note that the definition is sound because, to any given tape, at most one rewrite rule applies, hence $v \to v'$ is defined uniquely and so is $t \vdash t'$. Note that the only case where $\vdash$ is not defined on a tape is if the machine halts in this configuration, i.e.\ an undefined transition is reached. Applying $\vdash$ successively $n\in\mathbb{N}$ times is written $\vdash^n$. The transitive closure of $\vdash$ (i.e.\ applying $\vdash$ one or more times) is written $\vdash^+$; applying $\vdash$ zero or more times is written $\vdash^*$.

\begin{example}\label{ex:bouncer88427177}
    Consider \texttt{bbchallenge} machine\footnote{Accessible at \url{https://bbchallenge.org/88427177}} \#88,427,177 (Figure~\ref{fig:bouncers}~(d)), that has the following transition table:
    \[
        \begin{array}{l|ll}
              & 0          & 1          \\
            \hline
            A & 1\text{RB} & 1\text{LE} \\
            B & 1\text{LC} & 1\text{RD} \\
            C & 1\text{LB} & 1\text{RC} \\
            D & 1\text{LA} & 0\text{RD} \\
            E & \mbox{---} & 0LA
        \end{array}
    \]

    Given the above definitions, we will have the following rewrite rules with state A on the left-hand side:
    \begin{align*}
        0 \lhead{A}\,\,\,        & \to\,\,\, 1 \rhead{B}          \\
        \rhead{A} 0\,\,\,        & \to\,\,\, 1 \rhead{B}          \\
        1 \lhead{A}\,\,\,        & \to\,\,\, \lhead{E} 1          \\
        \rhead{A} 1\,\,\,        & \to\,\,\, \lhead{E} 1          \\
        0^\infty \lhead{A}\,\,\, & \to\,\,\, 0^\infty 1 \rhead{B} \\
        \rhead{A} 0^\infty\,\,\, & \to\,\,\, 1 \rhead{B} 0^\infty
    \end{align*}
    The other rewrite rules are those with left-hand side states B, C, D and E. Simulating the machine for 4 steps, starting from initial tape yields:
    $ 0^\infty \rhead{\text{A}} 0^\infty \;\vdash\; 0^\infty \, 1 \rhead{\text{B}} 0^\infty \;\vdash\; 0^\infty \, 1 \lhead{\text{C}} 1\, 0^\infty\ \;\vdash\; 0^\infty \, 1 \rhead{\text{C}} 1\, 0^\infty \;\vdash\; 0^\infty \, 1 1 \rhead{\text{C}} 0^\infty$. Hence, $0^\infty \rhead{\text{A}} 0^\infty \;\vdash^*\; 0^\infty \, 1 1 \rhead{\text{C}} 0^\infty$.
\end{example}

\subsubsection{Wall-repeater formula tapes}\label{sec:bouncers:formula-tapes}

In this section, we formalise the above-stated intuition that bouncers expand a tape by repeating \textit{repeater} patterns between fixed \textit{walls}. Then, we formally define bouncers in Definition~\ref{def:bouncers} and show that they do not halt in Theorem~\ref{th:bouncers}. For this, we are going to express bouncers' tapes using abbreviated regular expressions over the alphabet $\Delta \cup  \overline{\Sigma}$. Given $u\in\Sigma^*$, we represent the regular language $\{u\}^*$ (zero or more repetitions of the word $u$), as $(u)$. Also, we write $\Sigma^+ = \Sigma^* \setminus \{ \varnothing \}$. We define a \textit{wall-repeater formula tape} to be an expression of the form:
\begin{align}\label{math:formulaTapes}w_1(r_1)w_2(r_2)\dots w_n(r_n) w_{n+1} h w'_1(r'_1)w'_2(r'_2)\dots w'_m(r'_m) w'_{m+1}\end{align}

if the following conditions are met:
\begin{align*}
     & n,m \geq 0,                                     \\
     & h \in \Delta,                                   \\
     & r_1,\dots,r_n,r'_1,\dots,r'_m \in \Sigma^+,     \\
     & w_2,\dots,w_{n+1},w'_1,\dots,w'_m \in \Sigma^*, \\
     & w_1, w'_{m+1} \in  \overline{\Sigma}^*
\end{align*}

and, as in the definition of a tape, $w_1$ (resp. $w'_{m+1}$) either does not contain the symbol $0^\infty$ or it starts with it (resp. ends with it). Words $w_i, w'_j$ are called \textit{walls} and can be empty, while the \textit{repeaters}, $r_i, r'_j$ must be nonempty. Note that we allow $n=m=0$, hence a usual tape is also a wall-repeater formula tape. For the rest of this section, we abbreviate wall-repeater formula tapes as \textit{formula tapes}.

Given a formula tape $f$ let $\mathcal{L}(f)$ denote the language described by $f$, i.e.\ the set of tapes that match~it.

\begin{example}\label{ex:formulaTapes}
    Consider the formula tape $f = 0 \rhead{\text{D}} (01)$. We have $\mathcal{L}(f) = \{0\rhead{\text{D}},0\rhead{\text{D}}01,0\rhead{\text{D}}0101,\dots\}$. Consider the formula tape $f'=0^\infty(111)1110 \lhead{\text{A}} 010101(01)10^\infty$. Then: $0^\infty 1110 \lhead{A} 01010110^\infty$, $0^\infty 1111110 \lhead{A} 01010110^\infty$, and $0^\infty 1110 \lhead{A} 0101010110^\infty$ are elements of $\mathcal{L}(f')$.
\end{example}



\paragraph*{Extending $\vdash$ to formula tapes: \textit{shift rules}.} We wish to extend the Turing machine step relation $\vdash$ (see Section~\ref{sec:bouncers:directionalTM}) to formula tapes. This is quite straightforward when the head is pointing at a symbol of a wall (one of the $w_i, w'_j$ in \eqref{math:formulaTapes}): we simply apply a standard Turing machine step, leaving the definition of $\vdash$ unchanged.

However, we need to handle the case where the head is pointing at a repeater (one of the $r_i, r'_j$ in \eqref{math:formulaTapes}). Suppose that for some $u\in\Sigma^*$ and $r,\tilde{r}\in\Sigma^+$ and some state $s\in S$ we have $u \rhead{s} r \vdash^+ \tilde{r} u \rhead{s}$. Then, for any $n\geq 0$, we have $u \rhead{s} r^n \vdash^* \tilde{r}^n u \rhead{s}$. This motivates the definition of
(right) \textit{shift rules}, that is, rewrite rules for formula tapes, which rewrite a subword $u \rhead{s}(r)$ into $(\tilde{r})u\rhead{s}$, denoted by $u \rhead{s}(r) \to (\tilde{r})u\rhead{s}$. Similarly, left shift rules are of the form $(r)\lhead{s}u \to \;\lhead{s}u(\tilde{r})$ given that $r\lhead{s}u \vdash^+ \;\lhead{s}u\tilde{r}$. Note that repeaters $r$ and $\tilde{r}$ of a shift rule necessarily have the same size.

Hence, we have two cases to consider for defining $\vdash$ on the following formula tape $f$ (as defined in~\eqref{math:formulaTapes}):
$$f = w_1(r_1)w_2(r_2)\dots w_n(r_n) w_{n+1} h w'_1(r'_1)w'_2(r'_2)\dots w'_m(r'_m) w'_{m+1}$$

\begin{enumerate}
    \item (Usual step) If $h$ points at nonempty $w_{n+1}$ or $w'_1$, and $w_{n+1} h w'_1 \vdash \tilde{w}_{n+1} \tilde{h} \tilde{w}'_1$ as tapes, replace the terms $w_{n+1} h w'_1(r'_1)$ of $f$ with $\tilde{w}_{n+1} \tilde{h} \tilde{w}'_1$. Call the new formula $f'$, and define $f \vdash f'$.
          As for tapes, $\vdash$ is undefined if $w_{n+1} h w'_1$ corresponds to a halting configuration (i.e.\ undefined transition).
    \item (Shift rule) Two cases:
          \begin{enumerate}

              \item Right shift rule. If $h=\;\rhead{s}$ with $s\in S$ and $w'_1$ is empty, consider the set of shift rules $\mathcal{R} = \{ u \rhead{s}(r'_1) \to (\tilde{r})u\rhead{s} \; | \; \tilde{r}\in\Sigma^+,\; u\in\Sigma^* \text{ is a suffix of } w_{n+1}\}$. If $\mathcal{R}$ is not empty then apply the right shift rule of $\mathcal{R}$ with smallest $u$ (possibly empty), call the new formula $f'$, and we define $f \vdash f'$. If $\mathcal{R}$ is empty, $\vdash$ cannot be applied to $f$.
              \item Left shift rule. If $h=\;\lhead{s}$ with $s\in S$ and $w_{n+1}$ is empty, consider the set of shift rules $\mathcal{R} = \{ (r_{n})\lhead{s}u \to \;\lhead{s}u(\tilde{r}) \; | \; \tilde{r}\in\Sigma^+,\; u\in\Sigma^* \text{ is a prefix of } w'_{1}\}$. If $\mathcal{R}$ is not empty then apply the left shift rule of $\mathcal{R}$ with smallest $u$ (possibly empty), call the new formula $f'$, and we have $f \vdash f'$. If $\mathcal{R}$ is empty, $\vdash$ cannot be applied to $f$.

          \end{enumerate}

\end{enumerate}

From the above definition of $\vdash$ on a formula tape $f$, it is clear that (i) for usual tapes our new definition of $\vdash$ coincides with the old one, (ii) there is at most one formula tape $f'$ such that $f \vdash f'$, and (iii) the only cases where $\vdash$ is not defined on a formula tape are when the machine halts (usual step case) or no shift rule applies (shift rule case). Moreover, we get the following result:

\begin{lemma}\label{lem:vdashFormulaTapes} Let $f$ and $f'$ be formula tapes with $f \vdash f'$. Then, for all $t \in \mathcal{L}(f)$, there exists some $t' \in \mathcal{L}(f')$ such that $t \vdash^* t'$
    and, in case $f \vdash f'$ via usual step, $t \vdash^+ t'$.
\end{lemma}
\begin{proof}
    If $f'$ follows from $f$ by a usual step, then applying one usual step to $t$ yields $t'\in\mathcal{L}(f')$. If $f'$ follows from $f$ by a shift rule (of the form $u \rhead{s}(r) \vdash^k (\tilde{r})u\rhead{s}$ or $(r)\lhead{s}u \vdash^k \;\lhead{s}u(\tilde{r})$), and the corresponding repeater $(r)$ is used $n$ times in the match $t\ \in \mathcal{L}(f)$, we apply $nk$ steps to $t$ to obtain $t' \in \mathcal{L}(f')$. This amounts to a positive number of steps (justifying $\vdash^+$)  except in case of a shift rule applied $k=0$ times.
\end{proof}


\begin{example}\label{ex:shiftRules}
    Taking the machine of Example~\ref{ex:bouncer88427177}, we have the right shift rule $0 \rhead{\text{D}}(01) \to (11)0\rhead{\text{D}}$. Indeed, this is because: $0 \rhead{\text{D}}01 \vdash 0\lhead{\text{A}}11 \vdash 1\rhead{\text{B}}11 \vdash 11\rhead{\text{D}}1 \vdash 110\rhead{\text{D}}$, hence $0 \rhead{\text{D}}01 \vdash^* 110\rhead{\text{D}}$, giving the shift rule. Consider the tape formula of previous Example~\ref{ex:formulaTapes}: $f' = 0^\infty(111)1110 \lhead{\text{A}} 010101(01)10^\infty$. We have $f' \vdash^{13} 0^\infty (111)1111110110 \rhead{\text{D}} (01) 10^\infty$, at this point the head points at a repeater and the set of applicable right shift rules is $\mathcal{R} = \{ 0 \rhead{\text{D}}(01) \to (11)0\rhead{\text{D}},\; 10 \rhead{\text{D}}(01) \to (11)10\rhead{\text{D}},\; 110 \rhead{\text{D}}(01) \to (11)110\rhead{\text{D}}\}$, and following the definition of $\vdash$, we apply $0 \rhead{\text{D}}(01) \to (11)0\rhead{\text{D}}$ as it has the smallest left-hand side, giving: $0^\infty (111)111111011(11)0 \rhead{\text{D}} 10^\infty$.

\end{example}

\paragraph*{Aligning formula tapes.} One last tool that we need before characterising bouncers and proving that they do not halt (Theorem~\ref{th:bouncers}) is formula tape \textit{alignment} (Definition~\ref{def:alignment}): sometimes it is necessary to rewrite a formula tape in an equivalent, \textit{aligned} form in order for any shift rules to apply.

\begin{definition}[Alignment operator]\label{def:alignment}
    Take a formula tape, as given in~\eqref{math:formulaTapes}: $$f = w_1(r_1)w_2(r_2)\dots w_n(r_n) w_{n+1} h w'_1(r'_1)w'_2(r'_2)\dots w'_m(r'_m) w'_{m+1}$$

    The alignment operator $f \mapsto \mathcal{A}(f)$ moves repeaters away from the head $h$ by repeatedly applying any of the following rules until none apply anymore:
    \begin{enumerate}
        \item Replace $(r'_{j})v$ with $v(r)$ in $f$, if $r'_j v = v r$ with $v$ a nonempty prefix of $w'_{j+1}$, $r\in\Sigma^+$, and $1 \leq j \leq m$.

        \item Replace $v(r_{i})$ with $(r)v$ in $f$, if $v r_i = r v$ with $v$ a nonempty suffix of $w_{i}$, $r\in\Sigma^+$, and $1 \leq i \leq n$.
    \end{enumerate}

    Clearly, the order application of these rules does not matter, i.e.\ $\mathcal{A}$ is well-defined, and $\mathcal{A}(\mathcal{A}(f)) = \mathcal{A}(f)$.
\end{definition}

\begin{lemma}\label{lem:sameLanguage} For any formula tape $f$, $\mathcal{L}(\mathcal{A}(f)) = \mathcal{L}(f)$, i.e.\ both $f$ and $\mathcal{A}(f)$ represent the same set of tapes.
\end{lemma}

\begin{proof}
    Consider an alignment rule as in case (1) of Definition~\ref{def:alignment}: replacing $(r'_{j})v$ with $v(r)$ if $r'_j v = v r$. Then, for all $n\in\mathbb{N},\; {r'_j}^n v = v r^n$, hence $(r'_{j})v$ and $v(r)$ describe the same language: $\mathcal{L}((r'_{j})v) = \mathcal{L}(v(r))$. Same for case (2) and for multiple applications of (1) and (2) in any order, hence we have $\mathcal{L}(\mathcal{A}(f)) = \mathcal{L}(f)$.
\end{proof}

\begin{example}\label{ex:alignment}
    Take $f=0^\infty(111)1111\rhead{\text{B}}(01)01010110^\infty$ for the machine given in Example~\ref{ex:bouncer88427177}. One can verify that no right shift rule applies to $f$ hence $\vdash$ does not apply to $f$. However, we have $ \mathcal{A}(f)=0^\infty(111)1111\rhead{B}010101(01)10^\infty$, $\mathcal{L}(\mathcal{A}(f)) = \mathcal{L}(f)$, and we can apply $\vdash$ to $ \mathcal{A}(f)$ by performing usual steps: $ \mathcal{A}(f) \vdash^{12} 0^\infty(111)1111110110 \rhead{D} (01)10^\infty$ and now the shift rule of Example~\ref{ex:shiftRules} can apply, giving $0^\infty(111)111111011(11)0 \rhead{\text{D}} 10^\infty$.
    Another alignment example is $f=0^\infty 101(11)0 \rhead{\text{D}} 10 (100) 11 0^\infty$ for which $\mathcal{A}(f) = 0^\infty 10(11)10 \rhead{\text{D}} 101(001)10^\infty$.
\end{example}

The above example shows that alignment (which preserves the set of recognised tapes) can allow to run a shift rule on $f'$ with $\mathcal{A}(f) \vdash^* f'$ when no shift rule was applicable directly to $f$. In fact, one can show that the alignment operator can only \textit{increase} the number of applicable shift rules: if a shift rule is applicable to $f$ then a shift rule applicable to $f'$ with $\mathcal{A}(f) \vdash^* f'$ can always be found. This motivates the introduction of the notation $f \vdash_\mathcal{A} f'$ which means that we have $\mathcal{A}(f) \vdash f'$.

Given two formula tapes $f$ and $f'$ we will say that $f'$ is a \textit{special case} of $f$, if $\mathcal{A}(f')$ can be obtained from $\mathcal{A}(f)$ by replacing subwords of the form $(r)$ in $\mathcal{A}(f)$ by $r^n(r)r^m$ for some $n,m\geq 0$ and $r\in\Sigma^+$. Note that because of alignment, $n=0$ (resp. $m=0$) if $(r)$ in $\mathcal{A}(f)$ is to the left (resp. to the right) of the head. If $f'$ is a special case of $f$ then $\mathcal{L}(f') \subseteq \mathcal{L}(f)$, and the authors conjecture that the converse is true under some mild additional assumptions\footnote{See \url{https://discuss.bbchallenge.org/t/186}.}.

We finally get to the main result of this section, which characterises bouncers formally -- a bouncer is any machine to which Theorem~\ref{th:bouncers} applies:

\begin{definition}[Bouncers]\label{def:bouncers}
    A bouncer is a Turing machine such that there exists a wall-repeater formula tape $f$ which satisfies:
    \begin{enumerate}
        \item There is a reachable tape $t\in\mathcal{L}(f)$, i.e.\ $0^\infty \rhead{\text{A}} 0^\infty \vdash^* t$
        \item There is a formula tape $f'$ such that $f \vdash_\mathcal{A}^+ f'$ and $f'$ is a special case of $f$
    \end{enumerate}
    We say that formula tape $f$ \textit{solves} the bouncer.
\end{definition}

\begin{theorem}\label{th:bouncers}
    A bouncer does not halt.
\end{theorem}

\begin{proof}
    Unraveling $f \vdash_\mathcal{A}^+ f'$ gives $n>0$ and $f=f_0, \dots, f_n$ with $f_n = f'$ and $\mathcal{A}(f_i) \vdash f_{i+1}$ for $0 \leq i < n$.
    It follows from Lemma~\ref{lem:vdashFormulaTapes} and Lemma~\ref{lem:sameLanguage} that there exist tapes $t_0, \dots, t_n$, such that $t_0 = t$ and $t_i\in \mathcal{L}(f_i)$ for $0 \leq i \leq n$, and $t_i \vdash^* t_{i+1}$ for $0 \leq i < n$, giving $t_0 \vdash^* t_n$. Moreover, $t_0 \vdash^+ t_n$ if the any of the $\mathcal{A}(f_i) \vdash f_{i+1}$ relations are via usual steps.

    Indeed, this must be the case. Suppose instead we had a sequence of shift rules. Each one would preserve the direction of the head, and increment the number of repeaters behind the head in the formula tape. However, these properties are unchanged by passing from a formula tape $f$ to $\mathcal{A}(f)$ or a special case of $f$. It would follow that $f_0$ has strictly more repeaters behind the head than $f_0$, a contradiction.

    Since $f_n$ is a special case of $f_0$, we have $\mathcal{L}(\mathcal{A}(f_n)) \subseteq \mathcal{L}(\mathcal{A}(f_0))$, and using Lemma~\ref{lem:sameLanguage}, we have $\mathcal{L}(f_n) \subseteq \mathcal{L}(f_0)$ and thus, $t_n \in \mathcal{L}(f_0)$. We can repeat this construction indefinitely and yield an infinite sequence of tapes $(t_n)_{n\in\mathbb{N}}$ such that $t \vdash^+ t_i$ for all $i\in\mathbb{N}$, hence the machine does not halt.
\end{proof}

\begin{definition}[Bouncer certificate]\label{def:bouncer-certificate}
    For a given bouncer, a bouncer certificate consists at least of the following:
    \begin{enumerate}
        \item The formula tape $f$ of Definition~\ref{def:bouncers} that solves the bouncer.
        \item The time step at which tape $t$ such that $t\in\mathcal{L}(f_0)$ is reached from  $0^\infty \rhead{\text{A}} 0^\infty$.
        \item The number of \textit{macro steps} $n$ such that $f \vdash_\mathcal{A}^n f'$ with $f'$ special case of $f$, where a macro step on a formula tape consists in applying $\vdash_\mathcal{A}$, i.e.\ alignment followed by $\vdash$ (one usual step or one shift rule step).
    \end{enumerate}
    Given such certificate, one can verify that the machine is a bouncer by applying Theorem~\ref{th:bouncers} or that the certificate is erroneous. Additional information, such as the list of used shift rules can be added to the certificate for convenience.
\end{definition}

\begin{example}\label{ex:bouncerTheory}
    The machine of Example~\ref{ex:bouncer88427177} (Figure~\ref{fig:bouncers}~(d)) that is used in our series of examples for this section, is a bouncer. Indeed, we have $0^\infty \rhead{\text{A}} 0^\infty \vdash^{64} 0^\infty 11111101100 \rhead{\text{D}} 0^\infty$. The tape $t=0^\infty \vdash^{64} 0^\infty 11111101100 \rhead{\text{D}} 0^\infty$ is in the language the following formula tape:
    $$f_0 = 0^\infty (111)1110(11)00\rhead{\text{D}}0^\infty$$
    At this point, and for the next 25 usual steps alignment does not affect the formulas, and we get: $$f_{25} = 0^\infty (111) 1110 (11) \lhead{\text{A}} 01010110^\infty$$
    One shift rule gives:
    $$ f_{26} = 0^\infty (111) 1110  \lhead{\text{A}} (01) 01010110^\infty$$
    After alignment:
    $$ \mathcal{A}(f_{26}) = 0^\infty (111) 1110  \lhead{\text{A}} 010101(01)10^\infty$$
    From there, $\mathcal{A}(f_{26}) \vdash f_{27}$ with:
    $$ f_{27} = 0^\infty (111) 1111  \rhead{\text{B}} 010101(01)10^\infty$$
    After 12 usual steps, not affected by alignment, we arrive at:
    $$f_{39} = 0^\infty (111)1111110110\rhead{\text{D}}(01)10^\infty$$
    One shift rule gives:
    $$f_{40} = 0^\infty (111)111111011(11)0\rhead{\text{D}}10^\infty$$
    Aligning gives:
    $$\mathcal{A}(f_{40}) = 0^\infty (111)1111110(11)110\rhead{\text{D}}10^\infty$$
    Finally, from there $\mathcal{A}(f_{40}) \vdash f_{41}$ with one usual step:
    $$f_{41} =\mathcal{A}(f_{41}) = 0^\infty (111)1111110(11)1100\rhead{\text{D}}0^\infty$$

    Now, $f_{41}$ is a special case of $f_{0}$ because $\mathcal{A}(f_{41})$ differs from $f_0=\mathcal{A}(f_0)$ only by including one repetition of repeaters $(111)$ and $(11)$ in the walls directly to their right. The assumptions of Theorem~\ref{th:bouncers} hold and our Turing machine is a bouncer: it does not halt.

    A bouncer certificate (Definition~\ref{def:bouncer-certificate}) for this machine is $f_0 = 0^\infty (111)1110(11)00\rhead{\text{D}}0^\infty$, time step $64$ at which $0^\infty 11111101100 \rhead{\text{D}} 0^\infty \in \mathcal{L}(f_0)$ is reached and $41$ macro steps which transform $f_0$ into special case $f_{41}$ under successive $\vdash$ and alignement applications.
\end{example}

Note that Cyclers (Section~\ref{sec:cyclers}) and Translated Cyclers (Section~\ref{sec:translated-cyclers}) are special cases of bouncers, as they satisfy Theorem~\ref{th:bouncers}.

\subsubsection{Linear-quadratic growth}

For a given formula tape, using the notation of Equation~\ref{math:formulaTapes}: $$f=w_1(r_1)w_2(r_2)\dots w_n(r_n) w_{n+1} h w'_1(r'_1)w'_2(r'_2)\dots w'_m(r'_m) w'_{m+1}$$ call $C_f(k) \in \mathcal{L}(f)$ the tape where each repeater is used exactly $k \in \N$ times: $$C_f(k)=w_1 r_1^k w_2 r_2^k\dots w_n r_n^k w_{n+1} h w'_1 {r'_1}^k w'_2 {r'_2}^k \dots w'_m {r'_m}^k w'_{m+1}$$

We denote an infinite sequence $u_0, u_1, \dots u_n, \dots$ with $n\in\N$ by $(u_n)_{n\in\N}$. We define the difference operator $D: \Z^\N \to \Z^\N$ via $D((u_n)_{n\in\N})=(u_{n+1}-u_{n})_{n\in\N}$. The \textit{length} of a tape $t$ is the number of symbols of $\Sigma$ in $t$, i.e.\ ignoring head and $0^\infty$.

\newcommand{\sync}[1]{\operatorname{sync}(#1)}

Theorem~\ref{th:bouncers} characterises bouncers but we lack a practical criterion for recognising a bouncer from its sequence of successive tapes starting from $0^\infty \rhead{\text{A}} 0^\infty$. We make a step in that direction by showing that if a bouncer solved by formula tape $f$ that is not a cycler (Section~\ref{sec:cyclers}), then we can construct from $f$ a new formula tape $\sync{f}$, that also solves the bouncer but such that tapes $C_{\sync{f}}(k)$, where all repeaters are synchronously repeated $k$ times, are reached by the machine for all successive $k\in\mathbb{N}$. Importantly for being detected in practice, the length of tapes $C_{\sync{f}}(k)$ grows linearly in quadratic time:

\begin{theorem}[Linear-quadratic growth]\label{th:linquad}
    Let $M$ be a bouncer that is not a cycler, solved by some formula tape $f$. Call $(t_n)_{n\in\N}$ the sequence of tapes that it visits starting from $t_0 = 0^\infty \rhead{\text{A}} 0^\infty$. Call $l_n$ the length of $t_n$. Then, there is a formula tape called $\sync{f}$, obtained from $f$, that also solves the bouncer such that there is an extraction function $g: \mathbb{N} \to \mathbb{N}$ with $g(n) \geq n$ for all $n\in\N$, satisfying:

    \begin{enumerate}
        \item For all $n\in\N$, $t_{g(n)} = C_{\sync{f}}(n)$.
        \item $(l_{g(n)})_{n\in\N}$ is an arithmetic progression, i.e.\ $D((l_{g(n)})_{n\in\N}) = (K)_{n\in\N}$ for some constant $K\in\N$.
        \item $(g(n))_{n\in\N}$ is a quadratic progression, i.e.\ $D(D((g(n))_{n\in\N})) = (K')_{n\in\N}$ for some constant $K'\in\N$.
    \end{enumerate}
\end{theorem}

\begin{proof}
    For simplicity of notations, we suppose that there are no repeaters after the tape head in $f$, as we would apply exactly the same transformations after the head. Using Equation~\eqref{math:formulaTapes}, we write $f=w_1(r_1)w_2(r_2)\dots w_n(r_n) w_{n+1} h w$ with $n \geq 1$, $w_i \in \Sigma^*$ for $2 \leq i \leq n+1$, $w_1,w \in \overline{\Sigma}^*$ and $r_i \in \Sigma^+$ for $1 \leq i \leq n$ and $h\in\Delta$. Without loss of generality we suppose that $f$ is aligned (otherwise we just apply $\mathcal{A}$ to it).

    Because $f$ solves the bouncer, by Definition~\ref{def:bouncers}, the machine reaches tape $t$ of the form $t = w_1 r_1^{k_1} w_2 r_2^{k_2} \dots r_n^{k_n} w_{n+1} h w$ with $k_1, \dots, k_n \in \N$. By Theorem~\ref{th:bouncers}, we have $t \vdash^+ t'$ with $t' \in \mathcal{L}(f')$ and $t' = w_1 r_1^{k_1 + p_1} w_2 r_2^{k_2 + p_2} \dots r_n^{k_n + p_n} w_{n+1} h w$ with $p_1, \dots, p_n \in \N$ and same $w_i$ and $r_i$ as $f$ because $f$ is aligned. Note that $p_i$ is positive because shift rules preserve the size of repeaters. Because $t'\in\mathcal{L}(f') \subseteq \mathcal{L}(f)$ we can repeat this indefinitely and get $t \vdash^+ w_1 r_1^{k_1 + Np_1} w_2 r_2^{k_2 + Np_2} \dots r_n^{k_n + Np_n} w_{n+1} h w$ for all $N \in \N$.

    Hence, define $\sync{f} = w_1 r_1^{k_1} (r^{p_1}) w_2 r_2^{k_2} (r^{p_2}) \dots w_n r_n^{k_n} (r^{p_n}) w_{n+1} h w$. When $p_i=0$, then $r^{p_i} = \varnothing$ and we discard it. All $p_i$ cannot be $0$, otherwise the machine is a cycler which is excluded. Altogether, we can write: $\sync{f} = \tilde{w_1} (\tilde{r}_1) \tilde{w}_2 (\tilde{r}_2) \dots \tilde{w}_m (\tilde{r}_m) \tilde{w}_{m+1} h w$ with $2 \leq m \leq n$ and nonempty $\tilde{r}_i$, which is a valid formula tape which also solves the bouncer. By construction, the machine will reach tapes of the form $C_{\tilde{f}}(n)$ for all $n\in\N$ at increasing time steps. We immediately get Point 2 because $C_{\tilde{f}}(n+1)$ contains $K = \sum_{i=1}^n |\tilde{r}_i|$ more symbols than $C_{\tilde{f}}(n)$, which is a constant.

    Concerning Point 3, call $g(n)$ the time step at which $C_{\tilde{f}}(n)$ is reached. By construction, $g(n) \geq n$. The number of macro steps $M \geq 1$ such that we get $\sync{f} \vdash_\mathcal{A}^M f''$ with $f''$ special case of $\sync{f}$ is a constant. We write $M = M_u + M_s$ with $ M_u$ the number of usual steps and $M_s$ the number of shift rule applications. Call $M'_s$ the number of Turing machine steps taken in total to perform the base case (i.e.\ the steps needed to prove that a shift rule is correct) of each of the $M_s$ shift rules. When processing $C_{\tilde{f}}(n)$, each shift rule is applied once each to the $n$ repetitions of each repeater (because shift rules preserve the number of repetitions), hence we get: $g(n+1) - g(n) = M_u + n M'_s$, which means $D(D((g(n))_{n\in\N})) = (M'_s)_{n\in\N}$ which is constant, as needed.

\end{proof}


\begin{example}\label{ex:linquad}
    Using Example~\ref{ex:bouncerTheory}, we get that $f = 0^\infty (111)1110(11)00\rhead{\text{D}}0^\infty$ solves the bouncer given in Example~\ref{ex:bouncer88427177}. Using Theorem~\ref{th:linquad}, we have $\sync{f}= 0^\infty (111)1111110(11)11 00\rhead{\text{D}}0^\infty$.
    We have $\sync{f} \vdash_\mathcal{A}^{M} f'$ with $M = 47$ and $f'$ special case of $\sync{f}$. Note that $47 > 41$ found in Example~\ref{ex:bouncerTheory} because $\sync{f}$ is a bit bigger than $f$.

    We can decompose $M = M_u + M_s$ with $M_u = 45$ usual steps and $M_s = 2$ shift rule applications. The shift rules that are used are $0 \rhead{\text{D}}(01) \to (11)0\rhead{\text{D}}$ which takes $4$ Turing machine steps to execute (Example~\ref{ex:shiftRules}) and $(11) \lhead{\text{A}}\; \to \;  \lhead{\text{A}} (01)$ which takes $2$, hence $M'_s = 6$.

    Calling $g(n)$ the time step at which $C_{\sync{f}}(n)$ is reached, we have $g(0)= 64$. Using the proof of Theorem~\ref{th:linquad}, $g(n+1)-g(n) = M_u + n M'_s = 45 + 6n$, hence, $g(n) = 3n^2 + 42n + 64$. We also have $l_{g(n)} = 11 + 5n$. One can check by simulation that tapes $C_{\sync{f}}(n)$ are reached at time steps $g(n)$.

\end{example}
\newpage
\subsection{Deciding bouncers in practice}\label{sec:bouncers-algo}

In this section, we use the theory presented above in order to give a practical and efficient algorithm for deciding bouncers in practice.

\subsubsection{Formula tape fitting}

We introduce $\mathcal{A}_r$, the \textit{right-alignment} operator on formula tapes which only applies rule of type 1 in Definition~\ref{def:alignment} to all repeaters, both before and after the formula's head, bubbling them to the right.

In this section, we give a greedy algorithm (Algorithm~\ref{alg:greedy-formula-tape-fitting}) that fits a formula tape $f$ only given three tapes: $C_f(0)$, $C_f(1)$ and $C_f(2)$. The algorithm is guaranteed to work on right-aligned formula tapes with nonempty \textit{intermediary walls}, which are all walls but the first and last one, see Theorem~\ref{th:greedy-formula-tape-fitting}. First, we show that the precondition of having nonempty intermediary walls is without loss of generality, i.e.\ that all bouncers can be solved using such formulas. We prove this result in Lemma~\ref{lem:formula-tapes-WLOG}, using two technical lemmas, Lemmas~\ref{lem:right-alignement}~and~\ref{lem:word-theory}.

\begin{lemma}\label{lem:right-alignement}
    Assume that $M$ is a bouncer solved by some formula tape $f$, then:
    \begin{enumerate}
        \item $\mathcal{A}_r(f)$ also solves the bouncer.
        \item Any special case $f'$ of $f$ also solves the bouncer.

    \end{enumerate}
\end{lemma}
\begin{proof}

    \begin{enumerate}
        \item We have $\mathcal{A}(\mathcal{A}_r(f)) = \mathcal{A}(f)$ hence, Theorem~\ref{th:bouncers} will reach $\mathcal{A}(f)$ after one application of $\mathcal{A}$ in any case and the bouncer will get solved in the same way.
        \item By definition of special case $f'$ can be constructed from $f$ by replacing any $(r)$ by $r^n(r)r^m$ for some $n,m\in\N$. Aligning $f'$ will have the same effect as aligning $f$ with the addition that repeaters before the head will look like $(r)r^{n+m}$  and after the head $r^{n+m}(r)$. Then, shift rules will apply similarly to $f'$ and $f$, with the addition in $f'$ that segments of the form $r^{n+m}$ will be handled by simulating the same shift rule as for $(r)$ through usual steps. Hence, $f'$ solves the bouncer.
    \end{enumerate}

\end{proof}

\begin{lemma}\label{lem:word-theory}
    Let $a,b \in \Sigma^+$ be two nonempty words such that we have $a[i \text{ mod } |a|] = b^\infty [i]$ for all $i\in\N$, with $b^\infty [i]$ being the character at position $i$ in the infinite concatenation of word $b$ with itself. Call $k=\operatorname{gcd}(|a|,|b|)$. Then, $a$ and $b$ are powers of the same word $c$ such that $a = c^m$ and $b = c^n$, with $c$ the first $k$ characters of $a$ and $b$, and $m = |a|/k$ and $n = |b|/k$.
\end{lemma}
\begin{proof}
    Note that $|a| \neq 0$ and $|b| \neq 0$ by hypothesis. Call $f(i) = b^\infty[i]$. It is immediate that $|b|$ is a period of $f$. By hypothesis, we also have that $|a|$ is a period of $f$. Hence, any linear combination that has positive value of $|a|$ and $|b|$ is also a period of $f$. Hence, by Euclid's algorithm, $k=\operatorname{gcd}(|a|,|b|)$ is a period of $f$, from which we get $b = c^{|b|/k}$ with $c$ the first $k$ characters of $b$ since $k \leq |b|$. By hypothesis, we know that $a^{\infty} = b^{\infty}$, hence, we have $a = c^{|a|/k}$, as needed.

\end{proof}

\begin{lemma}\label{lem:formula-tapes-WLOG}
    Let a machine $M$ be a bouncer solved by formula tape $f$, then $f$ can be transformed into a formula tape $f'$ such that $\mathcal{A}_r(f')$ has nonempty intermediary walls that solves the bouncer.
\end{lemma}
\begin{proof}
    By Lemma~\ref{lem:right-alignement}, $\mathcal{A}_r(f)$ also solves the bouncer. For each case $\dots (r_1)(r_2) \dots$ where two repeaters in $\mathcal{A}_r(f)$ are separated by an empty wall, apply the following transformation to $\mathcal{A}_r(f)$, according to two cases:

    \begin{enumerate}
        \item There is $k\geq 1$ such that right-aligning $(r_1)r_2^k(r_2)$ yields a non-empty intermediary wall, replace $(r_1)(r_2)$ by $\mathcal{A}_r((r_1)r_2^k(r_2))$, yielding a new formula tape $f'$. Noticing that $(r_1)r_2^k(r_2)$ is a special case of $(r_1)(r_2)$ and combining both points of Lemma~\ref{lem:right-alignement} we get that $f'$ still solves the bouncer.

        \item For all $k\geq 1$ there is $r'_k\in\Sigma^+$ such that we have $\mathcal{A}_r((r_1)r_2^k(r_2)) = r_2^k (r'_k)(r_2)$. We deduce that $r_1[i \text{ mod } |r_1|] = r_2^\infty[i]$ for all $i\in\N$. By Lemma~\ref{lem:word-theory} we get that there is $c$ such that $r_1 = c^m$ and $r_2 = c^n$. Hence, we can replace $(r_1)(r_2) = (c^m)(c^n)$ by $(c^m c^n) = (r_1 r_2)$ in the formula tape and still solve the bouncer with this new formula tape. Indeed, because $f$ solves the bouncer we have $u \in\Sigma^*$, $s\in S$, $r', r'' \in \Sigma^+$ and two shift rules, for instance, right shift rules: $u \rhead{s} (c^m) \to (r') u \rhead{s}$ and $u \rhead{s} (c^n) \to (r'') u \rhead{s}$. Hence, considering $u  \rhead{s} c^m c^n$, we get $u  \rhead{s} c^m c^n \vdash^* r' u  \rhead{s} c^n \vdash^+ r' r'' u  \rhead{s}$, meaning that we have a right shift rule $u \rhead{s} (c^m c^n) \to (r'r'') u \rhead{s}$, as needed.

    \end{enumerate}

    By applying case (1) or (2) to all repeaters of $\mathcal{A}_r(f)$ that are separated by an empty wall, we get $f'$, a transformation of $f$ with nonempty intermediary walls that solves the bouncer. Because $f'$ is right-aligned by construction, $\mathcal{A}_r(f') = f'$  and we get the result.

\end{proof}

We are now ready to infer such a formula tape $f$ from three examples $t_i=C_f(i)$.
To simplify the procedure, we wish to drop the head and $0^\infty$ symbols from the tapes and re-attach them to the matching formula.

Given a tape $t$, we call $\mathcal{S}(t) \in \Sigma^*$ the \textit{headless} version of $t$ which is the tape without head and $0^\infty$ symbols. We also introduce headless formula tapes: $\mathcal{S}(f)$ is $f$ without head and $0^\infty$ symbols.

Algorithm~\ref{alg:greedy-formula-tape-fitting} is a greedy algorithm which fits a headless formula tape given three headless tapes $t_0, t_1, t_2$. It proceeds by recursively constructing the leftmost wall using the longest common prefix of $t_0, t_1$ and $t_2$ and then fits the leftmost repeater using the longest prefix $r$ of $t_1$ such that $rr$ is a prefix of $t_2$, we have:

\begin{algorithm}
    \caption{Greedy formula tape fitting algorithm {\sc FitFormulaTape}}\label{alg:greedy-formula-tape-fitting}
    \begin{algorithmic}[1]

        \Procedure{\textbf{HeadlessFormulaTape} {\sc FitFormulaTape}}{\textbf{Word} t0, \textbf{Word} t1, \textbf{Word} t2}

        \If{t0.\textbf{empty}() \&\& t1.\textbf{empty}() \&\& t2.\textbf{empty}()}
        \State \Return \textbf{HeadlessFormulaTape}::\textbf{empty}()

        \EndIf
        \State
        \If{t0.\textbf{len}() $>$ 0 \&\& t1.\textbf{len}() $>$ 0 \&\& t2.\textbf{len}() $>$ 0 \&\& t0[0] == t1[0] \&\& t1[0] == t2[0]}
        \State \Return \textbf{HeadlessFormulaTape}::\textbf{symbol}(t0[0]).\textbf{concat}(\Call{{\sc FitFormulaTape}}{t0[1:], t1[1:], t2[1:]})
        \EndIf
        \State
        \State
        \textbf{uint} longest\_prefix\_size = \textbf{get\_longest\_prefix\_size}(t1, t2)
        \State
        \For{l = longest\_prefix\_size; k $\geq$ 1; k -= 1 }

        \If{2*l $<$ t2.\textbf{len}() \&\& t2[:l] == t2[l:2*l] }
        \State \Return \textbf{HeadlessFormulaTape}::\textbf{repeater}(t2[:l]).\textbf{concat}(\Call{{\sc FitFormulaTape}}{t0, t1[l:], t2[2*l:]})
        \EndIf

        \EndFor
        \State
        \State \textbf{raise} \textbf{Failure}
        \EndProcedure
    \end{algorithmic}
\end{algorithm}

\begin{remark}[Implementation details of Algorithm~\ref{alg:greedy-formula-tape-fitting}]
    We assume that we are given a \textbf{HeadlessFormulaTape} construct, and a function \textbf{get\_longest\_match\_size} which returns the size of the longest prefix of two words. Furthermore, for an array \texttt{a}, we use the notation \texttt{a[b:e]} to mean the slice of this array between indices \texttt{b} and \texttt{e}, excluding \texttt{e}.
\end{remark}



\begin{theorem}[Greedy formula tape fitting]\label{th:greedy-formula-tape-fitting}
    Let $f$ be a headless formula tape with nonempty intermediary walls such that $\mathcal{A}_r(f) = w_1 (r_1) w_2 (r_2)\dots w_m (r_m) w$ with $m\in\N$, $r_i, w_i \in \Sigma^+$ with $2 \leq i \leq m$, and $w_1, w \in \Sigma^*$. Take $t_0, t_1, t_2 \in \Sigma^+$ satisfying:
    \begin{align*}
        t_0 & = \mathcal{S}(C_f(0))  = w_1\; w_2 \dots w_m\; w                       \\
        t_1 & = \mathcal{S}(C_f(1))  = w_1\; r_1\; w_2\; \dots w_m\; r_m\; w         \\
        t_2 & = \mathcal{S}(C_f(2))  = w_1\; r_1 r_1\; w_2\; \dots w_m\; r_m r_m\; w
    \end{align*}
    Then Algorithm~\ref{alg:greedy-formula-tape-fitting}, {\sc FitFormulaTape}($t_0$, $t_1$, $t_2$) does not raise failure and returns $\mathcal{A}_r(f)$.

\end{theorem}
\begin{proof}
    Notation: if a word $u\in\Sigma^*$ is not empty we use the notation $u[0]$ to mean the first symbol of $u$.
    Let's first prove the case $m=1$ and write $r = r_1$. Using the case of Algorithm~\ref{alg:greedy-formula-tape-fitting}, line 10, we get:
    \setlength{\columnsep}{-7.5cm}
    \begin{multicols}{3}
        \noindent
        \begin{align*}
            t_0          & = w_1\;  w      \\
            t_1          & = w_1\; r\;  w  \\
            t_2          & = w_1\; r r\; w \\
            f_\text{out} & = \varnothing
        \end{align*}
        \begin{align*}
             & \\
            \to_{\text{Algorithm~\ref{alg:greedy-formula-tape-fitting} goes to}}
        \end{align*}
        \begin{align*}
            t_0          & = w        \\
            t_1          & = r\;  w   \\
            t_2          & = r r\;  w \\
            f_\text{out} & = w_1
        \end{align*}
    \end{multicols}

    Indeed, this is because $w_1$ is the biggest prefix shared by $t_0, t_1$ and $t_2$; otherwise, it means that $w_1w[0]$ is a prefix of $t_1$ and we get $r[0] = w[0]$ which is excluded by right-alignment.

    From there, we have two cases: (1) if $w = \varnothing$ then the greedy algorithm returns $f_\text{out} = w_1 (r) = \mathcal{A}_r(f)$, as needed, (2) if $w \neq \varnothing$, because $\mathcal{A}_r(f)$ is right-aligned, we must have $w[0] \neq r[0]$ ($r$ is not empty by hypothesis), and $r$ is the longest prefix common to $t_1$ and $t_2$ since otherwise, using $t_2$ we get $w[0] = r[0]$. Hence, the greedy algorithm returns $f_\text{out} = w_1 (r) w = \mathcal{A}_r(f)$, as needed.

    Assuming $m\geq 2$, as above because of right-alignment, we get that $w_1$ is the longest prefix common to $t_0$, $t_1$ and $t_2$, hence:
    \setlength{\columnsep}{-1.9cm}
    \begin{multicols}{3}
        \noindent
        \begin{align*}
            t_0          & = w_1\; w_2\;\dots w_m\; w                              \\
            t_1          & = w_1\; r_1\; w_2\; r_2\; \dots w_m\; r_m\; w           \\
            t_2          & = w_1\; r_1r_1\; w_2\;  r_2r_2\; \dots w_m\; r_mr_m\; w \\
            f_\text{out} & = \varnothing
        \end{align*}
        \begin{align*}
             &                                                                   \\
            \to_{\text{Algorithm~\ref{alg:greedy-formula-tape-fitting} goes to}} \\
             &
        \end{align*}
        \begin{align*}
            t_0          & = w_2\;\dots w_m\; w                              \\
            t_1          & = r_1\; w_2\; r_2\; \dots w_m\; r_m\; w           \\
            t_2          & = r_1r_1\; w_2\;  r_2r_2\; \dots w_m\; r_mr_m\; w \\
            f_\text{out} & = w_1
        \end{align*}
    \end{multicols}
    By hypothesis, $w_2$ and $r_1$ are not empty and, because $\mathcal{A}_r(f)$ is right-aligned, we get $w_2[0] \neq r_1[0]$.
    Hence, Algorithm~\ref{alg:greedy-formula-tape-fitting} enters the repeater-fitting case on line 8. Necessarily, $r_1$ is the longest prefix to $t_1$ and $t_2$ otherwise using $t_2$ we get $r_1[0] = w_2[0]$ which contradicts right-alignment, hence:
    \begin{align*}
        t_0          & = w_2\;\dots w_m\; w                     \\
        t_1          & = w_2\; r_2\; \dots w_m\; r_m\; w        \\
        t_2          & = w_2\;  r_2r_2\; \dots w_m\; r_mr_m\; w \\
        f_\text{out} & = w_1 (r_1)
    \end{align*}
    From here we can inductively conclude that Algorithm~\ref{alg:greedy-formula-tape-fitting} will reach:
    \begin{align*}
        t_0          & = w_m\; w                           \\
        t_1          & = w_m\; r_m\; w                     \\
        t_2          & = w_m\; r_mr_m\; w                  \\
        f_\text{out} & = w_1 (r_1) \dots w_{m-1} (r_{m-1})
    \end{align*}

    Using the same argument as for $m=1$, we get that Algorithm~\ref{alg:greedy-formula-tape-fitting} returns $f_\text{out} = w_1 (r_1) \dots w_m (r_m) w = \mathcal{A}_r(f)$, as needed.

\end{proof}

\begin{example}
    Consider the headless right-aligned formula tape with nonempty intermediary walls $f=100(100)0000(0)$. Then the Algorithm~\ref{alg:greedy-formula-tape-fitting} proceeds as follows:

    \begin{multicols}{5}
        \noindent
        \begin{align*}
            t_0          & = 100\; 0000                      \\
            t_1          & = 100\; 100\; 0000 \; 0           \\
            t_2          & = 100\; 100\; 100\; 0000 \; 0\; 0 \\
            f_\text{out} & = \varnothing
        \end{align*}
        \begin{align*}
             &                                                                   \\
            \to_{\text{Algorithm~\ref{alg:greedy-formula-tape-fitting} goes to}} \\
             &
        \end{align*}
        \begin{align*}
            t_0          & = 0000                      \\
            t_1          & = 100\; 0000 \; 0           \\
            t_2          & = 100\; 100\; 0000 \; 0\; 0 \\
            f_\text{out} & = 100
        \end{align*}
        \begin{align*}
             &                                                                   \\
            \to_{\text{Algorithm~\ref{alg:greedy-formula-tape-fitting} goes to}} \\
             &
        \end{align*}
        \begin{align*}
            t_0          & = 0000           \\
            t_1          & =  0000 \; 0     \\
            t_2          & =  0000 \; 0\; 0 \\
            f_\text{out} & = 100 (100)
        \end{align*}
    \end{multicols}
    \begin{multicols}{5}
        \noindent
        \begin{align*}
            t_0 & = 0000           \\
            t_1 & =  0000 \; 0     \\
            t_2 & =  0000 \; 0\; 0 \\
            f   & = 100 (100)
        \end{align*}
        \begin{align*}
             &                                                                   \\
            \to_{\text{Algorithm~\ref{alg:greedy-formula-tape-fitting} goes to}} \\
             &
        \end{align*}
        \begin{align*}
            t_0          & = \varnothing    \\
            t_1          & =   0            \\
            t_2          & =   0\; 0        \\
            f_\text{out} & = 100 (100) 0000
        \end{align*}
        \begin{align*}
             &                                                                   \\
            \to_{\text{Algorithm~\ref{alg:greedy-formula-tape-fitting} goes to}} \\
             &
        \end{align*}
        \begin{align*}
            t_0          & = \varnothing        \\
            t_1          & =  \varnothing       \\
            t_2          & =  \varnothing       \\
            f_\text{out} & = 100 (100) 0000 (0)
        \end{align*}
    \end{multicols}

    Hence, the output is $f$, as claimed in Theorem~\ref{th:greedy-formula-tape-fitting}.

    Note that in some case, Algorithm~\ref{alg:greedy-formula-tape-fitting} is able to fit formula tapes that have some empty intermediary walls, such as $f=10110(110)(101)$. But in other cases, it cannot, such as with $f = 01 (1) (0111) 1$, which raises failure.

\end{example}

\subsubsection{Bouncers decider}

We finally piece all the elements of this chapter together and describe a decider for bouncers (excluding cyclers, decided in Section~\ref{sec:cyclers}), Algorithm~\ref{alg:decider-bouncers} which is proven correct in Theorem~\ref{th:bouncer-decider}.

In order to drastically limit the amount of plausible subsequences to check, we limit our interest to \textit{record-breaking} formula tapes, which are formula tapes where the head is pointing at $0^\infty$. In order to fit them, we can simply track record-breaking tapes that share the same head state and direction, i.e.\ tapes where the head is pointing at $0^\infty$. We first show that this is without loss of generality:

\begin{lemma}\label{lem:record-breaking}
    If $M$ is a bouncer that is not a cycler, solved by a formula tape $f$, then there is a record-breaking formula tape that also solves it.
\end{lemma}
\begin{proof}
    We have $f \vdash_\mathcal{A} f_1 \vdash_\mathcal{A} f_2 \dots \vdash_\mathcal{A} f_n$ with $n \geq 1$ and $f_n$ is a special case of $f$. There is $1 \leq i \leq n$ such that $f_{i}$ is record-breaking otherwise, the formula tapes do not grow and $M$ is a cycler, which is excluded. Because $f_n$ is a special case of $f$, it will eventually reach, under applications of $\vdash_\mathcal{A}$, a formula tape $f'$ that is special case of $f_i$ and we have the result.
\end{proof}

\begin{lemma}\label{lem:linquad-intermediary-walls}
    If $f$ is a formula tape with nonempty intermediary walls, then $\sync{f}$ built in Theorem~\ref{th:linquad} has nonempty intermediary walls.
\end{lemma}
\begin{proof}
    The construction of $\sync{f}$ only involves (1) removing some repeaters (2) replacing some repeaters by some powers of themselves (3) increasing the size of some walls. Hence, $\sync{f}$ has nonempty intermediary walls.
\end{proof}

Algorithm~\ref{alg:decider-bouncers} proceeds by (1) tracking record-breaking tapes with same head (i.e.\ same state and pointing-direction) that grow linearly in quadratic time, (2) fitting formula tapes from triples from these plausible subsequences until a formula tape solves the bouncer is found or some limits are met.

\begin{algorithm}[h!]
    \caption{{\sc Decider-Bouncers}}\label{alg:decider-bouncers}
    \begin{algorithmic}[1]

        \Procedure{\textbf{bool} {\sc Decider-Bouncers}}{\textbf{TM} machine, \textbf{uint} step\_limit, \textbf{uint} macro\_step\_limit, \textbf{uint} max\_formula\_tapes}

        \State \textbf{Map[TMHead,Vec[Tape]]} record\_breaking\_tapes = \textbf{get\_record\_breaking\_tapes}(machine, step\_limit)

        \For{\textbf{TMHead} head \textbf{in} record\_breaking\_tapes}
        \State \textbf{uint} num\_tested\_formula = 0
        \For{\textbf{uint} i, \textbf{Tape} tape4 \textbf{in} record\_breaking\_tapes[head].\textbf{enumerate}()}
        \State \textbf{if} i $<$ 3 \textbf{then continue}
        \State \textbf{bool} break\_outer = \textbf{false}
        \For{\textbf{uint} j, \textbf{Tape} tape3 \textbf{in} record\_breaking\_tapes[head][:i].\textbf{enumerate}()}
        \State \textbf{if} j $<$ 2 \textbf{then continue}
        \State \textbf{uint} len\_diff = tape4.\textbf{len}() - tape3.\textbf{len}()
        \State \textbf{uint} tape2\_len = tape3.\textbf{len}() - len\_diff
        \State \textbf{Tape or None} tape\_2 = record\_breaking\_tapes[head][:i].\textbf{binary\_search\_len}(tape2\_len)
        \If{tape\_2 \textbf{is} None}
        \State \textbf{continue}
        \EndIf
        \State
        \State \textbf{uint} tape1\_len = tape2.\textbf{len}() - len\_diff
        \State \textbf{Tape or None} tape\_1 = record\_breaking\_tapes[head][:i].\textbf{binary\_search\_len}(tape1\_len)
        \If{tape\_1 \textbf{is} None}
        \State \textbf{continue}
        \EndIf
        \State
        \If{\textbf{not} \textbf{is\_quadratic}([tape\_1.step,tape\_2.step,tape\_3.step,tape\_4.step])}
        \State \textbf{continue}
        \EndIf

        \State

        \State \textbf{try} \textbf{FormulaTape} f = {\sc FitFormulaTape}(tape\_1, tape\_2, tape\_3).\textbf{headless}()
        \If{\textbf{Failure} was raised}
        \State \textbf{continue}
        \EndIf
        \State
        \State f.\textbf{attach\_head}(tape\_1.head)
        \If{f.\textbf{reaches\_special\_case}(macro\_step\_limit)}

        \State \Return \textbf{true}

        \EndIf

        \State
        \State num\_tested\_formula += 1

        \State
        \If{num\_tested\_formula == max\_formula\_tapes}
        \State break\_outer = \textbf{true}
        \State \textbf{break}
        \EndIf

        \EndFor
        \State
        \If{break\_outer}
        \State \textbf{break}
        \EndIf
        \EndFor
        \EndFor
        \State \Return \textbf{false}
        \EndProcedure
    \end{algorithmic}
\end{algorithm}

\begin{theorem}[Deciding bouncers]\label{th:bouncer-decider}
    Let $M$ be a bouncer (Definition~\ref{def:bouncers}). Then, there exists a step limit $s\in\N$, a macro step limit $m\in\N$ and a formula tape testing limit $l$ such that Algorithm~\ref{alg:decider-bouncers}, {\sc Decider-Bouncers}($M$,$s$,$m$,$l$) outputs \textbf{true}.
\end{theorem}
\begin{proof}
    Because $M$ is a bouncer, using Lemma~\ref{lem:record-breaking} we know that there is a record-breaking formula tape $\tilde{f}$ that solves the bouncer. By Lemma~\ref{lem:formula-tapes-WLOG} we know that we can transform $\tilde{f}$ into $f$ with nonempty intermediary walls and that also solves the bouncer.

    Algorithm~\ref{alg:decider-bouncers} enumerates all record-breaking tapes triple that grow linearly in quadratic time. If $s$ and $l$ are large enough, by Theorem~\ref{th:linquad}, we know that it will, eventually meet a triple of the form: \begin{align*}
        t_0 & = C_{\sync{f}}(0) \\
        t_1 & = C_{\sync{f}}(1) \\
        t_2 & = C_{\sync{f}}(2)
    \end{align*}

    By Lemma~\ref{lem:linquad-intermediary-walls}, $f'=\sync{f}$ is with nonempty intermediary walls and, at line 23, Algorithm~\ref{alg:decider-bouncers} will feed $(\mathcal{S}(C_{f'}(0)),\mathcal{S}(C_{f'}(1)),\mathcal{S}(C_{f'}(2)))$ to Algorithm~\ref{alg:greedy-formula-tape-fitting}, which, by Theorem~\ref{th:greedy-formula-tape-fitting}, will output $\mathcal{S}(\mathcal{A}_r(\sync{f}))$ which can be reconstructed into $\mathcal{A}_r(\sync{f})$ at line 27. By Theorem~\ref{th:linquad}, we know that $\sync{f}$ also solves the bouncer, and by Lemma~\ref{lem:right-alignement}, $\mathcal{A}_r(\sync{f})$ also solves the bouncer. Theorem~\ref{th:bouncers} will be verified on $\mathcal{A}_r(\sync{f})$ using \textbf{reaches\_special\_case} at line 28, which, by hypothesis will output \textbf{true}, and we have the result.
\end{proof}

\begin{remark}
    While Theorem~\ref{th:bouncer-decider} assures us that the decider will be able to detect a bouncer in bounded time, the formula tapes found in practice by the algorithm can be smaller than the $\sync{f}$ construction we used to prove the upper bound.
\end{remark}

\begin{remark}[Implementation details of Algorithm~\ref{alg:decider-bouncers}]
    In Algorithm~\ref{alg:decider-bouncers}, we assume that we are given (1) a \textbf{get\_record breaking\_tapes} routine which returns the record-breaking tapes for each tape head $h\in\Delta$ (i.e.\ tape head state and pointing-direction), in increasing length (2) a \textbf{binary\_search\_len} routine which finds by binary search a tape of a given length, or returns \textbf{None} if it does not exist and (3) a \textbf{is\_quadratic} which tests that a sequence of integers is in quadratic progression (for instance by computing second differences and testing they are constant), (4) routines \textbf{headless} and \textbf{attach\_head} to manipulate tapes/formula tapes with and without head, and (5) a \textbf{reaches\_special\_case} routine on formula tapes which returns \textbf{true} if successive simulation and alignment of the formula tape, as described in Theorem~\ref{th:bouncers}, reaches a special case in a given amount of macro steps, where a macro step consists in performing alignment followed by either performing a usual step or a shift rule step. We can note that when detecting shift rules (see Section~\ref{sec:bouncers:formula-tapes}), it is important to implement cycler detection since it is possible for a machine to cycle indefinitely on a finite tape.
\end{remark}

\subsection{Implementations}\label{sec:bouncers-implem}

Here are the implementations of the decider that were realised:

\begin{enumerate}
    \item Tony Guilfoyle's C++ initial implementation (does not use the theory presented in this section): \url{https://github.com/TonyGuil/bbchallenge/tree/main/Bouncers}
    \item Iijil's Go implementation (does not use the theory presented in this section): \url{https://github.com/Iijil1/Bouncers}
    \item savask's Haskell implementation (basis of the theory presented in this section): \url{https://gist.github.com/savask/888aa5e058559c972413790c29d7ad72}
    \item mei's optimised Rust implementation, reproducing savask's: \url{https://github.com/meithecatte/busycoq/}. This implementation outputs certificates that are verified using Coq, more details about this approach will be given in future versions of this text.
    \item Tristan Stérin's (cosmo's) Rust implementation, reproducing savask's and mei's: \url{https://github.com/bbchallenge/bbchallenge-deciders/tree/main/decider-bouncers-reproduction}. This implementation follows this text to the letter, reproducing each concept and algorithm as presented here. It is less efficient than mei's.

\end{enumerate}

\textbf{Verifiers.} Verifiers for Theorem~\ref{th:bouncers}, i.e.\ programs that verify bouncer certificates (Definition~\ref{def:bouncer-certificate}) have also been given as part of the above implementations. Mei's implementation provides a Coq implementation of the verifier. There is an ongoing effort for standardising the format of bouncers certificates (and certificates in general).


\begin{figure}[h!]
    \centering
    \includegraphics[scale=0.35]{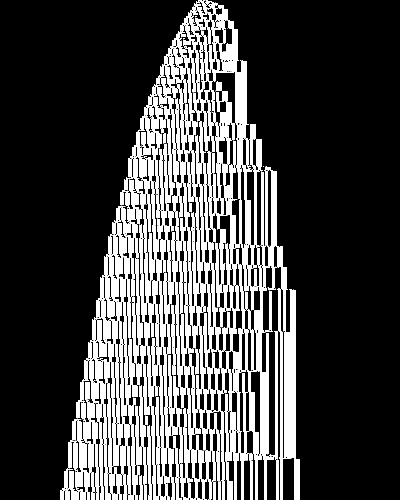}
    \caption{\small 50,000-step space-time diagram of \url{https://bbchallenge.org/5608043}. Out of the 29,799 decided bouncers, this bouncer takes the most steps (141,509) to be detected and fits the biggest formula tape, using Algorithm~\ref{alg:decider-bouncers} presented in this section.}\label{fig:big-bounce}
\end{figure}
\vspace{-2.5ex}
\paragraph*{Remarkable bouncers.} Out of the 29,799 bouncers that were decided using Algorithm~\ref{alg:decider-bouncers}, here are some remarkable facts:
\begin{enumerate}
    \item Using Algorithm~\ref{alg:decider-bouncers}, only two bouncers with 3 repeaters were found (and that's the maximum): \url{https://bbchallenge.org/347505} and \url{https://bbchallenge.org/8131743}. Otherwise, 2132 bouncers have 2 repeaters and the rest have only 1.
    \item\label{pt:big-formula-tape} The biggest fitted formula tapes by the algorithm have 328 symbols (summing walls and repeater symbols, not counting head and $0^\infty$), there are two of them, such as for machine \url{https://bbchallenge.org/5608043}, see Figure~\ref{fig:big-bounce}:
          \begin{align*}
               & 0^\infty\lhead{\text{A}}  10000011100011100000011100001110111001110111001110000111 \\ &0000111011100111000011101110011101110011100001110000111011100\\ &1110000111000011100001110000111011100111011100111011100111011\\ &100111011100111011100(111011100111011100111011100111011100)000\\ &1110000111000000111001111111111111110000111(111111111111)00111\\ &000000011100000011111100111111(0)^\infty
          \end{align*}

    \item The above machine of Point~\ref{pt:big-formula-tape} is also the machine that is detected after the most steps: 141,509. Over this dataset, it took 207 steps on average.
    \item The most macro steps (i.e.\ number of usual or shift rule steps in formula tape simulation) needed to conclude using Theorem~\ref{th:bouncers} was 41,628 for \url{https://bbchallenge.org/347505}. Otherwise, it took 66 macro steps on average.

\end{enumerate}

\appendix

\section{Author contributions}

\paragraph{The bbchallenge Collaboration (credits).} The following contributions resulted in the present work: Iijil (Halting Segment); Mateusz Na\'{s}ciszewski, Nathan Fenner, Tony Guilfoyle (Halting Segment reproductions); Justin Blanchard, Mateusz Naściszewski, Konrad Deka (FAR); Tony Guilfoyle, Tristan Stérin (FAR reproductions); Tony Guilfoyle (Bouncers); Iijil, savask, Maja Kądziołka, Tristan Stérin (Bouncers reproductions); savask, Tristan Stérin (Bouncers theory); Tristan Stérin, Justin Blanchard, savask (paper writing); Pavel Kropitz, Shawn Ligocki, Pascal Michel (paper review).

\begin{itemize}
  \item The bbchallenge Collaboration, \url{bbchallenge.org}, \texttt{bbchallenge@bbchallenge.org}
  \item Justin Blanchard, \texttt{UncombedCoconut@gmail.com}
  \item Konrad Deka, \texttt{deka.konrad@gmail.com}
  \item Nathan Fenner, \texttt{nfenneremail@gmail.com}
  \item Tony Guilfoyle, \texttt{tonyguil@gmail.com}
  \item Iijil, \texttt{hheussen@web.de}
  \item Maja Kądziołka, \texttt{bb@compilercrim.es}
  \item Pavel Kropitz
  \item Shawn Ligocki, \texttt{sligocki@gmail.com}
  \item Pascal Michel, \texttt{pascalmichel314@gmail.com}
  \item Mateusz Na\'{s}ciszewski
  \item Tristan Stérin, PRGM DEV, \texttt{tristan@prgm.dev}
\end{itemize}

\bibliographystyle{abbrv}

\bibliography{correctness-deciders}

\end{document}